\documentclass[prodmode,acmsmall,acmjacm,natbib=false]{acmart}

\RequirePackage[datamodel=acmdatamodel,backend=biber]{biblatex}
\usepackage{csquotes}

\usepackage{enumitem}

\usepackage{amsmath}
\usepackage{stmaryrd}
\usepackage{verbatim,bbm}
\usepackage{mathrsfs}
\usepackage{amsthm}
\usepackage{pdfsync}
\usepackage{xspace}
\usepackage{hyperref}
\hypersetup{bookmarks=true,bookmarksopen=true}
\usepackage{microtype}
\usepackage{tikz}
\usepackage[mathscr]{euscript}
\let\mathcal\mathscr

\usepackage{nowidow}
\usepackage{widows-and-orphans}

\usepackage[english]{babel}
\usetikzlibrary{decorations.text,arrows,decorations.pathreplacing,shapes,fit,shadows,calc,patterns,intersections}
\definecolor{my1}{cmyk}{0,.6,0,0}
\definecolor{my2}{cmyk}{.3,.0,.0,.0}

\usepackage[scaled=1.1]{dsserif}
\let\mathbbm\mathbb

\theoremstyle{plain}
\newtheorem{theorem}{Theorem}[section]
\newtheorem{thm}[theorem]{Theorem}

\newtheorem{cor}[theorem]{Corollary}

\newtheorem{prop}[theorem]{Proposition}

\newtheorem{lem}[theorem]{Lemma}
\newtheorem{fct}[theorem]{Fact}
\newtheorem{fact}[theorem]{Fact}
\newtheorem{remark}[theorem]{Remark}
\newtheorem{rem}[theorem]{Remark}
\newtheorem{exa}[theorem]{Example}

\newcommand{\dclosp}[1]{\ensuremath{\mathord{\downarrow_{#1}}}\xspace}
\newcommand{\dclosr}{\dclosp{R}}
\newcommand{\dclosq}{\dclosp{Q}}

\newcommand{\etac}{\ensuremath{\eta_\Cs}\xspace}
\newcommand{\canc}{\ensuremath{N_\Cs}\xspace}

\let\Max\max

\renewcommand{\max}{\ensuremath{\text{\scriptsize max}}\xspace}

\newcommand{\fo}{\ensuremath{\textup{FO}}\xspace}
\newcommand{\fow}{\mbox{\ensuremath{\fo(<)}}\xspace}

\newcommand{\fowm}{\mbox{\ensuremath{\fo(<,MOD)}}\xspace}
\newcommand{\fowam}{\mbox{\ensuremath{\fo(<,AMOD)}}\xspace}

\newcommand{\foc}{\ensuremath{\fo(\infsigc)}\xspace}

\newcommand{\ltl}{\ensuremath{\textup{LTL}}\xspace}

\newcommand{\ltlc}[1]{\ensuremath{\textup{LTL}(#1)}\xspace}
\newcommand{\ltlpc}[1]{\ensuremath{\textup{LTLP}(#1)}\xspace}
\newcommand{\ltla}[1]{\ensuremath{\textup{LTL}[#1]}\xspace}
\newcommand{\ltlpa}[1]{\ensuremath{\textup{LTLP}[#1]}\xspace}

\newcommand{\at}{\ensuremath{\textup{AT}}\xspace}
\newcommand{\md}{\ensuremath{\textup{MOD}}\xspace}
\newcommand{\abg}{\ensuremath{\textup{AMT}}\xspace}
\newcommand{\grp}{\ensuremath{\textup{GR}}\xspace}

\newcommand{\bsd}{\ensuremath{\mathit{SD}}\xspace}
\newcommand{\sfr}{\ensuremath{\mathit{SF}}\xspace}

\newcommand{\stzer}{\textup{ST}\xspace}

\newcommand{\sfp}[1]{\ensuremath{\mathit{SF}(#1)}\xspace}
\newcommand{\bsdp}[1]{\ensuremath{\mathit{SD}(#1)}\xspace}

\newcommand{\sfpo}{\ensuremath{\mathit{SF}}\xspace}
\newcommand{\bsdpo}{\ensuremath{\mathit{SD}}\xspace}

\newcommand{\imprint}{imprint\xspace}
\newcommand{\imprints}{imprints\xspace}

\newcommand{\id}[1]{#1-identity\xspace}
\newcommand{\ids}[1]{#1-identities\xspace}
\newcommand{\gid}{\id{\Gs}}
\newcommand{\gids}{\ids{\Gs}}
\newcommand{\optid}[1]{optimal \id{#1}}
\newcommand{\optids}[1]{optimal \ids{#1}}

\newcommand{\Optids}[1]{Optimal \ids{#1}}
\newcommand{\goptid}{\optid{\Gs}}
\newcommand{\goptids}{\optids{\Gs}}

\newcommand{\Goptids}{\Optids{\Gs}}

\newcommand{\tame}{multiplicative\xspace}

\newcommand{\Ratms}{Rating maps\xspace}

\newcommand{\ratms}{rating maps\xspace}
\newcommand{\ratm}{rating map\xspace}
\newcommand{\rata}{rating algebra\xspace}
\newcommand{\ratas}{rating algebras\xspace}

\newcommand{\Nice}{Nice\xspace}
\newcommand{\nice}{nice\xspace}

\newcommand{\mratm}{multiplicative rating map\xspace}
\newcommand{\mratms}{multiplicative rating maps\xspace}

\newcommand{\Mratms}{Multiplicative rating maps\xspace}

\newcommand{\prin}[2]{\ensuremath{\Is[#1](#2)}\xspace}

\newcommand{\pprin}[2]{\ensuremath{\Ps[#1](#2)}\xspace}

\newcommand{\opti}[2]{\ensuremath{\Is_{#1}[#2]}\xspace}
\newcommand{\copti}[1]{\opti{\Cs}{#1}}

\newcommand{\popti}[3]{\ensuremath{\Ps_{#1}[#2,#3]}\xspace}

\usepackage{stmaryrd}

\newcommand{\brataux}[2]{\ensuremath{\xi_{#1}[#2]}\xspace}

\newcommand{\bratauxd}{\brataux{\Ds}{\rho}}

\newcommand{\bratauxsfc}{\brataux{\sfp{\Cs}}{\rho}}

\newcommand{\bratauxsfg}{\brataux{\sfp{\Gs}}{\rho}}

\newcommand{\iopti}[2]{\ensuremath{\fri_{#1}[#2]}\xspace}

\newcommand{\ioptig}[1]{\iopti{\Gs}{#1}}

\newcommand{\Cs}{\ensuremath{\mathcal{C}}\xspace}
\newcommand{\Ds}{\ensuremath{\mathcal{D}}\xspace}

\newcommand{\Gs}{\ensuremath{\mathcal{G}}\xspace}
\newcommand{\Hs}{\ensuremath{\mathcal{H}}\xspace}
\newcommand{\Is}{\ensuremath{\mathcal{I}}\xspace}

\newcommand{\Ps}{\ensuremath{\mathcal{P}}\xspace}

\newcommand{\Xs}{\ensuremath{\mathcal{X}}\xspace}
\newcommand{\Ys}{\ensuremath{\mathcal{Y}}\xspace}

\newcommand{\Fb}{\ensuremath{\mathbf{F}}\xspace}
\newcommand{\Gb}{\ensuremath{\mathbf{G}}\xspace}
\newcommand{\Hb}{\ensuremath{\mathbf{H}}\xspace}

\newcommand{\Kb}{\ensuremath{\mathbf{K}}\xspace}
\newcommand{\Lb}{\ensuremath{\mathbf{L}}\xspace}

\newcommand{\Ub}{\ensuremath{\mathbf{U}}\xspace}
\newcommand{\Vb}{\ensuremath{\mathbf{V}}\xspace}
\newcommand{\Wb}{\ensuremath{\mathbf{W}}\xspace}

\newcommand{\frB}{\ensuremath{\mathbbm{B}}\xspace}

\newcommand{\frI}{\ensuremath{\mathbbm{I}}\xspace}

\newcommand{\frP}{\ensuremath{\mathbbm{P}}\xspace}

\newcommand{\frS}{\ensuremath{\mathbbm{S}}\xspace}

\newcommand{\frb}{\ensuremath{\mathbbm{b}}\xspace}

\newcommand{\fri}{\ensuremath{\mathbbm{i}}\xspace}

\newcommand{\vari}{prevariety\xspace}

\newcommand{\varis}{prevarieties\xspace}

\def\inv{^{-1}}

\newcommand{\veps}{\ensuremath{\varepsilon}\xspace}

\newcommand{\nat}{\ensuremath{\mathbb{N}}\xspace}

\newcommand{\until}[2]{\ensuremath{#1~\textup{U}~#2}\xspace}
\newcommand{\finally}[1]{\ensuremath{\textup{F}~#1}\xspace}
\newcommand{\nex}[1]{\ensuremath{\textup{X}~#1}\xspace}

\newcommand{\untilp}[3]{\ensuremath{#2~\textup{U}_{#1}~#3}\xspace}
\newcommand{\sincep}[3]{\ensuremath{#2~\textup{S}_{#1}~#3}\xspace}
\newcommand{\finallyp}[2]{\ensuremath{\textup{F}_{#1}~#2}\xspace}

\newcommand{\finallyl}[1]{\finallyp{L}{#1}}

\newcommand{\prefsig}[1]{\ensuremath{\frP_{#1}}\xspace}

\newcommand{\prefsigc}{\prefsig{\Cs}}

\newcommand{\prefsigg}{\prefsig{\Gs}}

\newcommand{\infsig}[1]{\ensuremath{\frI_{#1}}\xspace}

\newcommand{\infsigc}{\infsig{\Cs}}

\newcommand{\infsigg}{\infsig{\Gs}}

\newcommand{\infix}[3]{\ensuremath{#1(#2,#3)}\xspace}
\newcommand{\suffix}[2]{\infix{#1}{#2}{|#1|+1}}
\newcommand{\prefix}[2]{\infix{#1}{0}{#2}}

\newcommand{\poschar}{\textup{\sffamily{Pos}}}

\newcommand{\pos}[1]{\ensuremath{\poschar(#1)\xspace}}
\newcommand{\posc}[1]{\ensuremath{\poschar_{c}(#1)\xspace}}

\title{Closing star-free closure}

\author{Thomas~Place}
\email{tplace@labri.fr}

\author{Marc~Zeitoun}
\email{mz@labri.fr}

\affiliation{%
  \institution{LaBRI, Universit\'e de Bordeaux}
  \streetaddress{351 cours de la Lib\'eration}
  \city{Talence}
  \postcode{F-33405}
  \country{France}
}



\begin{abstract}
  We introduce an operator on classes of regular languages, the star-free closure. Our motivation is to generalize standard results of automata theory within a unified framework. Given an arbitrary input class \Cs, the star-free closure operator outputs the least class closed under Boolean operations and language concatenation, and containing all languages of \Cs as well as all finite languages. We establish several equivalent characterizations of star-free closure: in terms of regular expressions, first-order logic, pure future and future-past temporal logic, and recognition by finite monoids. A key ingredient is that star-free closure coincides with another closure operator, defined in terms of regular operations where Kleene stars are allowed in restricted~contexts.

  A consequence of this first result is that we can decide membership of a regular language in the star-free closure of a class whose separation problem is decidable. Moreover, we prove that separation itself is decidable for the star-free closure of any finite class, and of any class of group languages having itself decidable separation (plus mild additional properties). We actually show decidability of a stronger property, called covering.
\end{abstract}

\begin{document}

\maketitle

\section{Introduction}
\label{sec:intro}
\textbf{\sffamily Context.} The starting point of this paper is a major result from the theory of regular languages. It~states that it is equivalent for a language of finite words to be defined by:
\begin{enumerate}[leftmargin=.9cm]
  \item\label{item:sf} a star-free regular expression, \emph{i.e.}, which forbids Kleene star but allows complement,
  \item\label{item:bsd} a regular expression restricting Kleene stars to prefix codes of bounded synchronization~delay,
  \item\label{item:fo} a first-order logic sentence using the linear order and the alphabetic predicates,
  \item\label{item:ltl} a pure future temporal logic formula,
  \item\label{item:ltlp} a future-past temporal logic formula,
  \item\label{item:ap} a morphism into a finite aperiodic monoid.
\end{enumerate}
This statement compiles a series of theorems by Schützenberger~\cite{schutzsf,schutzbd} for \mbox{$\eqref{item:sf}\Leftrightarrow\eqref{item:bsd}\Leftrightarrow\eqref{item:ap}$}, McNaughton and Papert~\cite{mnpfosf} for $\eqref{item:fo}\Leftrightarrow\eqref{item:sf}$ and Kamp~\cite{kltl} for $\eqref{item:fo}\Leftrightarrow\eqref{item:ltl}\Leftrightarrow\eqref{item:ltlp}$. It attests to the robustness of a class of languages defined by seemingly unrelated formalisms: various types of regular expressions, of logics and of machine-based devices. Moreover, Property~\eqref{item:ap} can be decided on a specific  canonical morphism, which can be computed from the language. This yields an algorithm for checking whether a given regular language has any of the aforementioned properties. In other words, the \emph{membership problem} of a regular language to this class of languages is decidable.

\smallskip
This result had a profound influence on automata theory: its impact went far beyond the class of star-free languages. By highlighting the correspondence between specific regular expressions, fragments of second-order monadic logic, variants of temporal logic and classes of finite monoids, it initiated a line of research whose aim is to capture the expressive power of natural classes of regular languages---see~\cite[Part~B]{pingoodref} or \cite{sw-handbook21} for overviews. As in the theorem above, these classes are defined by restricting the syntax of the aforementioned formalisms. Historically, the way to study such a class was inspired by Schützenberger's contribution to the above result: the aim was to design \emph{membership algorithms}. There is an abundant literature on the subject, due to the number of interesting classes of regular languages. See for example~\cite{simonthm,knast83,ChaubardPS06,jeppgbg} (for variations on the quantifier alternation free fragment of first-order logic),~\cite{pwdelta2,glasserdd,sig2mod} (for variations on a more expressive fragment) or~\cite{twfo2,fo2mod,between} (for variations on two-variable first-order logic).

\medskip\noindent\textbf{\sffamily Operators.}
However, the number of publications in the field can also be explained by the fact that the classes that were investigated do not have a unique flavor. Indeed, logic and regular expressions come in a multitude of variants. For example, other versions of first-order logic can be envisaged by extending its signature, \emph{i.e.}, by allowing more predicates, thus increasing the expressive power. This leads to two classic variants: we can add predicates that test the value of a position modulo a certain integer, and more generally predicates that count the number of occurrences of a particular letter modulo an integer. Similarly, star-free expressions can be extended in a natural way: instead of starting with singleton languages, we can start with languages of a certain fixed class. Finally, there are several  extensions of temporal logic, usually obtained by adding more expressive temporal modalities (see for example~\cite{msoltl}).

\smallskip
Naturally, the historical approach has been to treat each of these variations individually. This means that the proofs have to be recast for each variation, which is often technical and sometimes nontrivial. To avoid such adaptations, it is desirable to develop a generic approach, which would encompass several variations of a given class at once. This is where the notion of operator comes in. An \emph{operator} $\mathit{Op}$ associates with any class \Cs of regular languages a larger class $\mathit{Op}(\Cs)$. For example, the star-free closure operator \sfpo takes as input a class \Cs and outputs \sfp{\Cs}, which is the least class containing \Cs and all finite languages, and which is closed under union, complement and concatenation. Notice that we recover the class of star-free languages as the star-free closure of the class consisting of two languages: the empty and the full languages.

\smallskip
Focusing on operators rather than on individual classes meets our main objective (understanding classes of regular languages). Indeed, most interesting classes are obtained from simpler ones by applying operators from a small set. The main operators are Boolean and polynomial closure~\cite{pwdelta2} (they appear in concatenation hierarchies, see for example~\cite{jep-dd45}), unambiguous polynomial closure~\cite{pzupol2} and star-free closure, which is the subject of this paper. Actually, it is more rewarding to concentrate on operators, as this allows multiple variants of the same class to be handled at once, leading to \emph{generic} results. Not only this avoids reproducing proofs for classes that are variations of each other, but also and more importantly, this simplifies the proofs and emphasizes the characteristics of the operator $\mathit{Op}$ and the assumptions needed on the class \Cs to decide $\mathit{Op}(\Cs)$-membership.

\smallskip
Ideally, for an operator $\mathit{Op}$, we would like to reduce $\mathit{Op}(\Cs)$-membership to \Cs-membership, \emph{i.e.}, to obtain a statement like: ``If \Cs has decidable membership, then so does $\mathit{Op}(\Cs)$''. Unfortunately, although this situation may occur~\cite{pzupol2}, it is uncommon: decidability of membership is rarely preserved by operators (see~\cite{ABR:identity92}, which provides negative examples in the context of classes of monoids). This observation leads to the following question:
\begin{quote}
  ``What properties should \Cs satisfy for the $\mathit{Op}(\Cs)$-membership problem to be~\hbox{decidable}?''
\end{quote}

This question, in turn, motivates us to consider a new problem: \emph{\Cs-separation}. It asks whether two regular input languages can be separated by a language from the class~\Cs, \emph{i.e.}, whether there exists a language from~\Cs containing the first input language while being disjoint from the second. There is an easy algorithmic reduction from \Cs-membership to \Cs-separation: a language belongs to~\Cs if and only if it can be \Cs-separated from its complement. Note that separation is more demanding than membership: it requires to exhibit a separating language, if possible, even when none of the input languages belong to the class under study. In contrast, solving membership only requires to prove that the input language does or does not belong to the class. For this reason separation is also more rewarding than membership: although more difficult, it brings more information, which can later be exploited to tackle classes of languages built on top of the one being investigated.

\smallskip
In particular, looking at separation provides a partial answer to the above question (``what properties should \Cs satisfy for the $\mathit{Op}(\Cs)$-membership problem to be decidable?''). Indeed, for some operators $\mathit{Op}$, being able to decide \Cs-separation is sufficient to decide $\mathit{Op}(\Cs)$-membership. This is the case when \textit{Op} is the polynomial closure operator~\cite{PZ:generic19} (assuming mild properties on~\Cs). If in addition, the class~\Cs consists of group languages (see Section~\ref{lem:gmorph}), this is also the case for the Boolean closure of the polynomial closure~\cite{pzconcagroup} (which, in fact, has then decidable \emph{separation}). For this reason, separation has replaced membership as the standard problem to understand a class of regular languages. It turns out that in order to tackle \Cs-separation, it is convenient to study an even more general problem called \emph{\Cs-covering}. Intuitively, it generalizes separation to an arbitrary number of input languages. The state of the art regarding the class of star-free languages is that it has decidable covering, hence also decidable separation (this follows from~\cite{pzfoj} and indirectly from~\mbox{\cite{Henckell88,MR1709911}}).

\medskip\noindent\textbf{\sffamily Contributions}. We investigate the star-free closure \emph{operator}. With any class of languages \Cs, it associates  the least class \sfp\Cs  containing \Cs, all finite languages, and which is closed under Boolean operations and language concatenation. Note that these operations preserve regularity and that Kleene star is explicitly forbidden. We generalize the known results in two orthogonal~directions:

\begin{itemize}[label=\small$\bullet$] 
  \item First, we generalize the Kamp-McNaughton-Papert-Schützenberger theorem. This means finding appropriate generalizations for each of the properties appearing in this theorem, and showing that they all characterize star-free closure. In other words, we need to find suitable operators generalizing the definition of the classes that appear in this result: languages of bounded synchronization delay, first-order definable languages, languages definable in pure future and future-past temporal logic, and languages recognized by finite aperiodic monoids. An important consequence of the algebraic characterization, is that \sfp\Cs-membership reduces to \Cs-separation.

  \item Secondly, we prove that under certain (strong) assumptions on the input class~\Cs, which we detail below, the covering problem for the star-free closure \sfp\Cs is decidable.
\end{itemize}

Let us comment on these two contributions. Concerning the first, one of the operators we have to define already exists: with each class \Cs, one can associate a variant \foc of first-order logic whose predicates depend on~\Cs~\cite{PZ:generic19}. It defines exactly the languages in the star-free closure of~\Cs. Its definition is simple: each language $L$ in \Cs yields a binary predicate that selects pairs of positions such that the infix between them belongs to~$L$.
On the other hand, all other operators are~new.

\smallskip
The main one is the $\bsdpo$ operator. It generalizes a class defined by Schützenberger~\cite{schutzbd} (see also~\cite{DiekertW16a,DiekertW17}). Roughly speaking, \bsdp\Cs is the least class containing all finite languages which is closed under intersection with languages of \Cs, disjoint union, unambiguous concatenation, and Kleene star applied to prefix codes of bounded synchronization delay. Unlike \sfpo, the \bsdpo operator prohibits complement. In fact, the definitions of these operators are of a different nature: the restrictions for \sfpo are syntactic (they constrain legal regular expressions), whereas being a disjoint union, an unambiguous concatenation or a prefix code with bounded synchronization delay are semantic notions: they depend on the languages themselves, not just on expressions used to describe them.

The \bsdpo operator is a key ingredient in the generalization of the Kamp-McNaughton-Papert-Schützenberger theorem: the first step, establishing the inclusion  $\bsdp\Cs\subseteq\sfp\Cs$, is particularly helpful. Indeed, proving inclusion in \sfp\Cs  is generally difficult, since this requires the construction of expressions that involve alternating complement and concatenation operations, which are hard to understand. On the other hand,
proving inclusion in \bsdp\Cs is easier, as we may use Kleene stars. In fact, several of the article's proofs are based on this capability.

\smallskip
The proof of the converse inclusion $\sfp\Cs\subseteq\bsdp\Cs$ is intertwined with the algebraic characterization. Here, we have to generalize Property~\eqref{item:ap}, which involves finite aperiodic monoids (\emph{i.e.}, which are such that the sequence of powers of any element eventually stabilizes). Given a monoid morphism into a finite monoid $M$, we define monoids in $M$ called \emph{\Cs-orbits} for this morphism. They are computable as soon as \Cs-separation is decidable. The generalized algebraic characterization states that a language is in $\bsdp\Cs=\sfp\Cs$ if and only if all the \Cs-orbits of its syntactic morphism are aperiodic. In particular, if \Cs has decidable separation, membership in \sfp\Cs is decidable. This is the way we generalize Schützenberger's membership~theorem.

\smallskip
At last, we generalize the correspondences with temporal logic. We first define an operator that associates a variant of pure future temporal logic with each class. This simply amounts to generalizing the ``Until'' temporal modality to take into account the input class~\Cs. More precisely, each language $L$ of \Cs produces a new ``Until'' modality $\untilp{L}{\!}{\!}$. Intuitively, this modality adds a constraint to the semantic of the standard Until: a formula $\untilp{L}\phi\psi$ holds at position $i$ in a word when there exists a position $j>i$ where $\psi$ holds, such that $\phi$ holds on all the intermediate positions, and such that the infix between $i$ and $j$ belongs to $L$. Adapting this construction to future-past temporal logic is straightforward. Again, we show that both temporal logic operators obtained in this way correspond to the star-free closure operator, thus generalizing Properties~\ref{item:ltl} and \ref{item:ltlp} of the Kamp-McNaughton-Papert-Schützenberger theorem.

\medskip
We now turn to the second contribution: covering algorithms for specific input classes. First, we show that the star-free closure of a \emph{finite class} has decidable covering (and therefore, decidable separation). We then use this result to establish our main theorem: the star-free closure of a class of \emph{group languages} with decidable \emph{separation} has decidable \emph{covering} (and therefore again, decidable separation). Let us mention some important features of this work.

\smallskip
A first point is that the case of a finite class is important by itself. Foremost, it is a crucial step for the main result on the star-free closure of classes of group languages (this is due to the fact that a language in the star-free closure of a class is built using a finite number of languages of the class). Second, it provides a new proof that covering is decidable for the original class of star-free languages (this is shown in~\cite{pzfo} or can be derived from~\cite{Henckell88,MR1709911}). This new proof is simpler and generic. While the original underlying technique goes back to Wilke~\cite{wilkeltl}, the proof has been simplified at several levels. The main simplification is obtained thanks to an abstract framework, introduced in~\cite{pzcovering2}. It is based on the central notion of \ratm, which is meant to measure the quality of a separator. For the framework to be relevant, we actually need to generalize separation to multiple input languages, which leads to the covering problem.
Another key difference is that existing proofs (specific to star-free languages) involve abstracting words by new letters at some point, which requires the working alphabet to be a parameter of the induction. Here, we cannot use this approach as the classes we build with star-free closure are less robust in general. We work with a fixed alphabet, which also makes the proof simpler. In fact, several proofs should look similar to the reader. This is not surprising, since in order to establish membership or covering, we have to build languages from the classes we are interested in.

\smallskip\noindent\textbf{\sffamily Applications.} Finally, let us present important applications of the result about covering for classes made of group languages.  First, one may look at the input class containing all group languages. Straubing~\cite{STRAUBING1979319} described an algebraic counterpart of the star-free closure of this class, whose membership was then shown to be decidable by Rhodes and Karnofsky~\cite{Karnofsky1982}. Altogether, this implies that membership is decidable for the star-free closure of group languages, as noted by Margolis and Pin~\cite{margpin85}. Here, we are able to generalize this result to separation and covering, as separation is known to be decidable for the class of all group languages~\cite{Ash91,pzgroups}.

Another important application is the class of languages definable by first-order logic with modular predicates \fowm. This class is known to have decidable membership~\cite{MIXBARRINGTON1992478}. Moreover, it is the star-free closure of the class consisting of the languages counting the length of words modulo some number. Since this input class is easily shown to have decidable separation (see~\cite{pzconcagroup} for example), our main theorem applies.

The third application concerns first-order logic endowed with predicates counting the number of occurrences of a letter before a position, modulo some integer. Indeed, the class of languages definable in this logic is exactly the star-free closure of the class of languages recognized by Abelian groups (this follows from a generic correspondence theorem between star-free closure of a class and variants of first-order logic~\cite{pzconcagroup,pinbridges}, as well as from the description of languages recognized by Abelian groups~\cite{EilenbergB}). Again, our main theorem applies, since the class of Abelian groups is known to have decidable separation: this follows from~\cite{abelianp} and~\cite{MR1709911} (see also~\cite{pzgroups}).

\smallskip\noindent\textbf{\sffamily Organization.} The paper is structured as follows. We set up the notation and recall the background in Section~\ref{sec:prelims}. We introduce the star-free closure operator and present some of its basic properties in Section~\ref{sec:sfdef}. In the same section, we introduce classes of group languages, for which this operator produces relevant classes. We define prefix codes of bounded synchronization delay in Section~\ref{sec:bsd} and the associated operator \bsdpo, which allows   Kleene star to be applied only to these languages. We also show that this new operator can be simulated by the star-free closure. We then develop in Section~\ref{sec:carac} the material needed to establish, for a class \Cs with mild properties, a common algebraic characterization of \sfp\Cs and \bsdp\Cs (thus proving the missing inclusion $\sfp\Cs\subseteq\bsdp\Cs$). As explained above, this characterization is decidable as soon as separation is decidable for the underlying class~\Cs. We establish the correspondences of star-free closure with first-order logic in Section~\ref{sec:folog} and with temporal logic in Section~\ref{sec:ltl}. Finally, we consider the covering problem. In Section~\ref{sec:ratms}, we recall the framework of \ratms, which is convenient for handling covering. We then prove that the star-free closure operator outputs a class whose covering is decidable in two cases: in Section~\ref{sec:finite}, when the input class is finite and in Section~\ref{sec:group}, when it is composed of group languages (plus lightweight additional properties).

\smallskip\noindent\textbf{\sffamily Related paper.} This paper completes results from~\cite{pzsfclos} and extend them.


\section{Preliminaries}
\label{sec:prelims}
In this section, we introduce the terminology used in the paper. We also present the membership, separation and covering problems, as well as key mathematical tools designed to handle them.

\subsection{Classes of regular languages}

For the whole paper, we fix a finite alphabet $A$. We denote by $A^*$ the set of all \emph{finite} words over $A$, including the empty word \veps. We let $A^{+}=A^{*}\setminus\{\veps\}$. For $u,v \in A^*$, we write $uv$ the word obtained by concatenating $u$ and $v$. Moreover, for every $w\in A^*$, we write $|w| \in \nat$ for its length. We shall also consider \emph{positions}. A word $w =a_1 \cdots a_{|w|} \in A^*$ is viewed as an \emph{ordered set $\pos{w} = \{0,1,\dots,|w|,|w|+1\}$ of $|w|+2$ positions}. A position $i$ such that $1 \leq i \leq |w|$ carries label $a_i \in A$. We write $\posc{w} = \{1,\dots,|w|\}$ for this set of labeled positions.  On the other hand, positions $0$ and $|w|+1$ are \emph{artificial} leftmost and rightmost positions, which carry \emph{no label}. Finally, given a word $w= a_1\cdots a_{|w|} \in A^*$ and $i,j \in \pos{w}$ such that $i < j$, we write $\infix{w}{i}{j} = a_{i+1} \cdots a_{j-1} \in A^*$ (\emph{i.e.}, the infix obtained by keeping the letters carried by the positions that are \emph{strictly} between $i$ and~$j$). Note that $\infix{w}{0}{|w|+1} = w$.

A \emph{language} is a subset of $A^*$. It is standard to extend concatenation to languages: given $K,L \subseteq A^*$, we write~$KL = \{uv \mid u \in K \text{ and } v \in L\}$. Finally, we use the Kleene star: if $K \subseteq A^*$, then $K^+$ denotes the union of all languages $K^n$ for $n \geq 1$ and $K^*$ denotes the language $K^+\cup \{\veps\}$.

\medskip
\noindent
{\bf Classes.} A class of languages \Cs is a set of languages. Such a class \Cs is a \emph{lattice} when $\emptyset\in\Cs$, $A^* \in \Cs$ and \Cs is closed under union and intersection: for every $K,L \in \Cs$, we have $K \cup L \in \Cs$ and $K \cap L \in \Cs$. A \emph{Boolean algebra} is a lattice which is closed under complement: if $K\in \Cs$, then $A^* \setminus K \in \Cs$. Finally, a class \Cs is \emph{quotient-closed} if for every $L \in \Cs$ and $u \in A^*$, the following properties~hold:
\[
  u^{-1}L \stackrel{\text{def}}{=}\{w\in A^*\mid uw\in L\} \text{\quad and\quad} Lu^{-1} \stackrel{\text{def}}{=}\{w\in A^*\mid wu\in L\}\text{\quad both belong to \Cs}.
\]
A \emph{\vari} is a quotient-closed Boolean algebra containing only \emph{regular languages}. The regular languages are those which can be equivalently defined by nondeterministic finite automata, finite monoids or monadic second-order logic. We work with the definition by monoids, which we recall~now.

\medskip
\noindent
{\bf Finite monoids and regular languages.} A \emph{semigroup} is a set $S$  endowed with an associative multiplication $(s,t)\mapsto s\cdot t$ (also denoted by~$st$). A \emph{monoid} is a semigroup $M$ whose multiplication has an identity element $1_M$, \emph{i.e.}, such that ${1_M}\cdot s=s\cdot {1_M}=s$ for every~$s \in M$. 

An \emph{idempotent} of a semigroup $S$ is an element $e \in S$ such that $ee = e$. We write $E(S) \subseteq S$ for the set of all idempotents in $S$. It is folklore that for every \emph{finite} semigroup~$S$, there exists a natural number $\omega(S)$ (denoted by $\omega$ when $S$ is understood) such that for every $s \in S$, the element $s^\omega$ is an idempotent.

Clearly, $A^*$ is a monoid whose multiplication is concatenation (the identity element is \veps). Thus, we may consider morphisms $\alpha: A^* \to M$ where $M$ is an arbitrary monoid. That is, $\alpha:A^*\to M$ is a map satisfying $\alpha(\veps)=1_M$ and $\alpha(uv)=\alpha(u)\alpha(v)$ for all $u,v\in A^*$. Given such a morphism and some language $L \subseteq A^*$, we say that $L$ is \emph{recognized} by $\alpha$ when there exists a subset $F$ of $M$ such that $L= \alpha\inv(F)$. It is standard and well known that the regular languages are those which can be recognized by a morphism into a \emph{finite} monoid.

\medskip
\noindent
{\bf Syntactic morphism.} Every language $L$ is recognized by a canonical morphism. Let us briefly recall its definition. One may associate to $L$ an equivalence relation $\equiv_L$ over $A^*$: the \emph{syntactic congruence of\/~$L$}. Given $u,v \in A^*$, we let,
\[
  \text{$u \equiv_L v$ if and only if $xuy \in L \Leftrightarrow xvy \in L$ for every $x,y \in A^*$}.
\]
As the name suggests, it is known and simple to verify that ``$\equiv_L$'' is a congruence on $A^*$: it is reflexive, symmetric and transitive, and for every $u,u',v,v' \in A^*$ such that $u \equiv_L v$ and $u' \equiv_L v'$, we have $uu' \equiv_L vv'$. Thus, the set of equivalence classes $M_L = {A^*}/{\equiv_L}$ is a monoid. It is called the \emph{syntactic monoid of\/ $L$}. Moreover, the map $\alpha_L: A^* \to M_L$ sending every word to its equivalence class is a morphism recognizing~$L$, called the \emph{syntactic morphism of\/ $L$}. Another characterization of regular languages is that $L$ is regular if and only if $M_L$ is finite (\emph{i.e.}, $\equiv_L$ has finite index): this is  Myhill-Nerode's theorem. In this case, one may compute the syntactic morphism $\alpha_L: A^* \to M_L$ from any representation of $L$ (such as a finite automaton or an arbitrary monoid morphism).

\subsection{Decision problems} We look at three decision problems. They all depend on an arbitrary class of languages \Cs. We use them as mathematical tools for analyzing \Cs. Indeed, obtaining an algorithm for one of these three problems requires a solid understanding of~\Cs.

The \emph{\Cs-membership problem} is the simplest one. It takes as input a single regular language~$L$ and asks whether $L\in \Cs$. The second problem, \emph{\Cs-separation}, is more general. Given three languages $K,L_1,L_2$, we say that $K$ \emph{separates} $L_1$ from $L_2$ if we have $L_1 \subseteq K$ and $L_2 \cap K = \emptyset$. Given a class of languages \Cs, we say that $L_1$ is \emph{\Cs-separable} from $L_2$ if some language in \Cs separates $L_1$ from~$L_2$. Observe that when \Cs is not closed under complement, the definition is not symmetrical: it is possible for $L_1$ to be \Cs-separable from $L_2$ while $L_2$ is not \Cs-separable from $L_1$. The separation problem associated to a given class \Cs, also called \emph{\Cs-separation problem}, takes two regular languages $L_1$ and $L_2$ as input and asks whether $L_1$ is \Cs-separable from $L_2$.

\begin{rem}
  The \Cs-separation problem generalizes the \Cs-membership problem. Indeed, a regular language belongs to $\Cs$ if and only if it is \Cs-separable from its complement, which is also regular.
\end{rem}

In the paper, we do not consider separation directly. Instead, we work with a third, even more general problem: \Cs-covering. It was introduced in~\cite{pzcovering2} and takes as input a single regular language $L_1$ and a \emph{finite set of regular languages} $\Lb_2$. It asks whether there exists a ``\Cs-cover of $L_1$ which is separating for $\Lb_2$''.

Given a language $L$, a \emph{cover of $L$} is a \emph{finite} set of languages \Kb such that $L \subseteq \bigcup_{K \in \Kb} K$. A cover \Kb is a \emph{\Cs-cover} if all languages $K \in \Kb$ belong to \Cs. Moreover, given two finite sets of languages \Kb and \Lb, we say that \Kb is \emph{separating} for \Lb if for every $K\in\Kb$, there exists $L\in\Lb$ such that $K \cap L = \emptyset$. Finally, given a language $L_1$ and a finite set of languages~$\Lb_2$, we say that the pair $(L_1,\Lb_2)$ is \emph{\Cs-coverable} if there exists a \Cs-cover of $L_1$ which is separating~for~$\Lb_2$.

The \emph{\Cs-covering problem} is now defined as follows. Given as input a regular language $L_1$ and a finite set of regular languages $\Lb_2$, it asks whether the pair $(L_1,\Lb_2)$ \Cs-coverable. It is straightforward to prove that covering generalizes separation if the class~\Cs is a lattice, as stated in the following lemma (see~\cite[Theorem~3.5]{pzcovering2} for the proof).

\begin{lem} \label{lem:covsep}
  Let \Cs be a lattice and $L_1,L_2$ be two languages. Then $L_1$ is \Cs-separable from~$L_2$ if and only if $(L_1,\{L_2\})$ is \Cs-coverable.
\end{lem}

\subsection{\Cs-morphisms}

We now present a central mathematical tool. Consider an arbitrary \vari \Cs. A \emph{\Cs-morphism} is a \emph{surjective} morphism $\eta: A^*\to N$ into a finite monoid~$N$ such that every language recognized by $\eta$ belongs to \Cs. While basic, the notion of \Cs-morphism is a central tool in the paper. First, it is connected to the membership problem via the following simple, yet crucial proposition.

\begin{prop} \label{prop:synmemb}
  Let \Cs be a \vari. A regular language~$L$ belongs to~\Cs if and only if its syntactic morphism $\alpha_L: A^* \to M_L$ is a \Cs-morphism.
\end{prop}

\begin{proof}
  The ``if'' implication is immediate since $L$ is recognized by its syntactic morphism. We prove the converse one: assuming that $L \in \Cs$, we prove that every language recognized by $\alpha_L: A^* \to M_L$ belongs to \Cs (recall that syntactic morphisms are surjective by definition). Clearly, every such language is a union of $\equiv_L$-classes. Hence, as \Cs is a \vari, it suffices to prove that all $\equiv_L$-classes belongs to \Cs. For every $s \in M_L$, we fix a word $x_s \in A^*$ such that $\alpha_L(x_s) = s$. Consider the following equivalence $\sim$ on $A^*$:
  \[
    \text{$u \sim v$ if and only if $x_sux_t \in L \Leftrightarrow x_svx_t \in L$ for every $s,t \in M_L$}.
  \]
  We first show that ${\sim}$ and $\equiv_L$ are the same relation. It is immediate by definition that ${\equiv_L} \subseteq {\sim}$. For the converse inclusion, let $u,v \in A^*$ be such that $u \sim v$. We prove that $u \equiv_L v$. Given $x,y \in A^*$, we need to prove that $xuy \in L \Leftrightarrow xvy\in L$. Let $s = \alpha_L(x)$ and $t = \alpha_L(y)$. By definition, we have $\alpha_L(xuy) = \alpha_L(x_sux_t)$ and $\alpha_L(xvy) = \alpha_L(x_svx_t)$. Consequently, $xuy \in L \Leftrightarrow x_sux_t \in L$ and $xvy \in L \Leftrightarrow x_svx_t \in L$. Finally, since $u \sim v$, we know that $x_sux_t \in L \Leftrightarrow x_svx_t \in L$. Altogether, this yields $xuy \in L \Leftrightarrow xvy\in L$, as desired.

  It remains to prove that every $\sim$-class belongs to \Cs. Let $w \in A^*$. We define the following subset $P_w$ and $N_w$ of $M_L^2$:
  \[
    P_w = \big\{(s,t) \in M_L^2 \mid x_swx_t \in L\big\} \quad \text{and} \quad N_w = \big\{(s,t) \in M_L^2 \mid x_swx_t \not\in L\big\}.
  \]
  One may now verify from the definition of $\sim$ that the $\sim$-class of $w$ is the following language:
  \[
    \biggl(\bigcap_{(s,t) \in P_w}\left(x_s\inv L x_t\inv\right)\biggr) \setminus \biggl(\bigcup_{(s,t) \in N_w}\left(x_s\inv L x_t\inv\right)\biggr).
  \]
  Since $L \in \Cs$ and \Cs is a \vari, it follows that the $\sim$-class of $w$ belongs to \Cs, which completes the~proof.
\end{proof}

In view of Proposition~\ref{prop:synmemb}, getting an algorithm for \Cs-membership boils down to finding a procedure to decide whether an input morphism $\alpha: A^* \to M$ is a \Cs-morphism. This is how we approach the question in the paper. 

\smallskip

Additionally, we shall use \Cs-morphisms as mathematical tools in proof arguments. They are convenient when manipulating arbitrary classes. We present a few properties that we shall need in this context. First, we have the following simple corollary of Proposition~\ref{prop:synmemb}.

\begin{prop}\label{prop:genocm}
  Let \Cs be a \vari and consider finitely many languages $L_1,\dots,L_k$ of $\Cs$. There exists a \Cs-morphism $\eta: A^* \to N$ such that $L_1,\dots,L_k$ are all recognized by $\eta$.
\end{prop}

\begin{proof}
  For every $i \leq k$, let $\alpha_i: A^* \to M_i$ be the syntactic morphism of $L_i$. We know from Proposition~\ref{prop:synmemb} that $\alpha_i$ is a \Cs-morphism. Let $M = M_1 \times \cdots \times M_k$ be the monoid equipped with the componentwise multiplication. Moreover, let $\alpha: A^* \to M$ be the morphism defined by $\alpha(w) = (\alpha_1(w_1),\dots,\alpha_k(w))$ for every $w \in A^*$. One may verify from the definition of $\alpha$ that all languages recognized by $\alpha$ are finite intersections of languages recognized by $\alpha_1,\dots,\alpha_k$ (in particular, $\alpha$ recognizes each $L_i$). Hence, all languages recognized by $\alpha$ belong to \Cs. It now suffices to let $\eta: A^* \to N$ be the surjective restriction of $\alpha$ to complete the proof.
\end{proof}

Finally, we consider the special case when \Cs is a \emph{finite} \vari (\emph{i.e.}, \Cs contains finitely many languages). In this case, Proposition~\ref{prop:genocm} yields a \Cs-morphism recognizing \emph{all} languages in \Cs. The following lemma implies that it is unique (up to renaming).

\begin{lem}\label{lem:cmdiv}
  Let \Cs be a finite \vari and let $\alpha: A^* \to M$ and $\eta: A^* \to N$ be two \Cs-morphisms. If $\alpha$ recognizes all languages in \Cs, then there exists a morphism $\gamma: M \to N$ such that $\eta = \gamma \circ \alpha$.
\end{lem}

\begin{proof}
  Assume that  $\alpha$ recognizes all languages in \Cs. We define $\gamma: M \to N$. For every $s \in M$, we fix a word $w_s \in \alpha\inv(s)$ (recall that \Cs-morphisms are surjective by definition) and define $\gamma(s) = \eta(w_s)$. It remains to prove that $\gamma$ is a morphism and that $\eta = \gamma \circ \alpha$. It suffices to prove the latter: since $\alpha$ is surjective, the former is an immediate consequence. Let $v \in A^*$. We show that $\eta(v) = \gamma(\alpha(v))$. Let $s = \alpha(v)$. By definition, $\gamma(s) = \eta(w_s)$. Hence, we need to prove that $\eta(v) = \eta(w_s)$. Since $\eta$ is a \Cs-morphism, we have $\eta\inv(\eta(w_{s})) \in \Cs$. Hence, our hypothesis implies that $\eta\inv(\eta(w_{s}))$ is recognized by $\alpha$. Since it is clear that $w_s \in \eta\inv(\eta(w_{s}))$ and $\alpha(v) = \alpha(w_s) = s$, it follows that $v \in \eta\inv(\eta(w_{s}))$ which exactly says that $\eta(v) = \eta(w_s)$, completing the proof.
\end{proof}

By Lemma~\ref{lem:cmdiv}, if \Cs is a finite \vari and $\alpha: A^* \to M$ and $\eta: A^* \to N$ are two \Cs-morphisms which both recognize \emph{all} languages in \Cs, there exist two morphisms $\gamma: M \to N$ and $\beta: N \to M$ such that $\eta=\gamma \circ \alpha$ and $\alpha = \beta \circ \eta$. This yiedls $\beta\circ\gamma\circ\alpha=\alpha$. Since $\alpha$ is surjective, it follows that $\beta \circ \gamma: M \to M$ is the identity morphism. Hence, both $\beta$ and $\gamma$ are isomorphisms, meaning that $\alpha$ and $\eta$ are the same object up to renaming. We call it the \emph{canonical \Cs-morphism} and denote it by $\etac: A^* \to \canc$. Let us emphasize that this object is only defined when \Cs is a \emph{finite \vari}.


\section{Star-free closure}
\label{sec:sfdef}
In this section, we introduce the classes investigated in the paper. Each of them is built from a simpler input class using a single operator: the \emph{star-free closure}, which we first define. Then, we present classes that we use as key inputs for this operator: those containing only \emph{group languages}.

\subsection{Definition} 

Consider a class \Cs. The \emph{star-free closure of\/ \Cs}, denoted by \sfp{\Cs}, is the least class of languages containing \Cs and all singletons $\{a\}$ for $a \in A$, and which is closed under union, complement and concatenation (that is, if $K,L \in \sfp{\Cs}$, then $K \cup L$, $A^* \setminus K$ and $KL$ belong to \sfp{\Cs} as well).

\begin{rem}
  Star-free closure is the generalization of a prominent single class: the class \sfr of \emph{star-free languages}. It contains exactly the languages that can be defined by a star-free expression, \emph{i.e.}, a~regular expression that cannot use the Kleene star but can use complement instead. One may verify that \sfr is exactly the star-free closure of the  class $\{\emptyset,A^*\}$, \emph{i.e.}, $\sfr = \sfp{\{\emptyset,A^*\}}$. Naturally, \sfr is also the star-free closure of itself. Therefore, $\sfp{\Cs}=\sfr$ for every class \Cs  included in \sfr and containing $\{\emptyset,A^*\}$. It follows that investigating the star-free closure is worthwhile only when applied to a class which is \textbf{not} included in \sfr. As we explain below, typical such classes are made of \emph{group languages}.
\end{rem}

In practice, we only apply star-free closure to input classes \Cs that are \varis. In this case, \sfp{\Cs} is a \vari as well. We prove this in the following proposition.

\begin{prop}\label{prop:vari}
  If  \Cs is a \vari, then \sfp{\Cs} is a \vari closed under concatenation.
\end{prop}

\begin{proof}
  It is immediate by definition that \sfp{\Cs} is a Boolean algebra closed under concatenation and containing only regular languages (indeed, it is well-known that regular languages are closed under Boolean operations and concatenation). Hence, it suffices to prove that \sfp{\Cs} is quotient-closed. By symmetry, we only present a proof for left quotients. Let $w \in A^*$. We use induction on the length of $w$ to prove that for every $L \in \sfp{\Cs}$, we have $w\inv L \in \sfp{\Cs}$. When $w = \veps$, we have $w\inv L=L$ for every  $L \subseteq A^*$. Hence, the result is trivial. Assume now that $w \in A^+$ and consider $L \in \sfp{\Cs}$. By hypothesis, there exist $u \in A^*$ and $a \in A$ such that $w = ua$. Hence, $w\inv L =  u\inv(a\inv L)$. We use a sub-induction on the construction of $L$ to prove that $a\inv L\in\sfp{\Cs}$. It will then be immediate by induction on the length of $w$ that $w\inv L =  u\inv(a\inv L) \in \sfp{\Cs}$.

  Since $L \in \sfp{\Cs}$, it is built from languages in \Cs and the singletons $\{b\}$ for $b \in A$ using only union, complement and concatenation. We use induction on this construction to prove that $a\inv L\in\sfp{\Cs}$. Assume first that $L \in \Cs$. In this case, $a\inv L\in\Cs\subseteq \sfp{\Cs}$ since \Cs is a \vari. Assume now that $L = \{b\}$ for some $b \in A$. Then, either $b \neq a$ and $L = \emptyset \in\Cs\subseteq \sfp{\Cs}$, or  $b = a$ and $L = \{\veps\}$, which also belongs to $\sfp\Cs$ since it is the complement of the union of all languages $A^*cA^*$ for~$c\in A$.

  We turn to the inductive cases. First, assume that $L = L_1 \cup L_2$ for languages $L_1,L_2 \in \sfp{\Cs}$ for which, by induction, we have $a\inv L_1 \in \sfp{\Cs}$ and $a\inv L_2 \in \sfp{\Cs}$. Since $a\inv L = a\inv L_1 \cup a\inv L_2$, we get $a\inv L \in \sfp{\Cs}$, as desired. Assume now that $L = A^* \setminus H$ for $H \in \sfp{\Cs}$ such that, by induction, $a\inv H \in \sfp{\Cs}$. One may verify that $a\inv L = A^* \setminus \left(a\inv H\right)$. Hence, we get  $a\inv L \in \sfp{\Cs}$, as desired. Finally, assume that $L = L_1L_2$ for  languages $L_1,L_2 \in \sfp{\Cs}$ for which, by induction, we have $a\inv L_1 \in \sfp{\Cs}$ and $a\inv L_2 \in \sfp{\Cs}$. One may verify that,
  \[
    a\inv L =
    \left\{
      \begin{array}{ll}
        (a\inv L_1)L_2 & \text{if $\veps \not\in L_1$}, \\
        (a\inv L_1)L_2 \cup a\inv L_2 & \text{if $\veps \in L_1$}.
      \end{array}
    \right.
  \]
  Since $a\inv L_1 \in \sfp{\Cs}$ and $a\inv L_2 \in \sfp{\Cs}$ by induction, we get $a\inv L \in \sfp{\Cs}$, which concludes the~proof.
\end{proof}

We complete the presentation with a characteristic property of star-free closure (for input classes that are \varis). We present it as a property of the \sfp{\Cs}-morphisms.

\begin{prop}\label{prop:aper}
  Let \Cs be a \vari and let $\alpha: A^* \to M$ be an \sfp{\Cs}-morphism. There exists a \Cs-morphism $\eta: A^* \to N$ such that:
  \[
    \text{for every } u \in A^*\text{, if } \eta(u)\text{ is idempotent, then } (\alpha(u))^{\omega+1} = (\alpha(u))^{\omega}.\]
\end{prop}

\begin{proof}
  Since $\alpha$ is an \sfp{\Cs}-morphism, we have $\alpha\inv(s) \in \sfp{\Cs}$ for every $s \in M$. This means that $\alpha\inv(s)$ is built from finitely many languages of \Cs and from the singletons $\{a\}$ (for $a \in A$) using union, complement and concatenation. Since \Cs is a \vari of regular languages, Proposition~\ref{prop:genocm} yields a \Cs-morphism $\eta: A^* \to N$ recognizing all the languages in~\Cs used in the construction of the languages $\alpha\inv(s)$ for $s \in M$. Let  $\Cs_{\eta}=\sfp{\{\eta^{-1}(t)\mid t\in N\}}$. By definition of $\eta$, we know that $\alpha\inv(s)\in\sfp{\Cs_\eta}$ for every $s \in M$. We prove that for every  language $L\in\sfp{\Cs_\eta}$, there exists an integer $k \geq 1$ such that the following property holds (recall that $\equiv_L$ denotes the syntactic congruence of~$L$):
  \begin{equation}\label{eq:sfclos:aperp}
    \text{for every } u\in A^*\text{, if }\eta(u)\text{ is idempotent, then } u^{k+1}\equiv_L u^{k}.
  \end{equation}
  Let us first explain why this implies the statement of the proposition: $(\alpha(u))^{\omega+1} = (\alpha(u))^{\omega}$ for any $u \in A^*$ such that $\eta(u)$ is idempotent. Let $u$ be such a word and let $s = (\alpha(u))^{\omega}$. Since $\alpha\inv(s)\in\sfp{\Cs_\eta}$, there exists $k \geq 1$ such that~\eqref{eq:sfclos:aperp} holds for $L = \alpha\inv(s)$. Let $p = \omega(M)$. Note that $\alpha(u^{pk}) = s$. This yields $u^{(p-1)k}\cdot u^{k}\cdot\veps = u^{pk} \in \alpha\inv(s)$, whence by~\eqref{eq:sfclos:aperp}, $u^{pk+1} = u^{(p-1)k}\cdot u^{k+1}\cdot\veps \in \alpha\inv(s)$. We get $(\alpha(u))^{\omega+1}=\alpha(u^{pk+1})=s=(\alpha(u))^{\omega}$, as~desired.

  If remains to prove that for every language $L\in\sfp{\Cs_{\eta}}$, there exists $k \geq 1$ such that~\eqref{eq:sfclos:aperp} holds. We argue by induction on the construction of $L$. The base cases are when $L=\eta^{-1}(t)$ for $t\in N$ and when $L$ is a singleton $\{a\}$. If $L=\eta^{-1}(t)$ for $t\in N$, then,~\eqref{eq:sfclos:aperp} holds for~$k = 1$. Indeed, given $u,x,y \in A^*$ such that $\eta(u)$ is idempotent, we have $\eta(xuuy) = \eta(xuy)$. Since $L$ is recognized by $\eta$, this yields $xuu y \in L \Leftrightarrow xuy \in L$, \emph{i.e.}, $u^{2}\equiv_{L}u$. Assume next that $L=\{a\}$ for $a\in A$. In this case, \eqref{eq:sfclos:aperp} holds for $k = 2$. Indeed, let $u,x,y \in A^*$ such that $\eta(u)$ is idempotent. If $u = \veps$, then $xu^{2+1}y = xy = xu^2 y$ hence we get $xu^{2+1} y \in L \Leftrightarrow xu^{2}y \in L$. Otherwise, $u \in A^+$ and we have $|xu^2y| > 1$ and $|xu^{2+1}y| > 1$. Therefore, since $L = \{a\}$, we have $xu^2y \not\in L$ and $xu^{2+1}y \not\in L$. In all cases, $u^{3}\equiv_{L}u^{2}$, concluding the proof of~\eqref{eq:sfclos:aperp} in the base cases.

  We turn to the inductive cases. Assume first that the last operation used to build~$L$ is union. We have $L = L_1 \cup L_2$ where $L_1,L_2$ are simpler languages of $\sfp{\Cs_{\eta}}$. By induction, this yields $k_1,k_2 \geq1$ such that for $i = 1,2$, if $u \in A^*$ is such that $\eta(u)$ is idempotent, we have $u^{k_i+1}\equiv_{L_{i}}u^{k_i}$. Hence,~\eqref{eq:sfclos:aperp} holds for $L$ with $k=\text{max}(k_1,k_2)$. We turn to complement. Assume that $L = A^* \setminus H$ where $H$ is a simpler language of $\sfp{\Cs_{\eta}}$. By induction, we get $k\geq1$ such that if $u \in A^*$ is such that $\eta(u)$ is idempotent, we have $u^{k+1}\equiv_{H}u^{k}$, \emph{i.e.}, $xu^{k} y \in H \Leftrightarrow xu^{k+1}y\in H$ for all $x,y\in A^{*}$. Since $L = A^* \setminus H$, the contrapositive states that $xu^{k} y \in L \Leftrightarrow xu^{k+1}y\in L$, and~\eqref{eq:sfclos:aperp} holds for $L$ with the same integer $k$ as for $H$.

  Finally, assume that the last operation used to construct $L$ is concatenation. We have $L = L_1 L_2$ where $L_1,L_2$ are simpler languages of $\sfp{\Cs_{\eta}}$. By induction, this yields $k_1,k_2 \geq1$ such that for $i = 1,2$, if $u \in A^*$ is such that $\eta(u)$ is idempotent, we have $u^{k_i+1}\equiv_{L_{i}}u^{k_i}$. Let~$m=\Max(k_1,k_2)$. We prove that~\eqref{eq:sfclos:aperp} holds for $k = 2m + 1$. Let $u,x,y \in A^*$ with $\eta(u)$ idempotent. We have to show that $xu^{k+1} y \in L \Leftrightarrow xu^ky\in L$. We concentrate on the left to right implication (the converse one is symmetrical): assuming that $xu^{k+1}y \in L$, we show that $xu^k y \in L$. Since $L = L_1L_2$, we get $w_1 \in L_1$ and $w_2 \in L_2$ such that $xu^{k+1}y = w_1 w_2$. Since $k \geq 2m + 1$, it follows that either $xu^{m+1}$ is a prefix of $w_1$ or $u^{m+1}y$ is a suffix of~$w_2$. By symmetry, we assume that the former property holds: we have $w_1 = xu^{m+1}z$ for some $z \in A^*$. Observe that since $xu^{k+1}y = w_1 w_2$, it follows that $zw_2 =u^{k-m}y$. Moreover, we have $m \geq k_1$ by definition of $m$. Since $xu^{m+1} z = w_1 \in L_1$, we know therefore that $xu^{m}z \in L_1$ by definition of $k_1$. Thus, $xu^{m}zw_2 \in L_1L_2 = L$. Since $zw_2 =u^{k-m}y$, this yields $xu^{k}y \in L$, concluding the~proof.
\end{proof}

\subsection{Group languages}

We now present a central kind of class. As we explained in the introduction, all classes investigated in the paper are built from basic ones using the star-free closure operator. Here, we introduce the basic classes used in this construction: the \emph{classes of group languages}.

A \emph{group} is a monoid $G$ such that every element $g \in G$ has an inverse $g\inv \in G$, \emph{i.e.}, $gg\inv = g\inv g = 1_G$. A language $L$ is a \emph{group language} if it is recognized by a morphism $\alpha: A^* \to G$ into a \emph{finite group $G$} (\emph{i.e.}, there exists $F \subseteq G$ such that $L = \alpha\inv(F)$). We write \grp for the class of \emph{all} group languages. One can verify that \grp is a \vari.

\begin{rem}
  No language theoretic definition of\/ \grp is known. There is however a definition based on automata: the group languages are those recognized by a permutation automaton~\cite{permauto} (\emph{i.e.}, which is simultaneously deterministic, co-deterministic and complete).
\end{rem}

A \emph{class of group languages} is a class consisting of group languages only, \emph{i.e.}, a subclass of \grp. The results of this paper apply to arbitrary \emph{\varis of group languages}.

While our results apply in a generic way to all \varis of group languages, there are \emph{four} main classes of this kind that we shall use for providing examples. One of them is \grp itself. Let us present the other three. First, we write $\stzer = \{\emptyset,A^*\}$, which is clearly a \vari of group languages (the notation from the fact that this class is the base level of the \emph{Straubing-Th\'erien} hierarchy~\cite{StrauConcat,TheConcat}). While trivial, we shall see that this class has important applications. Moreover, we look at the class \md of \emph{modulo languages}. For every $q,r \in \nat$ with $r < q$, we write $L_{q,r} = \{w \in A^* \mid |w| \equiv r \bmod q\}$. The class \md consists of all \emph{finite unions} of languages~$L_{q,r}$. One may verify that \md is a \vari of group languages. Finally, we shall consider~the class \abg of \emph{alphabet modulo testable languages}. For all $q,r \in \nat$ with $r < q$ and all $a \in A$, let $L^a_{q,r} = \{w \in A^* \mid |w|_a \equiv r \bmod q\}$. We define \abg as the least class consisting of all languages $L^a_{q,r}$ and closed under union and intersection. It is again straightforward to verify that \abg is a \vari of group languages.

We do not investigate classes of group languages themselves in the paper: we only use them as input classes for our operators.  In particular, we shall use \stzer, \md, \abg and \grp in order to illustrate our results. In this context, it will be important that \emph{separation} is decidable for these four classes. The techniques involved for proving this are independent from what we do in the paper. Actually, this can be difficult. On one hand, the decidability of \stzer-separation is immediate (two languages are \stzer-separable if and only if one of them is empty). On the other hand, the decidability of \grp-separation is equivalent to a difficult algebraic question~\cite{henckell:hal-00019815}, which remained open for several years before it was solved by Ash~\cite{Ash91}. Recent automata-based proofs that separation is decidable for \md, \abg and \grp are available in~\cite{pzgroups}.

We conclude this section with a useful result, which states a simple property of the \Gs-morphisms when \Gs is a group \vari.

\begin{lem} \label{lem:gmorph}
  Let \Gs be a group \vari and let $\alpha: A^* \to G$ be a \Gs-morphism. Then, $G$ is a group.
\end{lem}

\begin{proof}
  Let $g \in G$, we exhibit an inverse $g \inv \in G$ for $g$ (\emph{i.e}, such that $gg\inv = g\inv g = 1_G$). By hypothesis, $\alpha\inv(1_G) \in \Gs$. Since \Gs is a group \vari, there exists a morphism $\eta: A^* \to H$ into a finite group $H$ recognizing $\alpha\inv(1_G)$. Let $n = \omega(H)$. We define $g\inv = g^{n-1}$. Since $gg\inv = g\inv g = g^n$, it remains to prove that $g^n = 1_G$. Let $w \in \alpha\inv(g)$ (recall that \Gs-morphisms are surjective). Clearly, $\eta(w^n) = (\eta(w))^n$ is an idempotent of $H$ since $n = \omega(H)$. Hence, $\eta(w^n) = 1_H$ since $H$ is a group. Hence, $\eta(w^n) = \eta(\veps)$. Since $\veps \in \alpha\inv(1_G)$ and since $\alpha\inv(1_G)$ is recognized by $\eta$, we get $w^n \in \alpha\inv(1_G)$, \emph{i.e.}, $\alpha(w^n) = g^{n}=1_G$, as desired.
\end{proof}


\section{Bounded synchronization delay}
\label{sec:bsd}
We now present an alternate definition of star-free closure. More precisely, we introduce a second operator $\Cs \mapsto\bsdp{\Cs}$ whose definition is independent from that of star-free closure. We then prove that $\bsdp{\Cs}=\sfp{\Cs}$ if \Cs~is a \vari. This definition is less prominent than the main one and than the logical characterizations that we shall present below.  Yet, it is a key ingredient of the~paper. Whenever we have to construct languages in \sfp{\Cs} in proof arguments, we actually build them as languages of \bsdp{\Cs}.  For example, this is how we obtain the algebraic characterization of~\sfp{\Cs} (in fact, this argument is intertwined with the proof of the inclusion $\sfp{\Cs} \subseteq \bsdp{\Cs}$ that we present in this section).

This second definition was discovered by Schützenberger~\cite{schutzbd}. He defined a single class \bsd (in our terminology, this is the class \bsdp{\stzer}) and he proved that it coincides with~the class \sfr of star-free languages (see also the work of Diekert and Kufleitner~\cite{dksd} for a recent proof). This is a surprising result since \bsd seems antithetic to \sfr at first~glance. Its definition is based on the operations available in classical regular expressions: union, concatenation \emph{and} Kleene star. However, these operations are restricted to languages satisfying specific \emph{semantic conditions}. The main restriction concerns the Kleene star, which can only be applied to \emph{prefix codes of bounded synchronization delay} (this is a notion from code theory, which we recall below). Here, we generalize the definition of \bsd as an operator $\Cs \mapsto \bsdp{\Cs}$.

We first present preliminary notions from code theory that we shall need for the definition. Then, we define $\Cs \mapsto \bsdp{\Cs}$ properly and state the correspondence with star-free closure.

\subsection{Prefix codes of bounded synchronization delay}

The objects introduced in this section are based on a notion taken from code theory: \emph{prefix codes}. We briefly present them here and prove a few basic properties that we shall need. For a detailed presentation of code theory, we refer the reader to the book of Berstel, Perrin and Reutenauer~\cite{berstel_perrin_reutenauer_2009}.

\medskip
\noindent
{\bf Prefix codes.} A language $K \subseteq A^*$ is a \emph{prefix code} when $\veps \not\in K$ (\emph{i.e.}, $K \subseteq A^+$) and $K \cap KA^+ = \emptyset$ (\emph{i.e.}, no word in $K$ admits a strict prefix which is also a word in $K$).

\begin{exa}\label{ex:sfclos:prefix}
  If $A = \{a,b\}$, then the language $A$ is a prefix code. Any singleton language~$\{u\}$ with $u\neq\varepsilon$ is also a prefix code. Finally, $a^*b$ is a prefix code as well. On the other hand, $L=\{a,aa\}$ is not a prefix code, since $aa\in L\cap LA^+$.
\end{exa}

We now state the key property of prefix codes, which we verify directly using the definition.

\begin{fact}\label{fct:prefdec}
  Let $K$ be a prefix code. Consider $m,n \in \nat$, $u_1,\dots,u_m \in K$ and $v_1 \cdots,v_n \in K$. The two following properties hold:
  \begin{itemize}
	\item If $u_1 \cdots u_m$ is a prefix of $v_1 \cdots v_n$, then $m \leq n$ and $u_i = v_i$ for every $i \leq m$.
	\item If $u_1 \cdots u_m = v_1 \cdots v_n$, then $m = n$ and $u_i = v_i$ for every $i \leq m$.
  \end{itemize}
\end{fact}

\begin{proof}
  The second property is an immediate corollary of the first one. Hence, it suffices to show that if $u_1 \cdots u_m$ is a prefix of $v_1 \cdots v_n$, then $m \leq n$ and $u_i = v_i$ for every $i \leq m$. We proceed by induction on $m \in \nat$. If $m= 0$, then the property is immediate. Otherwise, $m \geq 1$. Clearly, $u_1 \cdots u_{m-1}$ is a prefix of $v_1 \cdots v_n$. Hence, induction yields that $m-1 \leq n$ and $u_i = v_i$ for every $i \leq m -1$. Now, since $u_1 \cdots u_m$ is a prefix of $v_1 \cdots v_n$, it follows that $u_m$ is a prefix of $v_m \cdots v_n$. Since $u_m \in K$ and $K$ is a prefix code, we have $u_m \neq \veps$, which implies that $v_m \cdots v_n \neq \veps$, \emph{i.e.}, $m \leq n$. Moreover, since $K$ is a prefix code and $u_m,v_m \in K$ we know that $u_m$ is \emph{not} a strict prefix of $v_m$ and $v_m$ is \emph{not} a strict prefix of $u_m$. Together with the hypothesis that $u_m$ is a prefix of $v_m \cdots v_n$, this yields $u_m = v_m$, concluding the~proof.
\end{proof}

The second assertion in Fact~\ref{fct:prefdec} implies that when $K$ is a prefix code, every word $w \in K^*$ admits a \emph{unique} decomposition witnessing this membership. This property is exactly the definition of a \emph{code}, which is therefore a notion more general than that of prefix code.

\medskip
\noindent
{\bf Bounded synchronization delay.} We turn to a more restrictive notion. Consider an integer $d \geq 1$. We say that a prefix code $K \subseteq A^+$ has \emph{synchronization delay $d$} when the following property holds:
\begin{equation}\label{eq:sfclos:sync}
  \text{for every $u,v,w \in A^*$,} \quad uvw \in K^+ \text{ and } v \in K^d \quad \Rightarrow \quad \text{$uv \in K^+$}.
\end{equation}
Furthermore, we say that a prefix code $K \subseteq A^+$ has \emph{bounded synchronization delay} when there exists some $d \geq1$ such that $K$ has synchronization delay $d$.

\begin{rem}
  It follows from the definition that if a prefix code has synchronization delay~$d$, then it has also synchronization delay $d'$ for all $d'\geq d$.
\end{rem}

\begin{rem}
  If $K$ is a prefix code with synchronization delay~$d$, then whenever $u,v,w$ are words such that $uvw\in K^+$ and $v\in K^d$, we have $w\in K^*$. Indeed, Condition~\eqref{eq:sfclos:sync} states that $uv\in K^+$. This means that there exist words $u_1,\dots,u_n,v_1,\dots,v_m$ in $K$ such that $uvw=u_1\cdots u_n$ and $uv=v_1\cdots v_m$. From Fact~\ref{fct:prefdec}, we deduce that $m\leq n$ and $w=u_{m+1}\cdots u_n$, which belongs to $K^*$. This explains the terminology: any infix $v$ in $K^d$ of a word $uvw\in K^+$ determines a decomposition $uv\cdot w$ of $uvw$ whose factors ($uv$ and $w$) both belong to $K^*$.
\end{rem}

\begin{exa}\label{ex:delay}
  Assume that $A = \{a,b\}$.
  \begin{itemize}
    \item Clearly, for any $B\subseteq A$, the language $B$ is a prefix code of synchronization delay~$1$.

    \item It is also immediate that any language $L\subseteq a^*b$ is a prefix code of synchronization delay~$1$. Indeed, if $uvw \in L^*$ and $v\in L$, then $v$ ends with a ``$b$'', whence $uv \in L^*$.

    \item Similarly, one may verify that $(aab)^*ab$ is a prefix code of synchronization delay~$2$ (this follows from Fact~\ref{fct:sprefnest} below applied to $K=\{ab,abb\}$ and $H=\{ab\}$, where $K$ is a prefix code of synchronization delay~1). However, it does not have synchronization delay $1$. Indeed, consider the decomposition $aabab=a\cdot ab\cdot ab$. We have $aabab,ab\in ((aab)^*(ab))^1$ but $aab\notin((aab)^*ab)^+$.

    \item Finally, $\{aa\}$ (which is a prefix code) does not have bounded synchronization delay. Indeed, given $d \geq 1$, we have $a(aa)^da \in (aa)^+$ and $(aa)^d \in \{aa\}^d$ but $a(aa)^d \not\in (aa)^*$.
  \end{itemize}
\end{exa}


We complete the definition with two properties of prefix codes of bounded synchronization delay. First, we show that every language included in a such a code retains the property to be a prefix code of bounded synchronization delay.

\begin{fact}\label{fct:sprefsub}
  Let $d \geq 1$ and let $K \subseteq A^+$ be a prefix code with synchronization delay $d$. Then, every language $H \subseteq K$ is also a prefix code with synchronization delay $d$.
\end{fact}

\begin{proof}
  It is immediate from the definitions that $H$ is itself a prefix code. It remains to prove that $H$ has synchronization delay $d$. Consider $u,v,w \in A^*$ such that $uvw \in H^+$ and $v \in H^d$. We show that $uv \in H^+$. Since $H \subseteq K$, we have $uvw \in K^+$ and $v \in K^d$. Since $K$ has synchronization delay $d$, we obtain $uv \in K^+$. Moreover, since $K$ is a prefix code and $uvw \in K^+$, it follows from the second property in Fact~\ref{fct:prefdec} that $uvw$ admits a unique decomposition into factors of $K$. Additionally, since $uvw \in H^+$ with $H \subseteq K$, all factors in this unique decomposition belong to $H$. Finally, since $uv\in K^+$, the first property in Fact~\ref{fct:prefdec} yields that $uv$ is a concatenation of factors in this unique decomposition. Hence, we have $uv \in H^+$, which concludes the proof.
\end{proof}

Let us now present a construction to build a new prefix code of bounded synchronization delay from another one.

\begin{fact}\label{fct:sprefnest}
  Let $d \geq 1$ and let $K \subseteq A^+$ be a prefix code with synchronization delay $d$. Let $H \subseteq K$. Then, the language $(K \setminus H)^*H$ is a prefix code with synchronization delay $d+1$.
\end{fact}

\begin{proof}
  We first verify that $(K \setminus H)^*H$ is a prefix code. Clearly, $(K \setminus H)^*H \subseteq A^+$ since $H \subseteq K \subseteq A^+$. Hence, we have to show that $(K \setminus H)^*H \cap (K \setminus H)^*HA^+ = \emptyset$. Assume by contradiction that there exists $w\in(K \setminus H)^*H \cap (K \setminus H)^*HA^+ $. In particular, we have $w \in (K \setminus H)^*H\subseteq K^*$. Since~$K$ is a prefix code, $w$ admits a \emph{unique} decomposition $w = w_1 \cdots w_n$ with $w_1,\dots,w_n \in K$. Since $w \in (K \setminus H)^*H$, the factor $w_n$ is the only one to be in~$H$ among all the~$w_i$'s. However, since $w \in (K \setminus H)^*HA^+$, the first property in Fact~\ref{fct:prefdec} implies that one of the factors $w_i$ for $i \leq n-1$ must belong to $H$. This is a contradiction. Therefore, $(K \setminus H)^*H$ is a prefix code.

  It remains to show that $(K \setminus H)^*H$ has synchronization delay $d+1$. Let $u,v,w \in A^*$ such that $uvw \in ((K \setminus H)^*H)^+$ and $v \in ((K \setminus H)^*H)^{d+1}$. We prove that $uv \in ((K \setminus H)^*H)^+$. Clearly $v = xy$ with $x \in ((K \setminus H)^*H)^{d}$ and $y \in (K \setminus H)^*H$. Observe that $x \in K^{n}$ for some $n \geq d$.  Hence, since $uxyw = uvw \in K^+$ and $K$ has synchronization delay $d$, it follows that $ux \in K^+$. Consequently $uv = uxy \in K^+(K \setminus H)^*H$, whence $uv \in K^*H\subseteq ((K \setminus H)^*H)^+$. This concludes the proof.
\end{proof}

\subsection{Definition}

We now define the operator $\Cs \mapsto \bsdp{\Cs}$. The definition involves two additional notions. First, we consider \emph{disjoint union}. Two languages $K,L \subseteq A^*$ are \emph{disjoint} if $K \cap L =\emptyset$. In this case, we write $K \uplus L$ for $K \cup L$ in order to emphasize disjointedness. Additionally, we consider \emph{unambiguous concatenation}. Given two languages $K,L \subseteq A^*$, their concatenation $KL$ is \emph{unambiguous} when every word $w \in KL$ admits a \emph{unique} decomposition witnessing this membership: if $u,u' \in K$, $v,v' \in L$ and $uv = u'v'$, then $u = u'$ and $v = v'$.


Let \Cs be some class of languages. We write \bsdp{\Cs} for the least class containing $\emptyset$ and $\{a\}$ for every $a \in A$, and which is closed under the following properties:
\begin{itemize}
  \item \textbf{Intersection with \Cs:} if $K \in \bsdp{\Cs}$ and $L \in \Cs$, then $K \cap L \in \bsdp{\Cs}$.
  \item \textbf{Disjoint union:} if $K,L \in \bsdp{\Cs}$ are disjoint then $K \uplus L \in \bsdp{\Cs}$.
  \item \textbf{Unambiguous concatenation:} if $K,L \in \bsdp{\Cs}$ and $KL$ is unambiguous, then $KL \in \bsdp{\Cs}$.
  \item \textbf{Kleene star for prefix codes of bounded synchronization delay:} if $K \in \bsdp{\Cs}$ is a prefix code of bounded synchronization delay, then $K^* \in \bsdp{\Cs}$.
\end{itemize}

An important special case is when the input class \Cs is the trivial \vari $\stzer = \{\emptyset,A^*\}$. In this case, \bsdp{\stzer} is the original class \bsd of Schützenberger~\cite{schutzbd}. His definition is slightly different, as it does not require unions to be disjoint, nor concatenations to be unambiguous. Yet, the two definitions are equivalent.

\begin{exa}\label{ex:bsd}
  Let us present two examples. Let $A = \{a,b\}$.
  \begin{itemize}
    \item We have $(ab)^* \in \bsdp{\stzer}$. Indeed, $\{a\},\{b\} \in \bsdp{\stzer}$ which implies that $\{ab\}\in \bsdp{\stzer}$ by closure under unambiguous concatenation. Since $\{ab\}$ is a prefix code of bounded synchronization delay (the delay is 1), we get $(ab)^* \in \bsdp{\stzer}$.
    \item We have $(aa+bb)^* \in \bsdp{\md}$ (on the other hand,  $(aa+bb)^* \not\in \bsdp{\stzer}$, which can be verified using Theorem~\ref{thm:sfcarac}). Clearly, $\{a\}$ and $\{b\}$ are prefix codes of bounded synchronization delay. Hence, $a^*,b^* \in \bsdp{\md}$. Moreover, since $(AA)^* \in \md$, we get $(aa)^*,(bb)^* \in \bsdp{\md}$ by closure under intersection with~\md. We then use unambiguous concatenation to get $(aa)^+(bb)^+ = aa(aa)^*bb(bb)^* \in \bsdp{\md}$. This is a prefix code of bounded synchronization delay. Hence, $((aa)^+(bb)^+)^* \in \bsdp{\md}$. Using unambiguous concatenation again, this yields $(bb)^*((aa)^+(bb)^+)^*(aa)^* \in \bsdp{\md}$. One may now verify that $(aa+bb)^* =  (bb)^*((aa)^+(bb)^+)^*(aa)^* \in \bsdp{\md}$.

  \end{itemize}
\end{exa}

\begin{rem}\label{rem:intersections}
  We do not explicitly require in the definition that \bsdp{\Cs} contains \Cs. Yet, this is a simple consequence of the definition. Clearly, $A^* \in \bsdp{\Cs}$ since $A \in \bsdp{\Cs}$ is a prefix code of bounded synchronization delay. Hence, $L = A^* \cap L \in \bsdp{\Cs}$ for every $L \in \Cs$.

  On the other hand, it is crucial to allow intersection with languages in \Cs. If we only require the inclusion $\Cs \subseteq \bsdp{\Cs}$ in the definition, we would end up with a weaker operator (which, therefore, does not correspond to star-free closure in general). For example, consider the class \md of modulo languages. Assume that $A = \{a,b\}$. As observed in Example~\ref{ex:bsd}, $(aa)^* \in \bsdp{\md}$.  On the other hand, one may verify that $(aa)^*$ cannot be built from the languages of\/ \md using only  union, concatenation and Kleene star for prefix codes of bounded synchronization delay.
\end{rem}

It is not immediate that the classes \bsdp{\Cs} have robust closure properties, even when this is the case for the input class \Cs. Actually, it is not even clear whether \bsdp{\Cs} is a lattice, since closure under intersection is not required in the definition, and closure under union is restricted. However, \bsdp{\Cs} \emph{does} have robust properties: if \Cs is a \vari, then \bsdp{\Cs} is a \vari closed under concatenation. This follows from Proposition~\ref{prop:vari} and the following theorem, which states the correspondence with star-free closure.

\begin{thm}\label{thm:bsdcar}
  Let \Cs be a \vari. Then, $\bsdp{\Cs} = \sfp{\Cs}$.
\end{thm}

The difficult direction in Theorem~\ref{thm:bsdcar} is the inclusion $\sfp{\Cs} \subseteq \bsdp{\Cs}$. We rely on an indirect approach based on the generic algebraic characterization of star-free closure, which we use as an intermediary result to prove this implication. In fact, the proof of the difficult inclusion is intertwined with the one of~the characterization itself. Hence, we postpone it to the next section. On the other hand, we prove the easier inclusion $\bsdp{\Cs} \subseteq \sfp{\Cs}$ now.

\begin{proof}[Inclusion $\bsdp{\Cs} \subseteq \sfp{\Cs}$ in Theorem~\ref{thm:bsdcar}]
  We fix a \vari \Cs and prove the inclusion $\bsdp{\Cs} \subseteq \sfp{\Cs}$. This amounts to proving that \sfp{\Cs} satisfies all properties in the definition of \bsdp{\Cs}. In all cases but one, this is immediate by definition of \sfp{\Cs}. Indeed, we have $\emptyset\in\sfp\Cs$ and $\{a\} \in \sfp{\Cs}$ for every $a \in A$. Moreover, \sfp{\Cs} is closed under union, intersection and concatenation by definition (this includes intersection with languages of \Cs since $\Cs \subseteq \sfp{\Cs}$). It remains to show that $\sfp{\Cs}$ is closed under Kleene star applied to a prefix code of bounded synchronization delay.

  We let $K \in \sfp{\Cs}$ be such a prefix code, and we let $d\geq 1$ be its synchronization delay. We have to show that $K^* \in \sfp{\Cs}$.
  Consider the following languages:
  \begin{align*}
	H &= \biggl(A^*K^d \cap \Bigl(A^* \setminus \bigl(A^*K^{d+1} \cup \bigcup_{0 \leq h \leq d} K^h\bigr)\Bigr) \biggr) A^*.\\
	G &= \Bigl(\bigcup_{0 \leq h \leq d-1} K^h\Bigr) \cup \Bigl(A^*K^d \cap (A^* \setminus H)\Bigr).
  \end{align*}

  Clearly, $H \in \sfp{\Cs}$ since both $K$ and $A^*$ belong to \sfp{\Cs}, which is closed under Boolean operations and concatenation. Therefore $G\in\sfp\Cs$ as well. We show that $K^* = G$, which will entail that $K^*\in\sfp{\Cs}$, concluding the proof of $\bsdp{\Cs} \subseteq \sfp{\Cs}$ in Theorem~\ref{thm:bsdcar}.

  We first show that $K^*\cap H=\emptyset$. Note that this is the only part of the proof  where we use the hypothesis that $K$ has synchronization delay $d$.

  \begin{fact}\label{fct:sfclos:delay}
	We have $K^* \subseteq A^* \setminus H$.
  \end{fact}

  \begin{proof}
	We have to show that $K^*\cap H=\emptyset$. Since $K^* \subseteq A^*K^{d+1} \cup \bigcup_{0 \leq h \leq d} K^h$, we have,
	\[
	  A^* \setminus \bigl(A^*K^{d+1} \cup \bigcup_{0 \leq h \leq d} K^h\bigr) \subseteq A^* \setminus K^*.
	\]
	By definition of $H$, this yields $H \subseteq \left(A^*K^d \cap \left(A^* \setminus K^*\right)\right) A^*$. Therefore, it suffices to show that $K^*\cap \left(A^*K^d \cap \left(A^* \setminus K^*\right)\right) A^*=\emptyset$, which follows immediately from the hypothesis that $K$ has synchronization delay~$d$.
  \end{proof}

  It remains to show that $K^*=G$. We start with the left to right inclusion. Recall that $G = (\bigcup_{0 \leq h \leq d-1} K^h) \cup \left(A^*K^d \cap \left(A^* \setminus H\right)\right)$. Consider $x \in K^*$. If $x \in K^h$ for $h \leq d-1$, it is immediate that $x \in G$. Otherwise, we have $x \in A^*K^d$ and since $x \in K^*$, we know that $x \in A^* \setminus H$ by Fact~\ref{fct:sfclos:delay}.
  This implies that $x \in A^*K^d \cap \left(A^* \setminus H\right)\subseteq G$, finishing the proof for this inclusion.

  \medskip

  For  the right to left inclusion, consider $x \in G$. We show that $x \in K^*$. If $x\in \bigcup_{0 \leq h \leq d-1} K^h$, this is immediate. Otherwise, $x\in A^*K^d \cap \left(A^* \setminus H\right)$. We proceed by induction on the length of~$x$. By hypothesis, $x \in A^*K^d$ and $x \not\in H$. By definition of $H$, this implies that,
  \[
    x \in A^*K^{d+1} \cup \bigcup_{0 \leq h \leq d} K^h.
  \]
  If $x \in \bigcup_{0 \leq h \leq d} K^h$, it is immediate that $x \in K^*$, which finishes the proof. Otherwise, $x \in A^*K^{d+1}$ which means that $x = x'y$ with $x' \in A^*K^d$ and $y \in K$. Since $\varepsilon \not\in K$ (as $K$ is a prefix code), we have $y \neq \veps$ which implies that $|x'| < |x|$. Moreover, since $x \in A^* \setminus H$ and $x'$ is a prefix of $x$, one may verify from the definition of $H$ that $x' \in A^* \setminus H$ as well. Altogether, we have $x' \in A^*K^d \cap \left(A^* \setminus H\right)$ and $|x'| < |x|$. Therefore by induction, $x' \in K^*$. Finally, since $y \in K$, we get $x = x'y \in K^*K \subseteq K^*$, which concludes the proof.
\end{proof}

\begin{rem}
  In the above proof, we used Fact~\ref{fct:sfclos:delay} only to establish that $K^*\subseteq G$. This inclusion relies on the assumption that $K$ has bounded synchronization delay (clearly, relying on this hypothesis to show that $K^*$ belongs to $\sfp{\Cs}$ is mandatory). On the other hand, the inclusion $G\subseteq K^*$ is independent from this hypothesis.
\end{rem}


\section{Algebraic characterization}
\label{sec:carac}
We present a generic algebraic characterization of the classes \sfp{\Cs} built with star-free closure from a \vari \Cs. It yields an effective reduction from \sfp{\Cs}-membership to \Cs-separation (here, we~mean reduction in the Turing sense: we get a generic algorithm for \sfp{\Cs}-membership that uses an oracle for \Cs-separation). Moreover, we use this characterization to prove the missing inclusion $\sfp{\Cs} \subseteq \bsdp{\Cs}$ in Theorem~\ref{thm:bsdcar}.

We characterize the languages in \sfp{\Cs} by a property of their syntactic morphisms. It generalizes Schützenberger's characterization of star-free languages as those whose syntactic monoid is \emph{aperiodic}~\cite{schutzsf}. First, with every class \Cs and every morphism $\alpha: A^*\to M$, we associate a relation on~$M$: the \emph{\Cs-pair relation for $\alpha$} (it was defined in~\cite{PZ:generic19}). Then, we use this relation to identify special subsets of $M$, which happen to be monoids when \Cs is a \vari: \emph{the \Cs-orbits of $\alpha$}. Finally, the characterization states that for every \vari \Cs, a language belongs to \sfp{\Cs} if and only if all \Cs-orbits of its syntactic morphism are aperiodic monoids. Let us now define \Cs-pairs.

\subsection{\Cs-Pairs}

Consider a class \Cs and a morphism $\alpha: A^* \to M$ into a finite monoid. We define the \Cs-pair relation for $\alpha$ on $M$ as follows. Let $(s,t) \in M^2$. We say that,
\begin{equation}\label{eq:cpairs}
  \text{$(s,t)$ is a \emph{\Cs-pair} (for $\alpha$) if and only if $\alpha\inv(s)$ is \emph{not} \Cs-separable from $\alpha\inv(t)$}.
\end{equation}

\begin{rem}
  While we often make this implicit, being a \Cs-pair depends on $\alpha$.
\end{rem}

By definition, the set of \Cs-pairs for $\alpha$ is finite: it is a subset of $M^2$. Moreover, having a \Cs-separation algorithm in hand is clearly enough to compute all \Cs-pairs associated to an input morphism $\alpha$. We complete the definition with some properties of \Cs-pairs. A simple and useful one is that the \Cs-pair relation is reflexive when the morphism $\alpha$ is surjective (which is always the case in practice). Moreover, it is symmetric when~\Cs is closed under complement (this is the case for all classes considered in the paper). On the other hand, the \Cs-pair relation is \emph{not} transitive in general, as the following example shows.

\begin{exa}\label{exa:cpairnottrans}
   Let $A=\{a,b\}$ and $\Cs=\at$ be the least Boolean algebra containing $A^*aA^*$ and $A^*bA^*$. Let~$M$ be the monoid $\{1,a,b,0\}$ where $1$ acts as an identity element, $0$ as an absorbing element, and the rest of the multiplication is given by $aa=ab=ba=bb=0$. Let $\alpha:A^*\to M$ be the morphism defined by $\alpha(a)=a$ and $\alpha(b)=b$. We have $\alpha\inv(a)=\{a\}$. Therefore, one may verify that any language of \at containing $\alpha\inv(a)$ also contains $a^{+}$, and therefore intersects $\alpha\inv(0)$ (which is the set of words of length at least~2). Hence, $(a,0)$ is an \at-pair. Likewise, $(0,b)$ is an \at-pair. However, $(a,b)$ is not an \at-pair, since the language $a^+\in\at$ separates $\alpha\inv(a)=\{a\}$ from $\alpha\inv(b)=\{b\}$. This example shows that the \Cs-pair relation is not transitive in general.
\end{exa}

We now provide a useful characterization of \Cs-pairs via \Cs-morphisms in the special case when~\Cs is a \vari, which, again, is the only case that we consider here.

\begin{lem}\label{lem:cmorph}
  Let \Cs be a \vari and let $\alpha: A^* \to M$ be a morphism into a finite monoid. The two following properties hold:
  \begin{enumerate}
    \item\label{item:cmoprh1} For every \Cs-morphism $\eta: A^* \to N$ and every \Cs-pair $(s,t) \in M^2$ for $\alpha$, there exist $u,v \in A^*$ such that $\eta(u) = \eta(v)$, $\alpha(u) = s$ and $\alpha(v) = t$.
    \item\label{item:cmoprh2} There exists a \Cs-morphism $\eta: A^* \to N$ such that for all $u,v \in A^*$, if $\eta(u) = \eta(v)$, then $(\alpha(u),\alpha(v))$ is a \Cs-pair for $\alpha$.
  \end{enumerate}
\end{lem}

\begin{proof}
  Let us start with the first assertion. Let $\eta: A^* \to N$ be a \Cs-morphism and let $(s,t) \in M^2$ be a \Cs-pair for $\alpha$. Let $F = \eta(\alpha\inv(s) )\subseteq N$. We have $\eta\inv(F) \in \Cs$ since $\eta$ is a \Cs-morphism. Moreover, it is immediate from the definition of $F$ that $\alpha\inv(s) \subseteq \eta\inv(F)$. Since $(s,t)$ is a \Cs-pair (meaning that $\alpha\inv(s)$ cannot be separated from $\alpha\inv(t)$ by a language in \Cs), it follows that $\eta\inv(F) \cap \alpha\inv(t) \neq \emptyset$. This yields $v \in A^*$ such that $\eta(v) \in F$ and $\alpha(v) = t$. Finally, since $\eta(v)\in F=\eta(\alpha\inv(s))$, we get $u \in A^*$ such that $\eta(u) = \eta(v)$ and $\alpha(u) = s$. This concludes the proof of the first~assertion.

  Let us turn to the second assertion. Let $P \subseteq M^2$ be the set of all pairs $(s,t) \in M^2$ which are \emph{not} \Cs-pairs. For every $(s,t) \in P$, there exists $K_{s,t} \in \Cs$ separating $\alpha\inv(s)$ from $\alpha\inv(t)$. Proposition~\ref{prop:genocm} yields a \Cs-morphism $\eta: A^* \to N$ such that every language $K_{s,t}$ for $(s,t) \in P$ is recognized by $\eta$. It remains to prove that for every $u,v \in A^*$, if $\eta(u) = \eta(v)$, then $(\alpha(u),\alpha(v))$ is a \Cs-pair. We prove the contrapositive. Assuming that $(\alpha(u),\alpha(v))$ is a \emph{not} a \Cs-pair, we show that $\eta(u) \neq \eta(v)$. By hypothesis, $(\alpha(u),\alpha(v)) \in P$, which means that $K_{\alpha(u),\alpha(v)} \in \Cs$ is defined and separates $\alpha\inv(\alpha(u))$ from $\alpha\inv(\alpha(v))$. In particular, $u \in K_{\alpha(u),\alpha(v)}$ and $v \not\in K_{\alpha(u),\alpha(v)}$. Since $K_{\alpha(u),\alpha(v)}$ is recognized by $\eta$, this implies that $\eta(u) \neq \eta(v)$.
\end{proof}

Finally, we prove that when \Cs is a \vari of regular languages, the \Cs-pair relation is compatible with multiplication.

\begin{lem}\label{lem:mult}
  Let \Cs be a \vari and let $\alpha: A^* \to M$ be a morphism into a finite monoid. If $(s_1,t_1),(s_2,t_2) \in M^2$ are \Cs-pairs, then $(s_1s_2,t_1t_2)$ is a \Cs-pair as well.
\end{lem}

\begin{proof}
  Item~\ref{item:cmoprh2} of Lemma~\ref{lem:cmorph} yields a \Cs-morphism $\eta: A^* \to N$ such that for all $u,v \in A^*$, if $\eta(u) = \eta(v)$, then $(\alpha(u),\alpha(v))$ is a \Cs-pair. Let $(s_1,t_1),(s_2,t_2) \in M^2$ be \Cs-pairs. Since $\eta$ is a \Cs-morphism, it follows from Item~\ref{item:cmoprh1} of Lemma~\ref{lem:cmorph} that there exist $u_i,v_i \in A^*$ for $i = 1,2$ such that $\eta(u_i) = \eta(v_i)$, $\alpha(u_i) = s_i$ and $\alpha(v_i) = t_i$. This yields $\eta(u_1u_2) = \eta(v_1v_2)$, $\alpha(u_1u_2) = s_1s_2$ and $\alpha(v_1v_2) = t_1t_2$. Hence, $(s_1s_2,t_1t_2)$ is a \Cs-pair by definition of $\eta$.
\end{proof}

\subsection{\Cs-orbits and \Cs-kernels}

Consider a class \Cs and a morphism $\alpha: A^* \to M$ into a finite monoid. For every idempotent $e\in E(M)$ we define the \emph{\Cs-orbit of $e$} (for $\alpha$) as the set consisting of all elements $ese\in M$ such that $(e,s)\in M^2$ is a \Cs-pair (for $\alpha$). More generally, the \emph{\Cs-orbits for $\alpha$} are all the subsets of $M$ which are the \Cs-orbit of some idempotent $e \in E(M)$. When \Cs is a \vari, we have the following simple corollary of Lemma~\ref{lem:mult}.

\begin{lem} \label{lem:orbitmono}
  Let \Cs be a \vari and let $\alpha: A^* \to M$ be a surjective morphism into a finite monoid. For every $e \in E(M)$, the \Cs-orbit of $e$ for $\alpha$ is a subsemigroup of $M$. Moreover, it is a monoid whose neutral element is $e$.
\end{lem}

\begin{proof}
  We write $N_e \subseteq M$ for the \Cs-orbit of $e$. Observe that $N_e$ is nonempty: $e = eee \in N_e$ since $(e,e)$ is a \Cs-pair (note that here, we need the hypothesis that $\alpha$ is surjective, as it implies that $\alpha\inv(e)\neq\emptyset$). Moreover, since $e$ is idempotent, it is clear that for every $q \in N_e$, we have $eq=qe = q$ since $q = ese$ for some $s \in M$. It remains to prove that $N_e$ is a subsemigroup of $M$. Let $q,r\in N_e$. We show that $qr\in N_e$. By definition, we get $s,t \in M$ such that $(e,s)$ and $(e,t)$ are \Cs-pairs, $q = ese$ and $r = ete$. Since $(e,e)$ is also a \Cs-pair and \Cs is a \vari, Lemma~\ref{lem:mult} implies that $(e,set)$ is a \Cs-pair. Hence, we have $qr = eseete = esete \in N_e$, as~desired.
\end{proof}

By definition, the \Cs-pairs associated to a morphism $\alpha: A^*\to M$ can be computed provided that \Cs-separation is decidable. Therefore, it is immediate that for each $e\in E(M)$, the \Cs-orbit of $e$ can be computed as well in this case. 

\begin{lem}\label{lem:seporbits}
  Let \Cs be a class of languages with decidable separation. Then, given as input a morphism $\alpha: A^*\to M$ into a finite monoid and an idempotent $e \in E(M)$, one can compute the \Cs-orbit of $e$ for~$\alpha$.
\end{lem}

We complete the definition of \Cs-orbits by connecting it with another notion tailored to classes that are \emph{group \varis}~\cite{pzupol2}. Given a class \Gs, we associate  with any morphism $\alpha: A^* \to M$ (where $M$ is a finite monoid) a subset of $M$. We call this subset of $M$ the \emph{\Gs-kernel} of $\alpha$. It consists of all elements $s \in M$ such that $\{\veps\}$ is \emph{not} \Gs-separable from $\alpha\inv(s)$.

\begin{rem}
  While the definition makes sense for an arbitrary class \Gs, it is meant to be used in the special case when \Gs is a group \vari.
\end{rem}

\begin{rem} \label{rem:stablemono}
  When \Gs is the class \md of modulo languages, it can be shown that the \md-kernel of a morphism corresponds to a standard notion: the stable monoid, defined in~\cite{stablemono}. Given a morphism $\alpha: A^*\to M$ into a finite monoid, it can be verified that there exists a number $d \geq 1$ such that $\alpha(A^{2d}) = \alpha(A^d)$. The least such number $d \geq 1$ is called the stability index of $\alpha$. The stable monoid of $\alpha$ is $N = \{1_M\} \cup \alpha(A^d)$. One may verify that $N$ is the \md-kernel of $\alpha$ (this follows from a simple analysis of \md-separation).
\end{rem}

Clearly, having a \Gs-separation algorithm in hand suffices to compute the \Gs-kernel of an input morphism $\alpha$. This yields the following lemma.

\begin{lem}\label{lem:kercomp}
  Let \Gs be a class with decidable separation. Given as input a morphism $\alpha: A^*\to M$ into a finite monoid, one may compute the \Gs-kernel of $\alpha$.
\end{lem}

We now characterize \Gs-kernels in terms of \Gs-orbits when \Gs is a \emph{group} \vari.

\begin{lem}\label{lem:kerorbit}
  Let \Gs be a group \vari and let $\alpha: A^* \to M$ be a surjective morphism into a finite monoid. Let $N$ be the \Gs-kernel of $\alpha$. The two following properties hold:
  \begin{itemize}
    \item Every \Gs-orbit for $\alpha$ is a subset of $N$.
    \item $N$ is exactly the \Gs-orbit of\/ $1_M$ for $\alpha$.
  \end{itemize}
\end{lem}

\begin{proof}
  First, let $e \in E(M)$ and let $N_e$ be the \Gs-orbit of $e$. We prove that $N_e \subseteq N$. Let $q \in N_e$. This yields $s \in M$ such that $(e,s)$ is a \Gs-pair and $q = ese$. By contradiction, assume that $q \not\in N$. By definition, we get a language $K \in \Gs$ separating $\{\veps\}$ from $\alpha\inv(q)$. That is, $\veps \in K$ and $K \cap \alpha\inv(q) = \emptyset$. We exhibit an element $x\in K\cap\alpha\inv(q)$, yielding a contradiction. Since $K\in\Gs$, Proposition~\ref{prop:genocm} yields a \Gs-morphism $\eta: A^* \to G$ recognizing $K$. Moreover, Lemma~\ref{lem:gmorph} implies that $G$ is a group since \Gs is a group \vari. Since $(e,s)$ is a \Gs-pair, Lemma~\ref{lem:cmorph} yields $u,v \in A^*$ such that $\eta(u) = \eta(v)$, $\alpha(u) = e$ and $\alpha(v) = s$. Let $k = \omega(G)$ and $x = u^kvu^{k-1}$. Since $\eta(u) = \eta(v)$ and $G$ is a group, we have $\eta(x) = (\eta(u))^k = 1_G=\eta(\veps)$. Hence, since $\veps \in K$ and $K$ is recognized by $\eta$, we have $x \in K$. This a contradiction since $\alpha(x) = e^k s e^{k-1} = ese = q$ and $K \cap \alpha\inv(q) = \emptyset$ by hypothesis.

  It remains to prove that $N$ is exactly the \Gs-orbit of $1_M$ for $\alpha$. Since we already proved that the latter is included in the former, it suffices to show that every $s \in N$ belongs to the \Gs-orbit of $1_M$. Since $s \in N$, we know that $\{\veps\}$ is not \Gs-separable from $\alpha\inv(s)$. Since $\veps \in \alpha\inv(1_M)$, it follows that $\alpha\inv(1_M)$ is not \Gs-separable from $\alpha\inv(s)$, \emph{i.e.}, $(1_M,s)$ is a \Gs-pair. Hence, $s = 1_Ms1_M$ belongs to the \Gs-orbit of $1_M$, as desired.
\end{proof}

\subsection{Characterization}

Let us first recall the definition of aperiodic monoids. We use an equational definition, specific to \emph{finite} monoids. We say that a finite monoid $M$ is \emph{aperiodic} when every $s \in M$ satisfies $s^{\omega+1}=s^{\omega}$. We are ready to state the generic characterization of \sfp{\Cs}. In fact, the statement also mentions the correspondence with \bsdp{\Cs}. This is important because we still have to prove the inclusion $\sfp{\Cs}\subseteq \bsdp{\Cs}$, and the argument is intertwined with the proof of the algebraic characterization.

\begin{thm}\label{thm:sfcarac}
  Let \Cs be a \vari and consider a regular language $L \subseteq A^*$. The following properties are equivalent:
  \begin{enumerate}
    \item $L \in \sfp{\Cs}$.
    \item $L \in \bsdp{\Cs}$.
    \item All \Cs-orbits for the syntactic morphism of\/ $L$ are aperiodic monoids.
  \end{enumerate}
\end{thm}

Before we prove Theorem~\ref{thm:sfcarac}, let us discuss its consequences. First, it yields a transfer result concerning the decidability of \sfp{\Cs}-membership. It follows from Lemma~\ref{lem:seporbits} that the \Cs-orbits associated to a morphism into a finite monoid can be computed as soon as \Cs-separation is decidable. Hence, we obtain an algorithm for \sfp{\Cs}-membership in this case. Given an input language $L$, one first computes its syntactic morphism $\alpha: A^* \to M$. Then, one computes all elements $s \in M$ belonging to a \Cs-orbit for $\alpha$ (this is possible since \Cs-separation is decidable). Finally, it follows from Theorem~\ref{thm:sfcarac} and the definition of aperiodicity that $L \in \sfp{\Cs}$ if and only if every such element $s \in M$ satisfies $s^{\omega+1} = s^{\omega}$.

\begin{cor} \label{cor:sfcarac}
  Let \Cs be a \vari with decidable separation. Then, \sfp{\Cs}-membership is decidable.
\end{cor}

Theorem~\ref{thm:sfcarac} can be simplified in the special case of classes \sfp{\Gs} where \Gs is a group \vari. In this case, it is possible to characterize the languages in \sfp{\Gs} using only the \emph{\Gs-kernel} of their syntactic morphisms. Indeed, we have the following statement as an immediate corollary of Theorem~\ref{thm:sfcarac}, Lemma~\ref{lem:kerorbit} and the definition of aperiodic monoids.

\begin{cor}\label{cor:sfcaracg}
  Let \Gs be a group \vari and consider a regular language $L$. The following properties are equivalent:
  \begin{enumerate}
    \item $L \in \sfp{\Gs}$.
    \item $L \in \bsdp{\Gs}$.
    \item The \Gs-kernel of the syntactic morphism of\/ $L$ is an aperiodic monoid.
  \end{enumerate}
\end{cor}

\begin{rem} \label{rem:schutzcor}
  Schützenberger's original characterization~\cite{schutzsf} of the class \sfr of star-free languages is an immediate consequence of Corollary~\ref{cor:sfcaracg}. Indeed, consider a regular language~$L$ and let $\alpha: A^* \to M$ be its syntactic morphism.  Since $\sfr = \sfp{\stzer}$, it follows from Corollary~\ref{cor:sfcaracg} that $L \in \sfr$ if and only if the \stzer-kernel of $\alpha$ is aperiodic. Moreover, since $\stzer = \{\emptyset,A^*\}$ and syntactic morphisms are surjective, it is immediate that the \stzer-kernel of $\alpha$ is the whole syntactic monoid $M$. Hence, $L \in \sfr$ if and only if its syntactic monoid $M$ is aperiodic. This is exactly Schützenberger's theorem.
\end{rem}

\begin{proof}[Proof of Theorem~\ref{thm:sfcarac}]
  We fix a \vari \Cs for the proof. Moreover, we let $L\subseteq A^*$ be a regular language and $\alpha: A^* \to M$ be its syntactic morphism. We already proved that $(2) \Rightarrow (1)$ in Theorem~\ref{thm:bsdcar}. Hence, it suffices to prove that $(1) \Rightarrow (3)$ and $(3) \Rightarrow (2)$. The latter implication corresponds of the inclusion $\sfp{\Cs}\subseteq \bsdp{\Cs}$, which we omitted in the proof of Theorem~\ref{thm:bsdcar}.

  \medskip
  \noindent
  {\bf Implication $(1) \Rightarrow (3)$.} Assume that $L \in \sfp{\Cs}$. For every idempotent $e \in E(M)$, we prove that the \Cs-orbit of $e$ for $\alpha$ is aperiodic. By definition, this boils down to proving that for all $s \in M$ such that $(e,s) \in M^2$ is a \Cs-pair for $\alpha$, we have $(ese)^{\omega+1} = (ese)^{\omega}$. The argument is based on Proposition~\ref{prop:aper}. By hypothesis on \Cs, it follows from Proposition~\ref{prop:vari} that \sfp{\Cs} is a \vari. Hence, since $L \in \sfp{\Cs}$, Proposition~\ref{prop:synmemb} implies that its syntactic morphism~$\alpha$ is an \sfp{\Cs}-morphism. Consequently, Proposition~\ref{prop:aper} yields a \Cs-morphism $\eta: A^* \to N$ such that for every $u \in A^*$, if $\eta(u)$ is idempotent, then we have $(\alpha(u))^{\omega+1} = (\alpha(u))^{\omega}$.

  Since $\eta$ is a \Cs-morphism and $(e,s)$ is a \Cs-pair for $\alpha$, Lemma~\ref{lem:cmorph} yields $u,v \in A^*$ such that $\eta(u) = \eta(v)$, $\alpha(u) = e$ and $\alpha(v) = s$. Let $k = \omega(N)$. Clearly, $\eta(u^{k-1}vu^k) = \eta(u^k)$ is an idempotent of~$N$. Hence, it follows from the definition of $\eta$ that $(\alpha(u^{k-1}vu^k))^{\omega+1} = (\alpha(u^{k-1}vu^k))^{\omega}$. Finally, since $\alpha(u) = e$ and $\alpha(v) = s$, this exactly says that $(ese)^{\omega+1} = (ese)^{\omega}$, as desired.

  \medskip
  \noindent
  {\bf Implication $(3) \Rightarrow (2)$.} We assume that all \Cs-orbits for $\alpha$ are aperiodic monoids and prove that $L \in \bsdp{\Cs}$. We first present a preliminary definition. Given a language $K \subseteq A^*$ and an element $s\in M$, we say that $K$ is \emph{$s$-safe} when $s\alpha(u)=s\alpha(v)$ for every $u,v \in K$. Additionally, given a language $P \subseteq A^*$, an \emph{\bsdp{\Cs}-partition of $P$} is a finite partition of $P$ into languages which all belong to  \bsdp{\Cs}. The argument is based on the following lemma.

  \begin{lem}\label{lem:sfclos:fromapertostar}
    Let $P \subseteq A^+$ be a prefix code with bounded synchronization delay such that there exists an \bsdp{\Cs}-partition \Hb of\/ $P$ whose elements are all $1_M$-safe. Then, for every $s \in M$, there exists an \bsdp{\Cs}-partition \Kb of\/ $P^*$ whose elements are all $s$-safe.
  \end{lem}

  We first apply Lemma~\ref{lem:sfclos:fromapertostar} to show that every language recognized by $\alpha$ (such as $L$) belongs to \bsdp{\Cs} and conclude the main argument. By definition, \bsdp{\Cs} is closed under disjoint union. Hence, it suffices to show that $\alpha\inv(t) \in \bsdp{\Cs}$ for every $t \in M$. We fix such an element $t$ in $M$.

  Note that $A \subseteq A^+$ is a prefix code with bounded synchronization delay. Moreover, the set $\Hb = \bigl\{\{a\} \mid a \in A\bigr\}$ is an \bsdp{\Cs}-partition of $A$ such that every $H \in \Hb$ is $1_M$-safe. Hence, we may apply Lemma~\ref{lem:sfclos:fromapertostar} in the case when $P = A$ and $s = 1_M$. We get an \bsdp{\Cs}-partition \Kb of $A^*$ such that every $K \in \Kb$ is $1_M$-safe. Being $1_M$-safe means that $\alpha(K)$ is a singleton for every $K\in\Kb$. Therefore, $\alpha\inv(t)$ is the disjoint union of all $K \in \Kb$ intersecting $\alpha\inv(t)$. Since $\bsdp{\Cs}$ is closed under disjoint union, we obtain that $\alpha\inv(t) \in \bsdp{\Cs}$, which concludes the main argument.

  \smallskip

  It remains to prove Lemma~\ref{lem:sfclos:fromapertostar}. Let $P \subseteq A^*$ be a prefix code with bounded synchronization delay and consider an \bsdp{\Cs}-partition \Hb of $P$ such that every $H \in \Hb$ is $1_M$-safe. Moreover, fix $s \in M$. We need to build an \bsdp{\Cs}-partition \Kb of $P^*$ such that every $K \in \Kb$ is $s$-safe. We proceed by induction on the three following parameters listed by order of importance:
  \begin{enumerate}
    \item The size of $\alpha(P^+) \subseteq M$.
    \item The size of $\Hb$.
    \item The size of $s \cdot \alpha(P^*) \subseteq M$.
  \end{enumerate}

  We distinguish two cases depending on whether the following property of $s$ and $\Hb$ holds. We say that \emph{$s$ is \Hb-stable} when:
  \begin{equation}\label{eq:sfclos:mostable}
    \text{for every $H \in \Hb$,} \quad  s \cdot \alpha(P^*) = s \cdot \alpha(P^*H).
  \end{equation}
  The base case happens when $s$ is \Hb-stable: we conclude directly without using induction. Otherwise, we use induction on our three parameters.

  \medskip
  \noindent
  {\bf Base case: $s$ is \Hb-stable.} This is the only part of the proof where we need the hypothesis that all \Cs-orbits for $\alpha$ are aperiodic monoids. Lemma~\ref{lem:cmorph} yields a \Cs-morphism $\eta: A^* \to N$ such that for every $u,v \in A^*$, if $\eta(u) = \eta(v)$, then  $(\alpha(u),\alpha(v))$ is a \Cs-pair for $\alpha$. We define,
  \[
    \Kb = \{P^* \cap \eta\inv(t) \mid t \in N\}.
  \]
  Clearly, \Kb is a partition of $P^*$. Moreover, it only contains languages in \bsdp{\Cs}. Indeed, we have $P \in \bsdp{\Cs}$: it is the disjoint union of all languages in the \bsdp{\Cs}-partition \Hb of $P$. Therefore, $P^* \in \bsdp{\Cs}$ since $P$ is a prefix code with bounded synchronization delay. Hence, $P^* \cap \eta\inv(t) \in \bsdp{\Cs}$ for every $t \in N$ since $\eta\inv(t) \in \Cs$ (as $\eta$ is a \Cs-morphism). It remains to show that every language $K \in \Kb$ is $s$-safe. The argument is based on the following fact, which is where we use the hypothesis that $s$ is \Hb-stable.

  \begin{fct}\label{fct:sfclos:icarbase}
    Let $q,f \in \alpha(P^*)$ such that $f$ is idempotent. Then, we have $sqf = sq$.
  \end{fct}

  \begin{proof}
    The proof is based on the following preliminary result. For every $u,v \in P^*$, we show that,
    \begin{equation}\label{eq:sfclos:prelim}
      \text{there exists $r\in \alpha(P^*)$ such that $sr \alpha(u) = s\alpha(v)$.}
    \end{equation}
    We fix $u,v \in P^*$ for the proof of~\eqref{eq:sfclos:prelim}. There exists a decomposition $u = u_1 \cdots u_n$ with $u_1,\dots,u_n \in P$. We use induction on the length $n$ of this decomposition. If $n = 0$, then $u = \veps$ and it suffices to choose $r = \alpha(v) \in \alpha(P^*)$. Otherwise, $u = wu'$ with $w \in P$ and $u' \in P^*$ admits a decomposition of length $n-1$. Induction yields $r' \in \alpha(P^*)$ such that $sr'\alpha(u') = s\alpha(v)$. Moreover, since $w \in P$ and \Hb is a partition of $P$, there exists some $H \in \Hb$ such that $w \in H$. Since $s$ is \Hb-stable and $r' \in \alpha(P^*)$, it follows from~\eqref{eq:sfclos:mostable} that there exists $r \in \alpha(P^*)$ and $x \in \alpha(H)$ such that $sr' = srx$. Additionally, recall that $H \in \Hb$ is $1_M$-safe by hypothesis. Hence, since $x,\alpha(w) \in \alpha(H)$, we have $x = \alpha(w)$. Therefore, $sr'=sr\alpha(w)$. Altogether, this yields $sr \alpha(u) = sr \alpha(w)\alpha(u') = sr'\alpha(u') = s\alpha(v)$, which concludes the proof of~\eqref{eq:sfclos:prelim}.

    It remains to prove the fact. Consider $q,f \in \alpha(P^*)$ such that $f$ is idempotent. By definition, there exist $u,v \in P^*$ such that $q = \alpha(v)$ and $f = \alpha(u)$. Hence~\eqref{eq:sfclos:prelim} yields $r \in \alpha(P^*)$ such that $srf = sq$. Since $f$ is idempotent, we obtain $sqf = srff = srf = sq$, which completes the proof.
  \end{proof}

  We are ready to show that every language $K \in \Kb$ is $s$-safe. By definition, $K = P^* \cap \eta\inv(t)$ for $t \in N$. Given $u,v\in K$, we have to show that $s\alpha(u)=s\alpha(v)$. Let $n = \omega(M)$ and $e = (\alpha(u))^n \in E(M)$. Since $u,v \in K$, we have $\eta(u) = \eta(v) = t$. Hence, $\eta(u^n) = \eta(vu^{n-1})$. By definition of $\eta$, it follows that $(e,\alpha(vu^{n-1})) =  (\alpha(u^n),\alpha(vu^{n-1}))$ is a \Cs-pair. Hence, $e\alpha(vu^{n-1})e$ belongs to the \Cs-orbit of $e$, which is aperiodic by hypothesis on $\alpha$. This yields $(e\alpha(vu^{n-1})e)^n = (e\alpha(vu^{n-1})e)^{n+1}$. Multiplying by $s$ on the left gives $s(e\alpha(vu^{n-1})e)^n = s(e\alpha(vu^{n-1})e)^{n+1}$. Since $n = \omega(M)$, we know that $(e\alpha(vu^{n-1})e)^n$ is an idempotent of $M$. Moreover, since $u,v \in K$, we have $(e\alpha(vu^{n-1})e)^n \in \alpha(P^*)$. Therefore, Fact~\ref{fct:sfclos:icarbase} yields $s(e\alpha(vu^{n-1})e)^n=s$. Together with  $s(e\alpha(vu^{n-1})e)^n = s(e\alpha(vu^{n-1})e)^{n+1}$, this yields $s = se\alpha(vu^{n-1})e$.  We now multiply by $\alpha(u)$ on the right to get $s\alpha(u) = se\alpha(v)e$ (recall that $e = \alpha(u^n)$). Finally, $e$ is an idempotent of $M$ and since $u,v \in K$, we have $e,\alpha(v) \in \alpha(P^*)$. Hence, we may apply Fact~\ref{fct:sfclos:icarbase} twice to get $se\alpha(v)e = s \alpha(v)$. Altogether, this yields $s\alpha(u)=s\alpha(v)$, as~desired.

  \medskip
  \noindent
  {\bf Inductive step: $s$ is not \Hb-stable.} By hypothesis, we know that~\eqref{eq:sfclos:mostable} does not hold. Therefore, we get some $H \in \Hb$ such that the following \emph{strict} inclusion holds,
  \begin{equation}\label{eq:sfclos:godinduc}
    s \cdot \alpha(P^*H) \subsetneq s \cdot \alpha(P^*).
  \end{equation}
  We fix this language $H \in \Hb$ for the remainder of the proof. The following fact is proved by induction on our second parameter (the size of $\Hb$).

  \begin{fact}\label{fct:sfclos:alphind}
    There exists an \bsdp{\Cs}-partition \Ub of $(P \setminus H)^*$ such that every $U \in \Ub$ is $1_M$-safe.
  \end{fact}

  \begin{proof}
    Clearly, $P \setminus H \subseteq P$ remains a prefix code with bounded synchronization delay by Fact~\ref{fct:sprefsub}. Moreover, it is immediate that $\Gb = \Hb \setminus \{H\}$ is an \bsdp{\Cs}-partition of $P\setminus H$ such that every $G \in \Gb$ is $1_M$-safe. Additionally, it is clear that $\alpha((P \setminus H)^+) \subseteq \alpha(P^+)$ (our first induction parameter has not increased) and $\Gb \subsetneq \Hb$ (our second parameter has decreased). Hence, we may apply induction in Lemma~\ref{lem:sfclos:fromapertostar} for the case when $P,\Hb$ and $s$ have been replaced by $P \setminus H$, $\Gb$ and $1_M$ respectively. This yields an \bsdp{\Cs}-partition \Ub of $(P \setminus H)^*$ such that every $U \in \Ub$ is $1_M$-safe.
  \end{proof}

  We fix the partition \Ub of $(P \setminus H)^*$ given by Fact~\ref{fct:sfclos:alphind} for the remainder of the proof. We distinguish two independent subcases. Since $H$ is an element of the partition \Hb of $P$, we have $H \subseteq P$. It is therefore immediate that the inclusion $\alpha(P^*H) \subseteq \alpha(P^+)$ holds. We use a different argument depending on whether this inclusion is strict or not.

  \medskip
  \noindent
  {\bf Subcase~1: we have the equality \boldmath{$\alpha(P^*H) = \alpha(P^+)$}.} We handle this subcase with induction on our third parameter (\emph{i.e.}, the size of $s\alpha(P^*)$). Recall that we have to build an \bsdp{\Cs}-partition \Kb of~$P^*$ containing only $s$-safe languages.

  Observe that the hypothesis that $H$ is $1_M$-safe means that there exists some element $t \in M$ satisfying $\alpha(H) = \{t\}$. Similarly, since every $U \in \Ub$ is $1_M$-safe, there exists some element $r_U \in M$ such that $\alpha(U) = \{r_U\}$. We fix these elements of $M$ for the rest of this subcase. The construction of \Kb is based on the following lemma, which is where we use our hypotheses and induction.

  \begin{fact}\label{fct:sfclos:sc1carac}
    For every $U \in \Ub$, there exists an \bsdp{\Cs}-partition $\Wb_{U}$ of $P^*$ such that every $W \in \Wb_U$ is $sr_Ut$-safe.
  \end{fact}

  \begin{proof}
    We fix $U \in \Ub$ for the proof. Since \Ub is a partition of $(P \setminus H)^*$, we have $\alpha(U) \subseteq \alpha(P^*)$ which means that $r_U \in \alpha(P^*)$. Thus, we have $sr_Ut \in s\alpha(P^*H)$. Therefore, $sr_Ut\alpha(P^*) \subseteq s\alpha(P^*HP^*)$ and since $H \subseteq P$, we get $sr_Ut\alpha(P^*) \subseteq s\alpha(P^+)$. Combined with our hypothesis in Subcase~1 (\emph{i.e.}, $\alpha(P^*H) = \alpha(P^+)$), this yields $sr_Ut\alpha(P^*) \subseteq s\alpha(P^*H)$. Finally, the hypothesis~\eqref{eq:sfclos:godinduc} of the inductive step yields the \emph{strict} inclusion $sr_Ut\alpha(P^*) \subsetneq s\alpha(P^*)$: the third parameter in our induction has decreased. On the other hand, the first two parameters have not increased, as they only depend on $P$ and~\Hb, which remain unchanged. Consequently, by induction, we may apply Lemma~\ref{lem:sfclos:fromapertostar} in the case when $s \in M$ has been replaced by $sr_Ut \in M$. This yields the desired \bsdp{\Cs}-partition $\Wb_{U}$ of~$P^*$.
  \end{proof}

  We are ready to define the partition \Kb of $P^*$. Using Fact~\ref{fct:sfclos:sc1carac}, we define,
  \[
    \Kb = \Ub \cup \bigcup_{U \in \Ub} \{UHW \mid W \in \Wb_U\}.
  \]
  It remains to show that \Kb is an \bsdp{\Cs}-partition of $P^*$ and that every $K \in \Kb$ is $s$-safe. Let us first verify that \Kb is a partition of $P^*$. Since $P$ is a prefix code, every word $w \in P^*$ admits a \emph{unique} decomposition $w = w_1 \cdots w_n$ with $w_1,\dots,w_n \in P$. If no factor $w_i$ belongs to $H$, then $w \in (P \setminus H)^*$ and $w$ belongs to some unique $U \in \Ub$. Otherwise, let $w_i$ be the leftmost factor such that $w_i \in H$. This implies that $w_1 \cdots w_{i-1} \in (P \setminus H)^*$, which yields a unique $U \in \Ub$ such that $w_1 \cdots w_{i-1} \in U$. Moreover, $w_{i+1} \cdots w_n \in P^*$, which yields a unique $W \in \Wb_U$ such that $w_{i+1} \cdots w_n \in W$. It follows that $w \in UHW$, and that $UHW$ is the unique element of \Kb containing $w$.

  Let us next verify that every $K \in \Kb$ belongs to \bsdp{\Cs}. If $K \in \Ub$, this is immediate by definition of \Ub in Fact~\ref{fct:sfclos:alphind}. Otherwise, $K = UHW$ with $U \in \Ub$ and $W \in \Wb_U$. We know that $U,H,W \in \bsdp{\Cs}$: this an hypothesis for $H$ and stated in Facts~\ref{fct:sfclos:alphind} and~\ref{fct:sfclos:sc1carac} for $U$ and $W$. Furthermore, one may verify that the concatenation $UHW$ is \emph{unambiguous} since $P$ is a prefix code, $U \subseteq (P \setminus H)^*$ and $W \subseteq P^*$. Altogether, it follows that $K \in \bsdp{\Cs}$.

  Finally, we prove that every $K \in \Kb$ is $s$-safe. If $K \in \Ub$, this is immediate since $K$ is actually $1_M$-safe by definition of \Ub in Fact~\ref{fct:sfclos:alphind}. Otherwise, $K=UHW$ with $U \in \Ub$ and $W \in \Wb_U$. Consider $w,w' \in K$. We show that $s\alpha(w) = s \alpha(w')$. By definition, $\alpha(H) = \{t\}$ and $\alpha(U) = \{r_U\}$ which implies that $s\alpha(w) = str_U \alpha(x)$ and $s\alpha(w') = str_U \alpha(x')$ for $x,x' \in W$. Moreover, $W \in \Wb_U$ is $sr_Ut$-safe by definition in Fact~\ref{fct:sfclos:sc1carac}. Therefore, $s\alpha(w) = s \alpha(w')$. This concludes the proof of Subcase~1.

  \medskip
  \noindent
  {\bf Subcase~2: we have the strict inclusion \boldmath{$\alpha(P^*H) \subsetneq \alpha(P^+)$}.} In this case, we conclude using induction on the first parameter (\emph{i.e.}, the size of $\alpha(P^+)$). Recall that our objective is to construct an \bsdp{\Cs}-partition \Kb of $P^*$ containing only $s$-safe languages.

  Consider a word $w \in P^*$. Since $P$ is a prefix code, $w$ admits a unique decomposition $w = w_1 \cdots w_n$ with $w_1,\dots,w_n \in P$. We may uniquely decompose $w$ in two (possibly empty) parts: a prefix $w_1 \cdots w_i \in ((P \setminus H)^*H)^*$ and a suffix in $w_{i+1} \cdots w_n \in (P \setminus H)^*$. Using induction, we construct \bsdp{\Cs}-partitions of the possible prefixes and suffixes. Then, we combine them to construct a partition of the whole set $P^*$. Actually, we already handled the suffixes: Fact~\ref{fct:sfclos:alphind} provides an \bsdp{\Cs}-partition \Ub of $(P \setminus H)^*$. It remains to partition the prefixes. We do so this in the following fact, which is proved using the hypothesis of Subcase~2 and induction.

  \begin{fact}\label{fct:sfclos:sc2carac}
    There exists an \bsdp{\Cs}-partition \Vb of $((P \setminus H)^*H)^*$ such that every $V \in \Vb$ is $1_M$-safe.
  \end{fact}

  \begin{proof}
    Let $Q = (P \setminus H)^*H$. In view of Fact~\ref{fct:sprefnest}, $Q$ is a prefix code with bounded synchronization delay. We apply induction in Lemma~\ref{lem:sfclos:fromapertostar} for the case when $P$ has been replaced by $Q$. Doing so requires building an appropriate \bsdp{\Cs}-partition of $Q$ and proving that one of our induction parameters has decreased.

    Let $\Fb = \{UH \mid U \in \Ub\}$. Since \Ub is a partition of $(P\setminus H)^*$ and $P$ is a prefix code, one may verify that \Fb is a partition of $Q = (P\setminus H)^*H$. Moreover, it only contains languages in \bsdp{\Cs}. Indeed, if $U \in \Ub$, then the concatenation $UH$ is \emph{unambiguous} since $U\subseteq (P\setminus H)^*$ and $P$ is a prefix code. Moreover, $U,H \in \bsdp{\Cs}$ by hypothesis. Finally, $UH$ is $1_M$-safe since this is the case for both $U$ and $H$ by definition. It remains to show that our induction parameters have decreased. Since $Q = (P \setminus H)^*H$, it is clear that $Q^+ \subseteq P^*H$. Now, $\alpha(P^*H) \subsetneq \alpha(P^+)$ by hypothesis in Subcase~2, whence $\alpha(Q^+) \subsetneq \alpha(P^+)$: our first induction parameter has decreased. Thus, we may apply Lemma~\ref{lem:sfclos:fromapertostar} in the case when $P,\Hb$ and $s$ have been replaced by $Q$, $\Fb$ and $1_M$ respectively. This yields the desired \bsdp{\Cs}-partition \Vb of $((P \setminus H)^*H)^*$.
  \end{proof}

  We are ready to construct the \bsdp{\Cs}-partition \Kb of $P^*$ and conclude the main argument. Let $\Kb = \{VU \mid V \in \Vb \text{ and } U \in \Ub\}$. It is immediate by definition that \Kb is a partition of $P^*$ since $P$ is a prefix code and $\Vb,\Ub$ are partitions of $((P \setminus H)^*H)^*$ and $(P \setminus H)^*$ respectively (cf.~the above discussion). Moreover, every $K \in \Kb$ belongs to \bsdp{\Cs}. Indeed, one may verify that $K$ is the \emph{unambiguous} concatenation $VU$ of $V \in \Vb$ and $U\in \Ub$ which both belong to $\bsdp{\Cs}$. It remains to prove that every $K \in \Kb$ is $s$-safe. Let $w,w' \in K$, we show that $s\alpha(w) = s\alpha(w')$. By definition, we have $K = VU$ with $V \in \Vb$ and $U \in \Ub$. Therefore, $w =vu$ and $w' = v'u'$ with $u,u' \in U$ and $v,v' \in V$. Since $U$ and $V$ are both $1_M$-safe by definition, we have $\alpha(u) = \alpha(u')$ and $\alpha(v) = \alpha(v')$. It follows that $s\alpha(w) = s\alpha(w')$, which concludes the proof of Lemma~\ref{lem:sfclos:fromapertostar}.
\end{proof}


\section{First logical characterization: first-order logic}
\label{sec:folog}
We now turn to the logical characterizations of star-free closure. In this section, we present the first one. It generalizes a well-known theorem of McNaughton and Papert~\cite{mnpfosf}, which characterizes the star-free languages as those which can be defined by a sentence of first-order logic equipped with the linear ordering (\emph{i.e.}, $\sfr=\sfp{\stzer} = \fow$). Here, we extend this theorem to all classes \sfp{\Cs} where \Cs is a \vari. More precisely, we associate a set \infsigc of first-order predicates to every \vari \Cs (its definition is taken from~\cite{PZ:generic19}). Then, we show that \sfp{\Cs} contains exactly the languages that can be defined by a sentence of first-order logic equipped with the predicates in \infsigc (\emph{i.e.}, $\sfp{\Cs} = \foc$). First, we briefly recall the definition of first-order logic over words. Then, we present the theorem itself.

\subsection{Definitions}

We view each word $w \in A^*$ as a logical structure. Its domain is the set $\pos{w} = \{0,\dots,|w|+1\}$ of positions in $w$. A position $i$ such that $1 \leq i \leq |w|$ carries a label in $A$. On the other hand, $0$ and $|w|+1$ are artificial \emph{unlabeled} positions. We use first-order logic (\fo) to express properties of words~$w$: a formula can quantify over the positions of $w$ with first-order variables and use a predetermined set of predicates to test properties of these positions. We also allow two constants ``$\mathit{min}$'' and ``$\mathit{max}$''  interpreted as the artificial unlabeled positions $0$ and $|w|+1$. Let us briefly recall the definition.

\medskip
\noindent
{\bf Signatures.}  A \emph{signature} is a (possibly infinite) set of predicates interpreted over words in $A^*$. Consider a natural number $k \in \nat$. A \emph{predicate of arity $k$} over $A$ is defined by a symbol $P$ and for every word $w \in A^*$, an interpretation of $P$ as a relation of arity $k$ over the set of positions of $w$. More precisely, with every word $w \in A^*$, the predicate $P$ associates a set of $k$-tuples of positions of~$w$ (\emph{i.e.}, a subset of $((\pos{w})^k$). If $(i_1,\dots,i_k)$ is a $k$-tuple in this set, we shall say that \emph{$P(i_1,\dots,i_k)$ holds (in $w$)}. All predicates that we consider in practice are either unary (they have arity $1$) or binary (they have arity $2$). Let us present them.

First, we use \emph{label predicates}. For every letter $a\in A$, we associate a unary predicate (also denoted by ``$a$''). It is interpreted as the unary relation selecting all positions whose label is $a$: given a word $w \in A^*$ and $i \in \pos{w}$, we have that $a(i)$ holds when the label of $i$ is ``$a$''. In particular, if $a(i)$ holds, then $i$ cannot be one of the two artificial positions $0$ and $|w|+1$. Abusing notation, we write ``$A$'' for the set of all label predicates. Moreover, we use a binary predicate ``$<$'', interpreted as the linear ordering between positions. Given a word $w$ and $i,j \in \pos{w}$, we have that ${<}(i,j)$ holds if $i < j$. For the sake of improved readability, we use the infix notation, writing $i < j$ instead of~${<}(i,j)$.

Finally, with each class \Cs, we associate two \emph{generic} sets of predicates. The first one, written \infsigc, contains a binary ``infix'' predicate $I_L(x,y)$ for every $L \in \Cs$. Given $w \in A^*$ and two positions $i,j \in \pos{w}$, we have $w \models I_L(i,j)$ when $i< j$ \emph{and} $\infix{w}{i}{j}\in L$. The second set, written  \prefsigc, contains a unary ``prefix'' predicate $P_L(x)$ for every $L \in \Cs$. Given $w \in A^*$ and a position $i \in \pos{w}$, we have $w \models P_L(i)$ when $0 < i$ \emph{and} $\infix{w}{0}{i} \in L$.

\begin{rem}\label{rem:order-as-IL}
  All classes \Cs that we consider in practice are \varis. In particular, this implies that $A^* \in \Cs$. Hence, set\/ \infsigc always contains the linear order predicate, since ``$x<y$'' is clearly equivalent to ``$I_{A^*}(x,y)$''. In particular, when \Cs is the trivial \vari $\stzer= \{\emptyset,A^*\}$, the signature \infsig{\stzer} contains only $I_{A^*}$ and $I_{\emptyset}$. Since $I_{\emptyset}(i,j)$ \emph{never holds}, using \infsig{\stzer} boils down to considering $\{<\}$.
\end{rem}

\medskip
\noindent
{\bf First-order formulas.} With a signature \frS (\emph{i.e.}, \frS is a possibly infinite set of predicates), we associate a set $\fo[\frS]$ of first-order formulas. They are built-up from simple expressions called \emph{atomic formulas}. The atomic formulas can test properties of the positions that were quantified by first-order variables using either equality or the predicates in \frS. More precisely, they are of the form:
\[
  x_1=x_2\quad\text{or}\quad  P(x_1,\dots,x_k),
\]
where $P \in \frS$ is a predicate of arity $k$ for some $k \in \nat$ and, for every $i \leq k$, $x_i$ is either a first-order variable or one of the two constants $min$ and $max$ (symbols can be repeated: it may happen that $x_i$ and $x_j$ are the same symbol for $i \neq j$). We define $\fo[\frS]$ as the least set of expressions containing the atomic formulas and closed under the following rules:
\begin{itemize}
  \item \emph{Disjunction}: if $\varphi$ and $\psi$ are $\fo[\frS]$ formulas, then so is $(\varphi \vee \psi)$.
  \item \emph{Negation}: if $\varphi$ is an $\fo[\frS]$ formula, then so is $(\neg \varphi)$.
  \item \emph{Existential quantification}: for any first-order variable $x$, if $\varphi$ is an $\fo[\frS]$ formula, then so is $(\exists x\ \varphi)$.
\end{itemize}
For the sake of improved readability, we omit the parentheses when there is no ambiguity. Moreover, we define the other standard logical connectives as abbreviations. We write $\varphi \wedge \psi$ for $\neg ((\neg \varphi) \vee (\neg \psi))$ and $\varphi \Rightarrow \psi$ for $(\neg \varphi) \vee \psi$. We also write $\forall x\ \varphi$ for $\neg (\exists x\ \neg \varphi)$.

Finally, we use the standard notion of ``\emph{free variable}''. Let $\varphi$ be an $\fo[\frS]$ formula. An occurrence of some variable $x$ in $\varphi$ is said to be \emph{bound} when it occurs inside an atomic formula that is under the scope of a quantification $\exists x$. For example, in the following formula $a(x) \wedge \exists y\ (x < y \wedge b(y)) \wedge \exists x\ c(x)$,
the occurrence of the variable $x$ inside the atomic formula $c(x)$ is bound. We say that a variable $x$ is \emph{free} in a formula $\varphi$ if there \emph{exists} an occurrence of $x$ in $\varphi$ that is \emph{not} bound. For example, in the formula above, there also exist two occurrences of $x$ that are not bound (inside the atomic formula $a(x)$ and inside the atomic formula $x < y$). Hence, $x$ is a free variable of this formula. A \emph{sentence} is a formula that has no free variable.

\medskip
\noindent
{\bf Semantics.} We define when a word $w \in A^*$ satisfies a \emph{sentence} $\varphi$ (a fact that we denote by $w \models \varphi$). The definition is by structural induction on the sentence $\varphi$. This means that we actually need to give the semantics of \emph{all formulas}, not just sentences. Indeed, in general, a sentence may have subformulas with free variables. To tackle this issue, we need the notion of variable assignment. Let $w \in A^*$ be a word and let \Xs be a finite set of variables. We define an \emph{assignment} of \Xs in $w$ as a map $\mu: \Xs \to \pos{w}$. Additionally, we canonically extend every such assignment as a map $\mu: \Xs \cup \{min,max\} \to \pos{w}$ by defining $\mu(min) = 0$ and $\mu(max) = |w|+1$.

Let $\varphi$ be an $\fo[\frS]$ formula and let $\Xs$ be a set of variables containing all free variables of $\varphi$. For every word $w \in A^*$ and every assignment $\mu: \Xs \to \pos{w}$, we write $w,\mu \models \varphi$ when one the following properties hold:
\begin{itemize}
  \item $\varphi :=$ ``$x_1=x_2$'' and $\mu(x_1)=\mu(x_2)$ holds.
  \item $\varphi :=$ ``$P(x_1,\dots,x_k)$'' for some predicate $P \in \frS$ and $P(\mu(x_1),\dots,\mu(x_k))$ holds.
  \item $\varphi :=$ ``$\psi \vee \chi$'' and either $w,\mu \models \psi$ or $w,\mu \models \chi$.
  \item $\varphi :=$ ``$\neg \psi$'' and $w,\mu \not\models \psi$ ($w$ does not satisfy $\psi$ under $\mu$).
  \item $\varphi :=$ ``$\exists y\ \psi$'' and there exists an assignment $\gamma: \Xs \cup \{y\} \to \pos{w}$ such that $\mu(x) = \gamma(x)$ for every $x \in \Xs \setminus \{y\}$ and $w,\gamma \models \psi$.
\end{itemize}
The definition depends on an assignment $\mu: \Xs \to \pos{w}$ where \Xs contains \emph{all} free variables of $\varphi$.  In particular, it may happen that \Xs contains variables that are \emph{not} free in $\varphi$. Yet, this is allowed only for the sake of simplifying the presentation: one may verify that whether $w,\mu \models \varphi$ only depends on the restriction of $\mu$ to the variables that are free in $\varphi$. In particular, when no variable is free variable (\emph{i.e.}, $\varphi$ is a sentence), whether $w,\mu \models \varphi$ holds is independent from the assignment $\mu$. Hence, we simply write $w \models \varphi$ in this case. Altogether, it follows that each sentence $\varphi$ of $\fo[\frS]$ defines a language: we let $L(\varphi) = \{w \in A^* \mid w \models \varphi\}$.

\medskip
\noindent
{\bf Classes associated to first-order logic.} To every set of predicates \frS, we associate a class of languages $\fo(\frS)$. It consists of all languages that can be defined by a~sentence in $\fo[A,\frS]$ (that is, we use the signature $A \cup \frS$, containing the label predicates \emph{and} those in \frS). For the sake of avoiding clutter, we often abuse terminology and speak of an $\fo(\frS)$-sentence to mean an  $\fo[A,\frS]$-sentence.

\begin{exa}
  Let $A = \{a,b\}$. Let us present some languages in \fow. We use the following abbreviation in first-order sentences: we write ``$x+1 = y$'' for the formula ``$(x < y) \wedge \neg \exists z\ (x < z \wedge z < y)$''. In other words, ``$+1$'' is interpreted as the successor relation over positions. We have $A^*aA^*bA^*a \in \fow$ since it is defined by the following sentence: $\exists x \exists y\ (x < y) \wedge a(x) \wedge b(y)) \wedge (\exists x\ a(x) \wedge (x+1 = max))$. Moreover, $(ab)^* \in \fow$ as well, since it is defined by the following \fow sentence:
  \[
    \forall x \forall y\ \left(
      x+1 = y
    \right) \Rightarrow
    \left(
      \begin{array}{ll}
        & \left(x = min \wedge y = max\right) \\
        \vee & \left(x = min \wedge a(y)\right)\\
        \vee & \left(a(x) \wedge b(y)\right) \\
		\vee & \left(b(x) \wedge a(y)\right) \\
        \vee& \left(b(x) \wedge y = max\right)
      \end{array}\right).
  \]
\end{exa}

We are interested in classes $\fo(\infsigc)$ where \Cs is a \vari. Indeed, we prove below that $\fo(\infsigc) = \sfp{\Cs}$ in that case. However, when \Cs is a \emph{group \vari}, the presentation of $\fo(\infsigc)$ can be simplified in view of the following lemma.

\begin{lem} \label{lem:gensig}
  Let \Gs be a group \vari. Then, $\fo(\infsigg) = \fo(<,\prefsigg)$.
\end{lem}

\begin{proof}
  The inclusion $\fo(<,\prefsigg) \subseteq \fo(\infsigg)$ is immediate since all predicates in $\{<\}\cup \prefsigg$ can be simulated using those in \infsigg. Indeed, $x<y$ is equivalent to $I_{A^*}(x,y)$ and $P_L(x)$ (for $L \in \Gs)$ is equivalent to $I_L(min,x)$. We now prove that $\fo(\infsigg) \subseteq \fo(<,\prefsigg)$. By definition, it suffices to prove that for every language $L \in \Gs$, the atomic formula $I_L(x,y)$ is equivalent to a formula of $\fo(<,\prefsigg)$. Proposition~\ref{prop:genocm} yields a \Gs-morphism  $\alpha: A^* \to G$ recognizing $L$. Since \Gs is a group \vari, Lemma~\ref{lem:gmorph} implies that $G$ is a group. Let $F \subseteq G$ be the set such that $\alpha\inv(F) = L$.

  For every $g\in G$, the language $\alpha\inv(g)$ belongs to \Gs, whence $P_{\alpha\inv(g)}$ is a predicate in $\prefsigg$.  Since $G$ is a group, it is immediate that $\alpha(v) = (\alpha(ua))\inv \alpha(uav)$ for all $u,v \in A^*$ and $a \in A$. Therefore, one may verify that $I_L(x,y)$ is equivalent to the following formula of $\fo(<,\prefsigg)$, where $T \subseteq G \times A \times G$ is the set of all triples $(g,a,h) \in G \times A \times G$ such that $(g\alpha(a))\inv h \in F$:
  \[
    (x < y) \wedge \Big((x = min \wedge P_L(y)) \vee   \bigvee_{(g,a,h) \in T} \big(P_{\alpha\inv(g)}(x) \wedge a(x) \wedge P_{\alpha\inv(h)}(y)\big)\Big).
  \]
  This concludes the proof.
\end{proof}

\begin{exa}
  Lemma~\ref{lem:gensig} applies to important sets of predicates. First, if \Gs is the trivial \vari $\stzer = \{\emptyset,A^*\}$, all predicates in $\prefsig{\stzer}$ are trivial. Hence, $\fo(<,\prefsig{\stzer}) = \fow$.

  Next, let us consider the class \md of \emph{modulo languages}, consisting in Boolean combinations of languages $\{w \in A^* \mid |w| \equiv k \bmod m\}$ with $k,m \in \nat$ such that $k<m$. In this case, we obtain first-order logic with \emph{modular predicates}. For all $k,m \in \nat$ such that $k < m$, the set ``$MOD$'' of \emph{modular predicates} contains a unary predicate $M_{k,m}$ selecting the positions $i$ such that $i \equiv k \bmod m$. One may use Lemma~\ref{lem:gensig} to verify that $\fo(<,\prefsig{\md}) = \fowm$.

  Finally, we consider the class \abg of \emph{alphabet modulo testable languages}. If $w\in A^*$ and $a \in A$, we let $\#_a(w) \in \nat$ be the number of occurrences of letter~$a$ in $w$. The class \abg consists of all Boolean combinations of languages $\{w \in A^* \mid \#_a(w) \equiv k \bmod m\}$ where $a \in A$ and $k,m \in \nat$ such that $k < m$ (these are the languages recognized by commutative groups). In this case, we obtain first-order logic with \emph{alphabetic modular predicates}. For all $a \in A$ and all $k,m \in \nat$ the set ``$AMOD$''  of \emph{alphabetic modular predicates} contains a unary predicate $M^a_{k,m}$ selecting the positions $i$ such $\#_a(\prefix{w}{i}) \equiv k \bmod m$). One may use Lemma~\ref{lem:gensig} to verify that $\fo(<,\prefsig{\abg}) = \fowam$.
\end{exa}

\subsection{Main theorem}

We may now present the main result of the section. It connects star-free closure to first-order logic for all input classes that are \varis.

\begin{thm} \label{thm:folog}
  Let \Cs be a \vari. Then, $\sfp{\Cs} = \foc$.
\end{thm}

Additionally, in view of Lemma~\ref{lem:gensig}, Theorem~\ref{thm:folog} can be simplified when the input class is a group \vari \Gs. More precisely, we have the following theorem.

\begin{cor} \label{cor:folog}
  Let \Gs be a group \vari. Then, $\sfp{\Gs} = \fo(<,\prefsigg)$.
\end{cor}

\begin{rem} \label{rem:fomod}
  Corollary~\ref{cor:folog} has interesting applications when combined with Corollary~\ref{cor:sfcaracg} (\emph{i.e.}, the algebraic characterization of the class \sfp{\Gs}): we obtain that a regular language belongs to $\fo(<,\prefsigg)$ if and only if the \Gs-kernel of its syntactic morphism is aperiodic. In particular when $\Gs=\md$, recall from Remark~\ref{rem:stablemono} that \md-kernels correspond to a standard notion: \emph{stable monoids}. Hence, we obtain that a regular language belongs to \fowm if and only if the stable monoid of its syntactic morphism is aperiodic. This is a well-known theorem of\/ Barrington, Compton, Straubing and Thérien~\cite{MIXBARRINGTON1992478}, whose original proof relies on entirely different techniques.
\end{rem}

We now concentrate on the proof of Theorem~\ref{thm:folog}. Let us point out that both directions of the proofs are handled directly: we ``translate'' $\foc$ sentences into expressions witnessing membership in \sfp{\Cs} and vice-versa.

\begin{proof}[Proof of Theorem~\ref{thm:folog}]
  We fix a \vari \Cs and show that $\sfp{\Cs}=\foc$. The two inclusions are proved independently. We start with $\sfp{\Cs} \subseteq \foc$, which is simpler.

  \medskip
  \noindent
  {\bf Inclusion $\sfp{\Cs} \subseteq \foc$.} By definition of \sfp{\Cs}, it suffices to prove that $\Cs \subseteq \foc$, that $\{a\} \in \foc$ for every $a \in A$ and that $\foc$ is closed under union, complement and concatenation. It is clear that every $L \in \Cs$ is defined by the sentence $I_L(min,max)$ of $\foc$. Therefore, $\Cs \subseteq \foc$. Moreover, for every letter $a \in A$, the language $\{a\}$ is defined by the sentence $\exists x \left(a(x) \wedge (min+1 = x) \wedge (x+1 = max)\right)$. It is also clear that $\foc$ is closed under union and complement since Boolean connectives can be used freely in $\foc$ sentences. It remains to prove that $\foc$ is closed under concatenation. The argument is based on the following lemma.

  \begin{lem}\label{lem:foconcat}
    Let $L \in \foc$. There exists an $\foc$ formula $\varphi_L(x,y)$ with two free variables $x$ and $y$ such that for every $w \in A^*$ and every positions $i,j \in \pos{w}$ in $w$ such that $i < j$, we have $w \models \varphi_L(i,j)$ if and only if $\infix{w}{i}{j} \in L$
  \end{lem}

  \begin{proof}
    By hypothesis, there exists a sentence $\psi$ of $\foc$ defining $L$. We build $\varphi_L(x,y)$ from~$\psi$ by restricting quantifications with respect to the free variables $x,y$: we only allow quantification over positions between $x$ and $y$. Moreover, we use $x$ and $y$ themselves as substitutes for the unlabeled positions. More precisely, $\varphi_L(x,y)$ is defined by applying the  following modifications to $\psi$:
    \begin{enumerate}
      \item Every subformula of the form $\exists z\ \Gamma$ is recursively replaced by,
      \[
        \exists z\ \left(\left((z = x) \vee (x < z \wedge z < y) \vee (z = y)\right) \wedge \Gamma\right).
      \]
      \item All occurrences of the constant $min$ are replaced by the free variable $x$ and all occurrences of the constant $max$ are replaced by $y$.
      \item Every atomic subformula of the form $a(z)$ for some $a \in A$ is replaced by,
      \[
        a(z) \wedge (x < z) \wedge (z < y).
      \]
    \end{enumerate}
    One may verify that $\varphi_L(x,y)$ satisfies the desired property.
  \end{proof}

  We may now prove that $\foc$ is closed under concatenation. Let $K,L \in \foc$. We prove that $KL \in \foc$. Lemma~\ref{lem:foconcat} yields two formulas $\varphi_K(x,y)$ and $\varphi_L(y,z)$ such that for every $w \in A^*$ and every positions $i,j \in \pos{w}$ in $w$ such that $i < j$, we have $w \models \varphi_K(i,j)$ if and only if $\infix{w}{i}{j} \in K$ and $w \models \varphi_L(i,j)$ if and only if $\infix{w}{i}{j} \in L$. It is now immediate that $KL$ is defined by the following sentence of \foc:
  \[
    \exists x\exists y \quad (x+1 = y) \wedge \varphi_K(min,y) \wedge \varphi_L(x,max).
  \]
  We obtain $KL \in \foc$, as desired. This completes the proof of the first inclusion: $\sfp{\Cs} \subseteq \foc$.

  \medskip
  \noindent
  {\bf Inclusion $\foc \subseteq \sfp{\Cs}$.} The proof of this inclusion is more involved. Yet, it is constructive as well: starting from an \foc sentence $\varphi$, we use structural induction on $\varphi$ to prove that the language it defines may be built from basic languages in \Cs using Boolean combinations and concatenations. Of course, this means that we shall have to deal with formulas that are \emph{not} sentences. We start with preliminary definitions.

  Let $n \in \nat$. An \emph{$n$-scheme} is a tuple $\bar{L} = (L_0,a_1,L_1,\dots,a_n,L_n)$ where $L_0, \dots,L_n \subseteq A^*$ and $a_1,\dots,a_n \in A$. Note that $0$-schemes are well-defined: they are simply languages. We write $\bar{L} \in \sfp{\Cs}$ to indicate that $L_0, \dots,L_n \in \sfp{\Cs}$. Additionally, we define an \emph{$n$-blueprint} as a \emph{finite} set \Lb of $n$-schemes. We write $\Lb \in \sfp{\Cs}$ to indicate that $\bar{L} \in \sfp{\Cs}$ for every $n$-scheme $\bar{L} \in \Lb$. Let us provide semantics for $n$-blueprints.

  Given $n \in \nat$, an \emph{$n$-split} is a \emph{linearly ordered} set \Xs of exactly $n+2$ first-order variables. For the sake of avoiding clutter, we often make the linear ordering implicit. For example, if we say that the set $\Xs = \{x_0,\dots,x_{n+1}\}$ is an $n$-split, we implicitly mean that the ordering is $x_0<x_1 < \cdots < x_n<x_{n+1}$. Moreover, given an $n$-split $\Xs = \{x_0,\dots,x_{n+1}\}$ and a word $w \in A^*$, we say that an assignment $\mu: \Xs \to \pos{w}$ is \emph{correct} to indicate that,
  \[
    0 = \mu(x_0) < \mu(x_1) < \cdots < \mu(x_{n+1}) = |w|+1.
  \]
  Now, consider an $n$-blueprint \Lb and an $n$-split $\Xs = \{x_0,\dots,x_{n+1}\}$, a word $w \in A^*$ and a correct assignment $\mu: \Xs \to \pos{w}$. We say that $(w,\mu)$ satisfies \Lb and write $w,\mu \models \Lb$ if and only if there exists an $n$-scheme $(L_0,a_1,L_1,\dots,a_n,L_n) \in \Lb$ which satisfies the two following conditions:
  \begin{enumerate}
    \item For all $1 \leq i \leq n$, the position $\mu(x_i) \in \pos{w}$ is labeled by $a_i$.
    \item For all $0 \leq i \leq n$, we have $\infix{w}{\mu(x_i)}{\mu(x_{i+1})}\in L_i$.
  \end{enumerate}
  Observe that when $n =0$, there exists only one correct assignment $\mu: \Xs \to \pos{w}$: we have $\mu(x_0) = 0$ and $\mu(x_1) = |w|+1$. Moreover, in that case, $w,\mu \models \Lb$ if and only there exists a language $L \in \Lb$ such that $w \in L$.

  The argument is based on the next proposition, proved by induction on the size of \foc formulas. The statement only applies to \emph{constant-free} formulas (\emph{i.e.}, which do not involve symbols ``$min$'' and ``$max$''), a restriction we shall deal with when using it to complete the main proof.

  \begin{prop} \label{prop:consf}
    Let $n \in \nat$, \Xs be an $n$-split and $\varphi$ be a constant-free \foc-formula whose free variables are contained in \Xs. Then, there exists an $n$-blueprint $\Lb_\varphi \in \sfp{\Cs}$ such that for all $w\in A^*$ and for all correct assignment $\mu: \Xs \to \pos{w}$, we have:
    \begin{equation} \label{eq:consf}
	  w,\mu \models \Lb_\varphi~~~\Longleftrightarrow~~~w,\mu \models \varphi.
    \end{equation}
  \end{prop}

  We first apply Proposition~\ref{prop:consf} to prove that $\foc \subseteq \sfp{\Cs}$. Let $L \in \foc$. We prove that $L \in \sfp{\Cs}$. The hypothesis $L \in \foc$ means that there is an \foc-sentence $\psi$ defining~$L$. Consider the $0$-split $\Xs = \{x_0,x_1\}$. Since variables can be renamed, we may assume without loss of generality that $x_0$ and $x_1$ do \emph{not} occur in $\psi$. Let $\varphi$ be the formula obtained from $\psi$ by replacing all occurrences of the constant symbols $min$ and $max$ by $x_0$ and $x_1$, respectively. It follows from Proposition~\ref{prop:consf} that there exists a $0$-blueprint $\Lb_\varphi \in \sfp{\Cs}$ satisfying~\eqref{eq:consf}. By definition $\Lb_\varphi$ is a finite set of languages in \sfp{\Cs}. We let $H = \bigcup_{K \in \Lb_\varphi} K \in \sfp{\Cs}$ and show that $H = L$ to complete the proof. Consider a word $w \in A^*$ and let $\mu: \Xs\to\pos{w}$ be the only correct assignment: $\mu(x_0) = 0$ and $\mu(x_1) = |w|+1$. By definition of $\varphi$ from $\psi$, it is immediate that $w \in L \Leftrightarrow w,\mu \models \varphi$. It then follows from~\eqref{eq:consf} that $w \in L \Leftrightarrow w,\mu \models \Lb_\varphi$. Finally by definition of $H$, we obtain $w \in L \Leftrightarrow w \in H$, as~desired.

  \smallskip

  This concludes the main argument. It remains to prove Proposition~\ref{prop:consf}. Let $n\in\nat$, \Xs be an $n$-split and $\varphi$ be a constant-free \foc formula whose free variables are contained in \Xs. We write $\Xs = \{x_0,\dots,x_{n+1}\}$. Moreover, we may assume without loss of generality that $\varphi$ does not contain the equality predicate, as subformulas of the form ``$y = z$'' may be replaced by the equivalent formula ``$\neg (y < z \vee z < y)$'' (recall from Remark~\ref{rem:order-as-IL} that the linear ordering is  available in~\infsigc). We use induction on the size of $\varphi$ (\emph{i.e.}, on the number of symbols in its syntax tree) to construct an $n$-blueprint $\Lb_\varphi \in \sfp{\Cs}$  satisfying~\eqref{eq:consf}.

  \medskip
  \noindent
  {\bf Atomic Formulas.} By hypothesis on $\varphi$ there are only two kinds of atomic formulas: those involving the label predicates and those involving the predicates in \infsigc. Note that since $\varphi$ is constant-free, there are no atomic formulas involving the constants $min$ and $max$.

  Assume first that $\varphi :=$ ``$a(x_h)$'' for some $a \in A$ and $h$ such that $0 \leq h \leq n+1$. There are two cases. If $h = 0$ or $h = n+1$, then it suffices to define $\Lb_\varphi = \emptyset \in \sfp{\Cs}$. Otherwise, we have $1 \leq h \leq n$. We define $\Lb_\varphi$ as the set of all $n$-schemes $(A^*,a_1,A^*,\dots,a_n,A^*)$ where $a_1,\dots, a_n \in A$ are letters such that $a_h = a$. it is clear that $\Lb_\varphi \in \sfp{\Cs}$ and one may verify from the definitions that it satisfies~\eqref{eq:consf}.

  Assume now that $\varphi :=$ ``$I_L(x_g,x_h)$'' for $g,h$ such that $0 \leq g,h \leq n+1$ and $L \in \Cs$. Recall that given a word $w \in A^*$ and two positions $i,j \in \pos{w}$, $I_L(i,j)$ holds if and only if $i < j$ and $\infix{w}{i}{j} \in L$. Hence, there are two cases. First, if $h \leq g$, it suffices to define $\Lb_\varphi = \emptyset \in \sfp{\Cs}$. Assume now that $g < h$. Since \Cs is a \vari,  Proposition~\ref{prop:genocm} yields a \Cs-morphism $\eta: A^* \to N$ recognizing $L$. Let $F \subseteq M$ such that $L = \alpha\inv(F)$. We define a set of tuples $T \subseteq M \times (A \times M)^{n}$ (when $n = 0$, we have $T \subseteq M$) as follows:
  \[
    T = \bigl\{(s_0,a_{1},s_{1},\dots,a_{n},s_{n})  \mid s_g\alpha(a_{g+1})s_{g+1} \cdots \alpha(a_{h-1})s_{h-1} \in F\bigr\}.
  \]
  We define $\Lb_\varphi$ as the set of all $n$-schemes of the form $(\eta\inv(s_0),\{a_1\},\eta\inv(s_1),\dots,\{a_n\},\eta\inv(s_n))$ such that  $(s_0,a_{1},s_{1},\dots,a_{n},s_{n}) \in T$. Since $\eta$ is a \Cs-morphism, it is immediate that $\Lb_\varphi \in \Cs \subseteq \sfp{\Cs}$. One may now verify from the definitions that $\Lb_\varphi$ satisfies~\eqref{eq:consf}.

  \medskip
  \noindent
  {\bf Disjunction.} Let us now assume that $\varphi :=$ ``$\psi_1 \vee \psi_2$''. For $i = 1,2$, induction yields an $n$-blueprint $\Lb_{i} \in \sfp{\Cs}$ which satisfies~\eqref{eq:consf} for $\psi_i$.  It is now immediate from the definitions that $\Lb_{\varphi} = \Lb_{1} \cup \Lb_{2} \in \sfp{\Cs}$ satisfies~\eqref{eq:consf} for $\varphi$.

  \medskip
  \noindent
  {\bf Negation.} We assume that $\varphi :=$ ``$\neg \psi$''. Induction yields an $n$-blueprint $\Lb_{\psi} \in \sfp{\Cs}$ such that if $w \in A^*$ and  $\mu: \Xs \to \pos{w}$ is a correct assignment, then $w,\mu \in \Lb_\psi \Leftrightarrow w,\mu \models \psi$. Let $\bar{L}_1,\dots, \bar{L}_k \in \sfp{\Cs}$ be the $n$-schemes such that $\Lb_{\psi} = \{\bar{L}_1,\dots,\bar{L}_k\}$. Finally, for every $j \leq k$, we let $(L_{0,j},a_{1,j},L_{1,j},\dots,a_{n,j},L_{n,j}) = \bar{L}_j$. Let $J \subseteq \{1,\dots,k\}$. We define,
  \[
    H_{i,J} = A^* \setminus \Bigl(\bigcup_{j \in J} L_{i,j}\Bigr) \text{ for $0 \leq i \leq n$}.
  \]
  Note that $H_{i,\emptyset} = A^*$. Moreover, since \sfp{\Cs} is closed under union and complement by definition, we have $H_{i,J} \in \sfp{\Cs}$. We now define $\Lb_{\varphi}$ as the set of all $n$-schemes $(H_{0,J_0},b_1,H_{1,J_1}, \dots,b_n,H_{n,J_n})$ where $b_1,\dots,b_n \in A$ and $J_0,\dots,J_n \subseteq \{1,\dots,k\}$ are such that for every $j \leq k$, either $b_i \neq a_{i,j}$ for some $i$ such that $1 \leq i \leq n$, or $j \in J_i$ for some $i$ such that $0 \leq i \leq n$. By definition, $\Lb_\varphi$ is an $n$-blueprint satisfying $\Lb_\varphi \in \sfp{\Cs}$. Moreover, it is straightforward to verify that if  $w \in A^*$ and  $\mu: \Xs \to \pos{w}$ is a correct assignment, then $w,\mu \models \Lb_\varphi$ if and only if $w,\mu \not\models \Lb_\psi$. Since $\varphi :=$ ``$\neg \psi$'', it is immediate that $\Lb_\varphi$ satisfies~\eqref{eq:consf}, as desired.

  \medskip
  \noindent
  {\bf First-order quantification.}  Finally, assume that $\varphi :=$ ``$\exists y\ \psi$''. Since variables may be renamed, we may assume without loss of generality that $y \not\in \Xs = \{x_0,\dots, x_{n+1}\}$. We define $\Lb_\varphi$ as the union of some $n$-blueprints that we build by induction. We use two kinds of $n$-blueprints in this union.

  For every $i$ such that $0 \leq i \leq n+1$, let $\psi_i$ be the formula obtained from $\psi$ by replacing every free occurrence of the variable $y$ with $x_i$. By definition all free variables in $\psi_i$ belong to $\Xs$. Moreover, the size of $\psi_i$ is the same as the one of $\psi$ which is strictly smaller than the size of $\varphi :=$ ``$\exists y\ \psi$''. Consequently, induction yields an $n$-blueprint $\Lb_{i} \in \sfp{\Cs}$ such that if $w \in A^*$ and $\mu: \Xs \to \pos{w}$ is a correct assignment, then $w,\mu \models \Lb_i \Leftrightarrow w,\mu \models \psi_i$.

  We turn to the second kind of $n$-blueprint. Let $i$ such that $0 \leq i \leq n$. We consider the linearly ordered set of first-order variables $\Ys_i = \{x_0,\dots,x_i,y,x_{i+1},\dots,x_{n+1}\}$ (\emph{i.e.}, $\Ys_i = \Xs \cup \{y\}$ and $y$ is placed between $x_i$ and $x_{i+1}$ for the linear ordering). Since the size of $\psi$ is strictly smaller than the one of $\varphi :=$ ``$\exists y\ \psi$', induction yields an $(n+1)$-blueprint $\Gb_{i} \in \sfp{\Cs}$ such that if $w \in A^*$ and $\gamma: \Ys_i \to \pos{w}$ is a correct assignment, then $w,\mu \models \Gb_i \Leftrightarrow w,\gamma \models \psi$. We use $\Gb_i$ to define an $n$-blueprint $\Hb_i$. Let $\bar{G} = (G_0,c_1,G_1,\dots,c_{n+1},G_{n+1})$ be an arbitrary $(n+1)$-scheme. We associate to $\bar{G}$ an $n$-scheme $f_{i}(\bar{G}) = (H_0,d_1,H_1,\dots,d_n,H_n)$ as follows:
  \begin{itemize}
    \item for every $j$ such that $0 \leq j \leq i-1$, we let $H_j = G_j$ and $d_{j+1} = c_{j+1}$.
    \item we let $H_i = G_ic_{i+1}G_{i+1}$.
    \item for every $j$ such that $i+1 \leq j \leq n$, we let $d_j = c_{j+1}$ and $H_{j} = G_{j+1}$.
  \end{itemize}
  Finally, we define $\Hb_i = \{f_i(\bar{G}) \mid \bar{G} \in \Gb_i\}$. Observe that since $\Gb_i \in \sfp{\Cs}$, $\{c_{i+1}\}\in\sfp{\Cs}$ and \sfp{\Cs} is closed under concatenation, it is immediate that $\Hb_i \in \sfp{\Cs}$.

  \smallskip

  We may now define the $n$-blueprint $\Lb_\varphi$. We let,
  \[
    \Lb_\varphi = \Bigl(\bigcup_{0 \leq i \leq n+1} \Lb_{i}\Bigr) \cup \Bigl(\bigcup_{0 \leq i \leq n} \Hb_{i}\Bigr).
  \]
  It it clear that $\Lb_\varphi$ is an $n$-blueprint such that $\Lb_\varphi \in \sfp{\Cs}$ since this is the case for all sets $\Lb_i$ and $\Hb_i$ by definition. Hence, it remains to verify that~\eqref{eq:consf} is satisfied. Let $w \in A^*$ and $\mu: \Xs \to \pos{w}$ a correct assignment. We prove that $w,\mu \models \Lb_\varphi \Leftrightarrow w,\mu \models \varphi$. There are two directions.

  First, assume that $w,\mu \models \varphi$. We prove that $w,\mu \models \Lb_\varphi$. Since $\varphi :=$ ``$\exists y\ \psi$'', there exists an assignment $\gamma: \Xs \cup \{y\} \to \pos{w}$ such that $\gamma(x) = \mu(x)$ for all $x \in \Xs$ and $w,\gamma \models \psi$  (note that at this stage, $\gamma$ is not a ``correct assignment'': indeed, this is not well-defined since we have not specified any linear ordering  on $\Xs\cup\{y\}$) yet. We distinguish two cases. First, assume that there exists $i$ satisfying $0 \leq i \leq n+1$ and $\gamma(y) = \gamma(x_i) = \mu(x_i)$. In that case, one may verify from the definitions that $w,\gamma \models \psi$ entails $w,\mu \models \psi_i$. By definition of $\psi_i$, it follows that $w,\mu \models \Lb_i$ and therefore that $w,\mu \models \Lb_\varphi$ as desired, since $\Lb_i \subseteq \Lb_\varphi$.  In the second case, since $\mu(x_0) = 0$ and $\mu(x_{n+1}) = |w|+1$ (by definition of increasing assignments), there exists $i$ such that $0 \leq i \leq n$ and $\mu(x_i) < \gamma(y) < \mu(x_{i+1})$. Hence, $\gamma$ can be viewed as a \emph{correct} assignment $\gamma: \Ys_i \to \pos{w}$. Since $w,\gamma \models \psi$, it follows that $w,\gamma \models \Gb_i$ by definition of $\Gb_i$. One may then verify from the definition of $\Hb_i$ from $\Gb_i$ that $w,\mu \models \Hb_i$. Consequently, $w,\mu \models \Lb_\varphi$ again, since $\Hb_i \subseteq \Lb_\varphi$ by definition.

  We turn to the converse implication. Assume that $w,\mu \models \Lb_\varphi$ . We prove that $w,\mu \models \varphi$. By definition of $\Lb_\varphi$ as the union of smaller $n$-blueprints, there are two cases. In the first case, we assume that there exists $i$ such that $0 \leq i \leq n+1$ and $w,\mu \models \Lb_i$. By definition of $\Lb_i$, it follows that $w,\mu \models \psi_i$. Moreover, by definition of $\psi_i$ from $\psi$, this implies that $w,\gamma \models \exists y\ \psi$ where $\gamma: \Xs \cup \{y\} \to \pos{w}$ is the assignment defined by $\gamma(x) = \mu(x)$ for every $x \in \Xs$ and $\gamma(y) = \mu(x_i)$. Thus, since $\varphi :=$ ``$\exists y\ \psi$'', it follows that $w,\mu \models \varphi$, as desired. In the second case, we assume that there exists $i$ such that $0 \leq i \leq n$ and $w,\mu \models \Hb_{i}$. By definition of $\Hb_i$ from $\Gb_i$, one may verify that there exists a correct assignment $\gamma: \Ys_i \to \pos{w}$ such $\mu(x) = \gamma(x)$ for every $x \in \Xs$ and $w,\gamma \models \Gb_{i}$. By definition of $\Gb_i$, it follows that  $w,\gamma \models \psi$. Since $\varphi :=$ ``$\exists y\ \psi$'', it follows that $w,\mu \models \varphi$, concluding the~proof.
\end{proof}


\section{Second logical characterization: linear temporal logic}
\label{sec:ltl}
We present a second logical characterization of star-free closure. It also generalizes a well-known result concerning the star-free languages: they are exactly those that can be defined in linear temporal logic (\ltl). This is a consequence of Kamp's theorem~\cite{kltl} which implies the equality $\fow=\ltl$. It then follows from the theorem of McNaughton and Papert that $\sfr = \fow = \ltl$. Here, we introduce a generalized definition of linear temporal logic, which is parameterized by a class \Cs. Actually, we define two classes \ltlc{\Cs} and \ltlpc{\Cs}. They generalize the two classical variants of linear temporal logic: without and with past,  respectively. Then, we prove that $\sfp{\Cs} =\ltlc{\Cs} = \ltlpc{\Cs}$ for every \vari \Cs. The proof argument relies heavily on the characterizations of star-free closure that we already presented. The proof that $\ltlpc{\Cs}\subseteq\sfp{\Cs}$ is based on the equality $\sfp{\Cs} = \fo(\infsigc)$ (Theorem~\ref{thm:folog}). Moreover, the proof that $\sfp{\Cs} \subseteq \ltlc{\Cs}$ uses the algebraic characterization of star-free closure (Theorem~\ref{thm:sfcarac}).

\subsection{Preliminaries} We first define the generalized notion of ``linear temporal logic over finite words''.  Then, we present some useful results about it, which we shall need later when proving the correspondence with star-free~closure.

\medskip
\noindent
{\bf Syntax.} For every class \Cs, we define two sets of temporal formulas denoted by \ltla{\Cs} and \ltlpa{\Cs}, which generalize the classical notions of ``linear temporal logic'' and ``linear temporal logic with past''. We first define the \ltlpa{\Cs} formulas, which are more general.

A particular formula is  built from the atomic formulas using Boolean connectives and temporal modalities. The atomic formulas are: $min$, $max,\top$ and $a$ for every letter $a \in A$. Moreover, we allow all Boolean connectives: if $\psi_1$ and $\psi_2$ are \ltlpa{\Cs} formulas, then so are $(\psi_{1} \vee \psi_{2})$, $(\psi_{1} \wedge \psi_{2})$ and $(\neg \psi_1)$. Finally, there are two binary temporal modalities ``\emph{until}'' and ``\emph{since}''. They are parameterized by a language in \Cs. If $\psi_1$ and $\psi_2$ are \ltlpa{\Cs} formulas and $L \in \Cs$, then the following expressions are \ltlpa{\Cs} formulas as well:
\[
  (\untilp{L}{\psi_1}{\psi_2}) \quad\qquad \text{and} \quad\qquad (\sincep{L}{\psi_1}{\psi_2}).
\]
Moreover, we write ``$\textup{U}$'' for ``$\textup{U}_{A^*}$'' and ``$\textup{S}$'' for ``$\textup{S}_{A^*}$'' (their semantics will be the same as the standard modalities ``\emph{until}'' and ``\emph{since}'' in classical linear temporal logic). We also use the following abbreviations: given an \ltlpa{\Cs} formula $\psi$ and a language $L \in \Cs$, we write $\finallyl{\psi}$ for $\untilp{L}{\top}{\psi}$ and $\nex{\psi}$ for $\until{(\neg \top)}{\psi}$.

Finally, an \ltla{\Cs} formula is an \ltlpa{\Cs}  formula that only contains ``\emph{until}'' modalities (\emph{i.e.}, ``\emph{since}'' is disallowed).

\smallskip
\noindent
{\bf Semantics.} In order to evaluate an \ltlpa{\Cs} formula $\varphi$, one needs a word $w \in A^*$ \emph{and} a position $i \in \pos{w}$. We use structural induction to define when the pair \emph{$(w,i)$ satisfies the formula $\varphi$}. We denote this property by $w,i\models \varphi$:
\begin{itemize}
  \item {\bf Atomic formulas:} We always have $w,i\models\top$. For $a \in A$, we have $w,i \models a$ when $a$ is the letter at position $i$ in $w$. Moreover, we have $w,i \models \mathit{min}$ when $i = 0$ (\emph{i.e.}, $i$ is the leftmost unlabeled position) and $w,i \models \mathit{max}$ when $i = |w| + 1$ (\emph{i.e.}, $i$ is the rightmost unlabeled position).
  \item {\bf Disjunction:} $w,i \models \psi_1 \vee \psi_2$ when $w,i \models \psi_1$ or $w,i \models \psi_2$.
  \item {\bf Conjunction:} $w,i \models \psi_1 \wedge \psi_2$ when $w,i \models \psi_1$ and $w,i \models \psi_2$.
  \item {\bf Negation:} $w,i \models \neg \psi$ when $w,i \models \psi$ does not hold.
  \item {\bf Until:} $w,i \models \untilp{L}{\psi_1}{\psi_2}$ when there exists $j \in \pos{w}$ such that $i< j$, $\infix{w}{i}{j} \in L$, and,
        \begin{enumerate}
          \item For every $k \in \pos{w}$ such that $i < k < j$, we have $w,k \models \psi_1$, and,
          \item $w,j \models \psi_2$.
        \end{enumerate}
  \item {\bf Since:} $w,i \models \sincep{L}{\psi_1}{\psi_2}$ when there exists $j \in \pos{w}$ such that $j < i$, $\infix{w}{j}{i} \in L$ and,
        \begin{enumerate}
          \item For every $k \in \pos{w}$ such that $j < k < i$, we have $w,k \models \psi_1$, and,
          \item $w,j \models \psi_2$.
        \end{enumerate}
\end{itemize}
It remains to define what it means for a single word $w$ (without any distinguished position) to satisfy an \ltlpa{\Cs} formula $\varphi$: we evaluate formulas at the \emph{leftmost unlabeled position} of each word. That is, we say that a word \emph{$w \in A^*$ satisfies $\varphi$} and write $w \models\varphi$ if and only if $w,0\models\varphi$. The language \emph{defined by the formula $\varphi$} is $L(\varphi) = \{w \in A^* \mid w \models \varphi\}$.

Finally, we let \ltlc{\Cs} (resp. \ltlpc{\Cs}) be the class consisting of all languages defined by a formula in \ltla{\Cs} (resp. \ltlpa{\Cs}). The classes \ltlc{\stzer} and \ltlpc{\stzer} (where $\stzer = \{\emptyset,A^*\}$ is the trivial \vari) correspond to the classical variants of linear temporal logic from the literature. Typically, the classes \ltlc{\Gs} and \ltlpc{\Gs} associated to some standard group \vari \Gs (such as \md or \abg), are also natural. We present an example using the class \md of modulo~languages.

\begin{exa} \label{exa:ltlex}
  Let $A = \{a,b\}$. The language $(ab)^*$ belongs to $\ltlc{\stzer}$ since it is defined by the \ltla{\stzer} formula $\nex{(a \vee max)} \wedge \bigl(\until{(a \Rightarrow \nex{b}) \wedge (b \Rightarrow \nex{(a \vee max)})}{max}\bigr)$. Moreover $(aa+bb)^* \in \ltlc{\md}$. It is defined by the following  \ltla{\md} formula:
  \[
    (\finallyp{(AA)^*}{max}) \wedge \Bigl(\until{\bigl((\finallyp{(AA)^*A}{max}) \Rightarrow ((a \wedge \nex{a}) \vee (b \wedge \nex{b}))\bigr)}{max}\Bigr).
  \]
\end{exa}

\noindent
{\bf Properties.} We present a few properties of the classes \ltlc{\Cs}, which we shall use later for proving that $\sfp{\Cs} \subseteq \ltlc{\Cs}$ (when \Cs is a \vari). A key point is that the proof involves auxiliary arbitrary alphabets, independent from the alphabet $A$ that we fix at the beginning. When $B$ is such an alphabet, we need to specify what are the languages over $B$ corresponding to the class \ltlc{\Cs} (this is not clear since \Cs will be defined over our fixed alphabet $A$).  We shall do so using morphisms $\eta: B^* \to N$ that we obtain from the class \Cs.

Consider an alphabet $B$ and a morphism $\eta: B^* \to N$ into a finite monoid. An \ltla{\eta} formula is an \ltla{\Ds} formula $\varphi$ where \Ds is the (finite) class  consisting of all languages recognized by $\eta$ (in particular, this means that atomic formulas in $\varphi$ are $min$, $max$, $\top$ and letters from $B$). Additionally, we write \ltlc{\eta} for the class consisting of all languages (over $B$) that can be defined by an \ltla{\eta} formula. We complete these definitions with two lemmas, which are useful to build \ltla{\eta} formulas.

\begin{lem} \label{lem:ltlsuffix}
  Let $B$ be an alphabet, $\eta: B^*\to N$ be a morphism into a finite monoid and $L\in\ltlc{\eta}$. There exists a formula $\varphi \in \ltla{\eta}$ such that for every $w \in B^*$ and $i \in \pos{w}$ we have $w,i \models \varphi$ if and only if\/ $\suffix{w}{i} \in L$.
\end{lem}

\begin{proof}
  By definition, there exists a formula $\varphi \in \ltla{\eta}$ such that $L(\varphi) = L$. Moreover, as $\varphi$ is evaluated at the leftmost unlabeled position in words, we may assume without loss of generality that $\varphi$ is a Boolean combination of formulas of the form $\untilp{H}{\psi_1}{\psi_2}$. Since \ltla{\eta} formulas only contain ``until'' modalities by definition, it can now be verified that for every $w \in B^*$ and $i \in \pos{w}$ we have $w,i \models \varphi$ if and only if $\suffix{w}{i} \in L$.
\end{proof}

\begin{lem} \label{lem:ltlprefix}
  Let $B$ be an alphabet, $\eta: B^*\to N$ be a morphism into a finite monoid, $L\in\ltlc{\eta}$ and $\zeta \in \ltla{\eta}$. There exists a formula $\varphi \in \ltla{\eta}$ such that for all $w\in B^*$ and $i \in \pos{w}$, we have $w,i\models\varphi$ if and only if there exists $j \in \pos{w}$ satisfying the three following conditions:
  \begin{enumerate}
    \item  $i < j$ and for all $k \in \pos{w}$ such that $i < k < j$, we have $w,k \not\models \zeta$,
    \item $w,j\models\zeta$, and
    \item $\infix{w}{i}{j} \in L$.
  \end{enumerate}
\end{lem}

\begin{proof}
  Lemma~\ref{lem:ltlsuffix} yields a formula $\psi \in \ltla{\eta}$ such that for every $w \in B^*$ and $i \in \pos{w}$ we have $w,i \models \psi$ if and only if $\suffix{w}{i} \in L$. We build  a new formula $\psi'$ by applying the two following modifications to $\psi$:
  \begin{itemize}
    \item we replace every occurrence of the atomic formula $max$ by $\zeta$.
    \item we recursively replace every sub-formula of the form $\untilp{H}{\psi_1}{\psi_2}$ by  $(\neg \zeta) \wedge (\untilp{H}{(\psi_1 \wedge \neg \zeta)}{\psi_2})$.
  \end{itemize}
  It can now be verified that the formula $\varphi := (\finally{\zeta}) \wedge \psi'$ satisfies the property described in the lemma.
\end{proof}

\subsection{Main Theorem}

It is well-known that we have $\sfr = \fow = \ltlc{\stzer} = \ltlpc{\stzer}$. The equality $\fow = \ltlc{\stzer} =\ltlpc{\stzer}$ follows from Kamp's theorem~\cite{kltl} (these equalities are only an instance of Kamp's theorem, which is more general, as itconnects first-order logic to linear temporal logic for more general structures than finite words). Then, the equality $\sfr = \fow$ follows from the work of McNaughton and Papert~\cite{mnpfosf}. Here, we generalize this result to arbitrary input classes that are \varis. More precisely, we prove the following theorem.

\begin{thm} \label{thm:ltlfo}
  Let \Cs be a \vari. Then, $\sfp{\Cs} = \ltlc{\Cs} = \ltlpc{\Cs}$.
\end{thm}

Note that, in view of Corollary~\ref{cor:sfcarac}, Theorem~\ref{thm:ltlfo} implies that for every \vari \Cs with decidable separation, the class $\ltlc{\Cs} = \ltlpc{\Cs}$ has decidable membership. We now concentrate on the proof of Theorem~\ref{thm:ltlfo}. 

\begin{proof}[Proof of Theorem~\ref{thm:ltlfo}]
  We fix a \vari \Cs for the proof. The inclusion $\ltlc{\Cs}\subseteq\ltlpc{\Cs}$ is trivial. We prove that $\ltlpc{\Cs}\subseteq\sfp{\Cs}$ and $\sfp{\Cs}\subseteq\ltlc{\Cs}$. Let us start with the former, which is simpler.

  \medskip
  \noindent
  {\bf Inclusion} $\ltlpc{\Cs} \subseteq \sfp{\Cs}$. The proof is based on Theorem~\ref{thm:folog}. It states the equality $\sfp{\Cs}=\foc$. Thus, it suffices to prove that $\ltlpc{\Cs} \subseteq \foc$. Let $\varphi$ be an \ltlpa{\Cs} formula. We use structural induction on $\varphi$ to build an \foc formula $[\varphi](x)$ with at most one free variable $x$ which satisfies the following condition:
  \begin{equation}\label{eq:ltltofo}
    \text{for all $w \in A^*$ and $i \in \pos{w}$,} \quad  w,i \models \varphi \quad \text{if and only if} \quad w,x \mapsto i \models [\varphi](x),
  \end{equation}
  where $x\mapsto i$ denotes the assignment that maps $x$ to $i$.
  It will then be immediate that the language $L \subseteq A^*$ defined by the \ltla{\Cs} formula $\varphi$ (\emph{i.e.}, $L = \{w \in A^* \mid w \models \varphi\}$) is defined by the \foc sentence $[\varphi](min)$ (which is obtained from $[\varphi](x)$ by replacing every free occurrence of the variable $x$ by the constant $min$). Hence, we obtain $\ltlpc{\Cs} \subseteq \foc$ as desired.

  If  $\varphi = \top$, then we let $[\varphi](x) := (\mathit{min} = \mathit{min})$. If $\varphi = min $ or $\varphi = max$, then we define $[\varphi](x) := (x = min)$ or $[\varphi](x) := (x = max)$, respectively. If $\varphi = a$ for some $a \in A$, then we let $[\varphi](x) := a(x)$. Logical connectives are handled in the usual way. It remains to treat the temporal operators \emph{until} and \emph{since}.
  \begin{enumerate}
    \item If $\varphi = \untilp{L}{\varphi_1}{\varphi_2}$ for some $L \in \Cs$, we define:
          \[
          [\varphi](x) := \exists x_2\ I_L(x,x_2) \wedge [\varphi_2](x_2) \wedge \forall x_1\ (x < x_1 \wedge x_1 < x_2) \Rightarrow [\varphi_1](x_1).
          \]
    \item If $\varphi = \sincep{L}{\varphi_1}{\varphi_2}$ for some $L \in \Cs$, we define:
          \[
          [\varphi](x) := \exists x_2\ I_L(x_2,x) \wedge [\varphi_2](x_2) \wedge \forall x_1\ (x_2 < x_1 \wedge x_1 < x) \Rightarrow [\varphi_1](x_1).
          \]
  \end{enumerate}
  It is simple to verify that this construction satisfies~\eqref{eq:ltltofo}, as desired.

  \medskip
  \noindent
  {\bf Inclusion} $\sfp{\Cs} \subseteq \ltlc{\Cs}$. Let $L \in \sfp{\Cs}$. We prove that $L \in \ltlc{\Cs}$. The argument is based on Theorem~\ref{thm:sfcarac} which implies that all the \Cs-orbits for the syntactic morphism $\alpha: A^*\to M$ of $L$ are aperiodic. We use this hypothesis and induction to construct an \ltla{\Cs} formula defining $L$. This implies that $L \in \ltlc{\Cs}$, as desired.

  The construction borrows ideas from the argument of Theorem~\ref{thm:sfcarac}, which proves that under the same hypotheses on $L$, we have $L \in \bsdp{\Cs}$. Yet, there are key differences, as \ltlc{\Cs} and \bsdp{\Cs} are distinct formalisms. In particular, as mentioned above, we shall consider auxiliary alphabets independent from $A$. We start with preliminary definitions aimed at manipulating them.

  First, each time we consider an auxiliary alphabet $B$, we shall have to recast the morphism $\alpha: A^* \to M$ into a morphism $\beta: B^* \to M$ and reformulate on $\beta$ the hypothesis that all~\Cs-orbits for $\alpha$ are aperiodic. For this, we consider another morphism $\eta: B^*\to N$ into a finite monoid. Roughly, $\eta$ is used as an abstraction of the class \Cs over the alphabet $B$. 
  We say that the pair $(\beta,\eta)$ is \emph{tame} to indicate that the~following property holds:
  \begin{equation} \label{eq:goodm}
    \text{For all $u,v\in B^*$, if $\eta(u)= \eta(v)$ and $\beta(u) \in E(M)$, then $(\beta(uvu))^{\omega} = (\beta(uvu))^{\omega+1}$.}
  \end{equation}
  We first connect this definition to our hypothesis in the following simple fact.

  \begin{fct} \label{fct:ptame}
    There exists a \Cs-morphism $\eta: A^* \to N$ such that the pair $(\alpha,\eta)$ is tame.
  \end{fct}

  \begin{proof}
    Lemma~\ref{lem:cmorph} yields a \Cs-morphism $\eta: A^* \to N$ such that for every $u,v \in A^*$, if $\eta(u) = \eta(v)$, then  $(\alpha(u),\alpha(v))$ is a \Cs-pair. Since all \Cs-orbits for $\alpha$ are aperiodic by hypothesis, it follows that if we additionally know that $\alpha(u) \in E(M)$, then $(\alpha(uvu))^{\omega} = (\alpha(uvu))^{\omega+1}$. Hence, \eqref{eq:goodm} holds and $(\alpha,\eta)$ is tame.
  \end{proof}

  Given an alphabet $B$, a morphism $\eta: B^* \to N$ into a finite monoid and $P \subseteq B^*$, an \ltlc{\eta}-partition of $P$ is a \emph{finite} partition \Kb of $P$ into languages of \ltlc{\eta}. Moreover, given a morphism $\beta: B^* \to M$ (where $M$ is the original finite monoid used in $\alpha$) and $s \in M$, we say that \Kb is \emph{$(\eta,\beta,s)$-safe} to indicate that for every $K\in\Kb$ and every $w,w' \in K$, we have $\eta(w) = \eta(w')$ and $\beta(w)s=\beta(w')s$. We may now start the proof. The argument is based on the following lemma.

  \begin{lem}\label{lem:ltlmain}
    Let $B$ be an alphabet and consider a morphism $\beta: B^* \to M$ into the fixed monoid $M$ and another morphism $\eta: B^* \to N$ into a finite monoid such that $(\beta,\eta)$ is tame. Let $C \subseteq B$ and $s \in M$. Then, there exists an  $(\eta,\beta,s)$-safe \ltlc{\eta}-partition of\/ $C^*$.
  \end{lem}

  Let us first apply Lemma~\ref{lem:ltlmain} to prove that $L \in \ltlc{\Cs}$. Fact~\ref{fct:ptame} yields a \Cs-morphism $\eta: A^* \to N$ such that the pair $(\alpha,\eta)$ is tame. Since $\eta$ is a \Cs-morphism, $\ltlc{\eta} \subseteq \ltlc{\Cs}$. Hence, it suffices to prove that $L \in \ltlc{\eta}$. We apply the lemma for $B = C = A$, $\beta = \alpha$ and $s = 1_M$. This yields an  $(\eta,\alpha,1_M)$-safe \ltlc{\eta}-partition \Kb of $A^*$. Now, \Kb being $(\eta,\alpha,1_M)$-safe implies that for every $K \in \Kb$, there exists $s\in M$ such that $K \subseteq \alpha\inv(s)$. Since $\Kb$ is a partition of $A^*$ and $L$ is recognized by $\alpha$, it follows that $L$ is a union of languages of \Kb. Since \ltlc{\eta} is closed under union, it follows that  $L$ itself belongs to $\ltlc{\eta}$, which completes the main argument.

  \medskip

  It remains to prove Lemma~\ref{lem:ltlmain}. Let $B$ be an alphabet and consider two morphisms $\beta: B^*\to M$ and $\eta: B^*\to N$ such that $(\beta,\eta)$ is tame. Moreover, let $C \subseteq B$ and $s \in M$. We build an \ltlc{\eta}-partition of $C^*$ which is $(\eta,\beta,s)$-safe using induction on the three following parameters listed by order of importance:
  \begin{enumerate}
    \item The size of $\beta(C^+) \subseteq M$.
    \item The size of $C$.
    \item The size of $\beta(C^*) \cdot s \subseteq M$.
  \end{enumerate}

  \begin{remark}
    As already mentioned, the proof is similar to that of Theorem~\ref{thm:sfcarac}. In particular, the current proof resembles to that of Corollary~\ref{cor:sfcaracg}. The reader may wonder why the element $s$, which serves as a buffer in these proofs, is not on the same side. The reason is that it was easier to consider $s\alpha(P^{*})$ in Corollary~\ref{cor:sfcaracg}, due to the fact that $P$ is a prefix code, while it is easier to consider $\beta(C^*) \cdot s $ here, due to the fact that we are dealing with pure future linear time temporal logic.
  \end{remark}

  We distinguish two cases depending on the following property of $\beta$, $C$ and $s$. We say that $s$ is \emph{$(\beta,C)$-stable} when the following holds:
  \begin{equation} \label{eq:ltlstable}
    \text{for every $c \in C$,} \quad  \beta(C^*) \cdot s = \beta(cC^*) \cdot s.
  \end{equation}
  We first consider the case when $s$ is $(\beta,C)$-stable. This is the base case. Otherwise, we use induction on our three parameters.

  \smallskip
  \noindent
  {\bf Base case: $s$ is $(\beta,C)$-stable.} In that case, we define $\Kb = \{C^*\cap\eta\inv(t) \mid t \in N\}$. Clearly, this is a finite partition of $C^*$. Moreover, it is clear that $C^* \cap \eta\inv(t) \in \ltlc{\eta}$ for every $t \in N$. Indeed, it is defined by the \ltla{\eta} formula $\left(\untilp{\eta\inv(t)}{\left(\bigvee_{c \in C}c\right)}{max}\right)$.

  It remains to show that \Kb is $(\eta,\beta,s)$-safe. We use the hypothesis that $(\beta,\eta)$ is tame. First, we use the hypothesis that $s$ is $(\beta,C)$-stable to prove the following statement, analogous to Fact~\ref{fct:sfclos:icarbase}.

  \begin{fct}\label{fct:ltlbase}
    Let $q,f \in \beta(C^*)$ such that $f$ is idempotent. Then, we have $fqs = qs$.
  \end{fct}

  \begin{proof}
    The proof is based on the following preliminary result. For $u,v \in C^*$, we show that,
    \begin{equation}\label{eq:ltlprelim}
      \text{there exists $r\in \beta(C^*)$ such that $\beta(u)rs = \beta(v)s$.}
    \end{equation}
    We fix $u,v \in C^*$ for the proof of~\eqref{eq:ltlprelim}. We use induction on the length of $u$. If $u =\veps$, it suffices to choose $r = \beta(v) \in \beta(C^*)$. Otherwise, $u = u'c$ with $u' \in C^*$, $c \in C$. Induction yields $r' \in \beta(C^*)$ such that $\beta(u')r's= \beta(v)s$. Moreover, since $s$ is $(\beta,C)$-stable and $r' \in \beta(C^*)$, it follows from~\eqref{eq:ltlstable} that there exists $r \in \beta(C^*)$ such that $r's = \beta(c)rs$. Altogether, this yields $\beta(u')\beta(c)rs = \beta(v)s$ and as $u = u'c$, we get $\beta(u)rs = \beta(v)s$, concluding the proof of~\eqref{eq:ltlprelim}.

    We now prove the fact. Let $q,f \in \beta(C^*)$ such that $f$ is idempotent. By definition, we $u,v \in C^*$ such that $q = \beta(v)$ and $f = \beta(u)$. Hence,~\eqref{eq:ltlprelim} yields $r \in \alpha(C^*)$ such that $frs = qs$. Since $f$ is idempotent, this implies that $fqs = ffrs = frs = qs$.
  \end{proof}

  We now prove that every $K \in \Kb$ is $(\eta,\beta,s)$-safe. By definition, $K = C^* \cap \eta\inv(t)$ for $t \in N$. Given $u,v\in K$, we have to show that $\eta(u) = \eta(v)$ and  $\beta(u)s=\beta(v)s$. Let $n = \omega(M)$. Since $u,v \in K$, we have $\eta(u) = \eta(v) = t$. Hence, $\eta(u^n) = \eta(u^{n-1}v)$ and since $\beta(u^n)$ is idempotent, the hypothesis that $(\beta,\eta)$ is tame yields $(\beta(u^{2n-1}vu^n))^n = (\beta(u^{2n-1}vu^n))^{n+1}$. We now multiply by $s$ on the right to obtain $(\beta(u^{2n-1}vu^n))^n s = (\beta(u^{2n-1}vu^n))^{n+1}s$. Since $n = \omega(M)$, we know that $(\beta(u^{2n-1}vu^n))^n$ is an idempotent of $\beta(C^*)$. Therefore, Fact~\ref{fct:ltlbase} yields $(\beta(u^{2n-1}vu^n))^n s=s$. Altogether, we obtain that $s = \beta(u^{2n-1}vu^n) s$. We now multiply by $\beta(u)$ on the left to get  $\beta(u)s = \beta(u^{2n}vu^n) s$. Finally, $\beta(u^n) \in \beta(C^*)$ is an idempotent, we may apply Fact~\ref{fct:ltlbase} twice to get $\beta(u^{2n}vu^n) s = \beta(v) s$. Altogether, this yields $\beta(u)s=\beta(v)s$, as desired.

  \medskip
  \noindent
  {\bf Inductive case: $s$ is not $(\beta,C)$-stable.} By hypothesis, there exists some letter $c \in C$ such that the following property holds:
  \begin{equation}\label{eq:fo:godinduc}
    \beta(cC^*) \cdot s \subsetneq \beta(C^*) \cdot s.
  \end{equation}
  We fix this letter $c \in C$ for the rest of the argument and we let $D$ be the sub-alphabet $D = C \setminus \{c\}$.

  The restrictions $\beta:D^*\to M$ and $\eta:D^*\to N$ still form a tame pair~$(\beta,\eta)$. Therefore, we may apply induction in Lemma~\ref{lem:ltlmain} when replacing $C$ by $D$.
  Indeed, the first induction parameter (the size of $\beta(C^+)$) has not increased here since $D \subseteq C$ and $\beta$ remains unchanged, while the second parameter has decreased: $|D|<|C|$. This yields an $(\eta,\beta,1_M)$-safe \ltlc{\eta}-partition \Hb of $D^*$. We may assume without loss of generality that $H \neq \emptyset$ for every $H \in \Hb$.

  \begin{fct}\label{fct:sfclos:alphind2}
    There exists an $(\eta,\beta,1_M)$-safe \ltlc{\eta}-partition \Hb of  $D^*$ made of nonempty languages.
  \end{fct}

  We distinguish two independent subcases. Observe that the inclusion $\beta(cC^*) \subseteq \beta(C^+)$ holds. The argument differs depending on whether it is strict or not.

  \medskip
  \noindent
  {\bf Subcase~1: $\beta(cC^*) = \beta(C^+)$.} We use induction on our third parameter (\emph{i.e.}, the size of $\beta(C^*)s$). Let $H \in \Hb$. Since \Hb is a partition of $D^*$ which is $(\eta,\beta,1_M)$-safe by definition and $H \neq \emptyset$, there exists a unique element $t_H \in \beta(D^*)$ such that $\beta(x) = t_H$ for every $x \in H$. The construction of \Kb is based on the following fact (this is where we use induction).

  \begin{fct}\label{fct:ltlc1carac}
    For all $H \in \Hb$, there exists an $(\eta,\beta(c)t_H s)$-safe \ltlc{\eta}-partition $\Ub_H$ of\/ $C^*$.
  \end{fct}

  \begin{proof}
    We fix $H \in \Hb$. Since $t_H \in \beta(D^*)$, we have $\beta(c)t_H s \in \beta(cD^*)s$. Therefore, we have $\beta(C^*)\beta(c)t_H s \subseteq \beta(C^+)s$. Combined with our hypothesis in Subcase~1 (\emph{i.e.}, $\beta(cC^*) = \beta(C^+)$), this yields $\beta(C^*)\beta(c)t_H s \subseteq \beta(cC^*)s$. Finally, we obtain from~\eqref{eq:sfclos:godinduc} (\emph{i.e.}, $\beta(cC^*)s \subsetneq \beta(C^*)s$) that the \emph{strict} inclusion $\beta(C^*)\beta(c)t_H s \subsetneq \beta(C^*)s$ holds. Hence, we may apply induction on our third parameter in Lemma~\ref{lem:ltlmain} (\emph{i.e.}, the size of $\beta(C^*)s$) to obtain the desired finite partition $\Ub_H$ of $C^*$ which is $(\eta,\beta(c)t_H s)$-safe. Note that here, our first two parameters have not increased as $\beta$ and $C$ remain unchanged.
  \end{proof}

  We may now define the desired partition \Kb of $C^*$. Using the partitions $\Ub_H$ given by Fact~\ref{fct:ltlc1carac}, we define,
  \[
    \Kb = \Hb \cup \{UcH \mid H \in \Hb \text{ and } U \in \Ub_H\}.
  \]
  It remains to show that \Kb is indeed an \ltlc{\eta}-partition of $C^*$ which is $(\eta,\beta,s)$-safe. One may verify that \Kb is a partition of $C^*$ since \Hb is a partition of $D^*$ and $\Ub_H$ is a partition of $C^*$ for every $H \in \Hb$ (recall that $D = C \setminus \{c\}$). Let us prove that every $K \in \Kb$ belongs to \ltlc{\eta}. This is immediate if $K \in \Hb$ by hypothesis on \Hb. Otherwise, there exist $H \in \Hb$ and $U \in \Ub_H$ such that $K = UcH$. Since $H \in \ltlc{\eta}$, it follows from Lemma~\ref{lem:ltlsuffix} that there exists a formula $\psi_H \in \ltla{\eta}$ such that for every $w \in B^*$ and every $i \in \pos{w}$, we have $w,i \models \psi_H$ if and only if $\suffix{w}{i} \in H$. Moreover, let $\zeta \in \ltla{\eta}$ be the formula $\zeta := c \wedge \neg \finally{c}$ (given a word $w \in B^*$ and $i \in \pos{w}$, we have $w,i \models \zeta$ if and only if $i$ is the rightmost position in $w$ carrying the letter $c$). Since $U \in \ltlc{\eta}$ (by definition of $\Ub_H$), Lemma~\ref{lem:ltlprefix} yields a formula $\psi_U \in \ltla{\eta}$ such that for every $w \in B^*$, we have $w \models \psi_U$ if and only if there exists $j \in \pos{w}$ such that $w,j \models \zeta$ (by definition of $\zeta$, $j$ must be unique) and $\prefix{w}{j} \in U$. Since $c \not\in D$ and $H \subseteq D^*$, one may now verify that $K = UcH$ is defined by the formula $\psi_U \wedge \finally{(\zeta \wedge \psi_H)}$. Hence, we get $K \in \ltlc{\eta}$, as desired.

  It remains to prove that \Kb is $(\eta,\beta,s)$-safe. Consider $K \in \Kb$ and $w,w' \in K$. We have to show that $\eta(w)=\eta(w')$ and $\beta(w)s= \beta(w')s$. By definition of $\Kb$, there are two cases: first, if $K\in\Hb$ we know by Fact~\ref{fct:sfclos:alphind2} that $\Hb$ is an $(\eta,\beta,1_M)$-safe \ltlc{\eta}-partition of $D^*$. Therefore, we obtain $\eta(w)=\eta(w')$ and $\beta(w) = \beta(w')$, whence $\beta(w)s = \beta(w')s$, as desired. Otherwise, $K=UcH$ with $H \in \Hb$ and $U \in \Ub_H$. Thus, we get $x,x' \in H$ and $u,u' \in U$ such that $w = ucx$ and $w' = u'cx'$. By definition of $t_H$, we have $\beta(x) = \beta(x') = t_H$. Moreover, since $\Hb$ is $(\eta,\beta,1_M)$-safe and $\Ub_{H}$ is $(\eta, \beta(c)t_{H}s)$-safe by Fact~\ref{fct:ltlc1carac}, we also have $\eta(x)=\eta(x')$ and $\eta(u)=\eta(u')$, whence $\eta(w)=\eta(w')$. Finally, $\beta(w)s = \beta(u)\beta(c)t_H s$ and $\beta(w')s = \beta(u')\beta(c)t_H s$ and since $\Ub_H$ is $(\eta,\beta(c)t_H s)$-safe, we obtain $\beta(w)s = \beta(w')s$. This concludes the proof of this subcase.

  \smallskip
  \noindent
  {\bf Subcase~2: $\beta(cC^*) \subsetneq \beta(C^+)$.} We use induction on our first parameter (\emph{i.e.}, the size of $\beta(C^+)$). Consider a word $w \in C^*$. Since $D = C \setminus \{c\}$, $w$ admits a unique decomposition $w = uv$ such that $u \in D^*$ and $v \in (cD^*)^*$ (\emph{i.e.}, $u$ is the largest prefix of $w$ in $D^*$ and $v$ is the corresponding suffix). Using induction, we construct \ltlc{\eta}-partitions of the possible prefixes and suffixes. Then, we combine them to construct a partition of the whole set $C^*$. Actually, not that we already partitioned the set of prefixes: we have an \ltlc{\eta}-partition \Hb of $D^*$ which is $(\eta,\beta,1_M)$-safe. It remains to partition the set of suffixes:  this is where we use induction.

  \begin{lem} \label{lem:sc2carac}
    There exists an \ltlc{\eta}-partition \Vb of $(cD^*)^*$ which is $(\eta,\beta,1_M)$-safe.
  \end{lem}

  \begin{proof}
    For each language $H \in \Hb$, we create a letter written $\frb_H$ and let $\frB = \{\frb_H \mid H \in \Hb\}$ as a new alphabet. Moreover, we define new morphisms $\gamma: \frB^* \to M$ and $\delta: \frB^* \to N$. Let $H \in \Hb$ and consider the letter $\frb_H$. Since \Hb is a partition of $D^*$ which is $(\eta,\beta,1_M)$-safe by definition and $H \neq \emptyset$, there exist unique elements $t_H \in \beta(D^*)$ and $q_H \in \eta(D^*)$ such that $\eta(x) = q_H$ and $\beta(x) = t_H$ for every $x \in H$. We let $\gamma(\frb_H) =\beta(c)t_H$ and $\delta(\frb_H)=\eta(c)q_H$.

    Observe that the pair $(\gamma,\delta)$ is tame. Indeed, let $u,v \in \frB^*$ such that $\delta(u)= \delta(v)$ and $\gamma(u) \in E(M)$. By definition of \frB, there exist $w_u,w_v \in (cD^*)^*$ such that $\gamma(u) = \beta(w_u)$, $\gamma(v) = \beta(w_v)$, $\delta(u) = \eta(w_u)$ and $\delta(v) = \eta(w_v)$. Hence, $\eta(w_u) = \eta(w_v)$ and $\beta(w_u) \in E(M)$. Since $(\beta,\eta)$ is tame, it then follows from~\eqref{eq:goodm} that $(\beta(w_uw_vw_u))^{\omega} = (\beta(w_uw_vw_u))^{\omega+1}$. This exactly says that $(\gamma(uvu))^{\omega} = (\gamma(uvu))^{\omega+1}$, as desired. Moreover, by definition of $\frB$, one may verify that $\gamma(\frB^+) = \beta((cD)^+) \subseteq \beta(cC^*)$. Hence, since  $\beta(cC^*) \subsetneq \beta(C^+)$ (this is our hypothesis in Subcase~2), we get $|\gamma(\frB^+)| < |\beta(C^+)$. Consequently, we may apply induction on the first parameter in Lemma~\ref{lem:ltlmain} (\emph{i.e.}, the size of $\beta(C^+)$) to get an \ltlc{\delta}-partition \Gb of $\frB^*$ which is $(\delta,\gamma,1_M)$-safe. We use it to construct \Vb.

    First, we define a map $\mu: (cD^*)^* \to \frB^*$. Observe that since $c \not\in D$, every word $u \in (cD^*)^*$ admits a unique decomposition $u = cu_1 \cdots cu_n$ with $u_1,\dots u_n\in D^*$. For every $i \leq n$, we let $H_i$ as the unique language in \Hb such that $u_i \in H_i$ (recall that \Hb is partition of $D^*$). We then define $\mu(w) = \frb_{H_1} \cdots \frb_{H_n}$. Note that by definition, each position $i \in \pos{w}$ which is labeled by a ``$c$'' corresponds to a unique position in $\mu(w)$. We may now define $\Vb = \{\mu\inv(G) \mid G \in \Gb\}$. It remains to show that \Vb is an \ltlc{\eta}-partition of $(cD^*)^*$ which is $(\eta,\beta,1_M)$-safe. Clearly,  \Vb is a partition of $(cD^*)^*$ by definition since \Gb is a partition of $\frB^*$.

    We first prove that every $V \in \Vb$ belongs to \ltlc{\eta}. By definition, $V=\mu\inv(G)$ for some $G\in\Gb$. Let $\zeta := c \vee max \in \ltla{\eta}$. We know that every $H \in \Hb$ belongs to \ltlc{\eta}. Hence, Lemma~\ref{lem:ltlprefix} yields a formula $\psi'_H \in \ltla{\eta}$ such that for every $w \in B^*$ and $i \in \pos{w}$, we have $w,i \models \psi'_H$ if and only if there exists $j \in \pos{w}$ such that $i < j$, $w,j \models \zeta$, $w,k \not\models \zeta$ for every $i < k < j$ and $\infix{w}{i}{j}\in H$. We let $\psi_H = c \wedge  \psi'_H$. By definition, $w,i \models \psi'_H$ if and only if $i$ has label $c$ and the greatest prefix of \suffix{w}{i} which is in $D^*$ belongs to $H$. The key idea is that when $w \in (cD^*)^*$, the formula $\psi_H$ holds for the positions $i \in \pos{w}$ which are labeled by $c$ and such that the position of $\mu(w) \in \frB^*$ corresponding to $i$ is labeled by $\frb_H \in \frB$. Moreover, since $G \in \Gb$, there exists an \ltla{\delta} formula $\Gamma_G$ defining $G$ by hypothesis on \Gb. We modify $\Gamma_G$ into an $\ltla{\eta}$ formula $\varphi_V$ defining $V = \mu\inv(G)$. First, let $\varphi'_V$ be the formula obtained from $\Gamma_G$ by applying the two following modifications:
    \begin{enumerate}
      \item We replace each atomic sub-formula ``$\frb_H$'' for $H \in \Hb$ by the \ltla{\eta} formula $\psi_H$.
      \item We recursively replace all sub-formulas $\untilp{X}{\varphi_1}{\varphi_2}$. Since $\Gamma_G$ is an \ltla{\delta}-formula, we have $F \subseteq N$ such that $X = \delta\inv(F)$. We recursively replace $\untilp{X}{\varphi_1}{\varphi_2}$ by,
            \[
            \begin{array}{ll}
              & \left(\nex{\zeta} \wedge \untilp{\eta\inv(F)}{(\zeta \Rightarrow \varphi_1)}{(\zeta \wedge \varphi_2)}\right) \\
              \vee & \left(\left(\neg \nex{\zeta}\right) \wedge \left(\until{\left(\neg\nex{\zeta}\right)}{\left(\nex{\zeta} \wedge \left(\untilp{\eta\inv(F)}{(\zeta \Rightarrow \varphi_1)}{(\zeta \wedge \varphi_2)}\right) \right)}\right)\right).
            \end{array}
            \]

    \end{enumerate}
    Finally, we let $\varphi_V := \nex{(c \vee max)} \wedge \varphi'_V$. One may now verify from the definition that $\varphi_V$ defines $V = \mu\inv(G)$.

    It remains to prove that \Vb is $(\eta,\beta,1_M)$-safe. Let $V \in \Vb$ and $v,v' \in V$. By definition, there exists $G \in \Gb$ such that $V = \mu\inv(G)$. Hence, we have $\mu(v),\mu(v') \in G$ and since \Gb is $(\delta,\gamma,1_M)$-safe, this yields $\gamma(\mu(v)) = \gamma(\mu(v'))$ and $\delta(\mu(v)) = \delta(\mu(v'))$. One may verify that $\delta(\mu(v)) = \eta(v)$, $\delta(\mu(v')) = \eta(v')$, $\gamma(\mu(v)) = \beta(v)$ and $\gamma(\mu(v')) = \beta(v')$. Thus, we get $\beta(v) = \beta(v')$ and $\delta(v) = \delta(v')$, concluding the proof.
  \end{proof}

  We are ready to build our \ltlc{\eta}-partition \Kb of $C^*$. Let \Vb be the $(\eta,\beta,1_M)$-safe \ltlc{\eta}-partition of $(cD^*)^*$ given by Lemma~\ref{lem:sc2carac}. We let,
  \[
    \Kb = \{HV \mid H \in \Hb \text{ and } V \in \Vb\}.
  \]
  It is immediate by definition that \Kb is a partition of $C^*$ since $D = C \setminus \{c\}$ and \Hb is a partition of~$D^*$. Let us verify that every $K\in \Kb$ belongs to \ltlc{\eta}. By definition, one can write $K=HV$ for some $H\in\Hb$ and $V \in \Vb$. Let $\zeta:= c \vee max$. Since $H \in \ltlc{\eta}$, Lemma~\ref{lem:ltlprefix} yields an $\ltla{\eta}$ formula $\psi_H$ such that for every $w \in B^*$, we have $w \models \psi_H$ if and only if there exists $j \in \pos{w} \setminus \{0\}$ such that $w,j \models \zeta$, $w,k \not\models \zeta$ for every $k \in \pos{w}$ such that $0 < k < j$ and $\prefix{w}{j} \in H$. Since $V \in \ltlc{\eta}$ (by hypothesis on \Vb), Lemma~\ref{lem:ltlsuffix} yields an $\ltla{\eta}$ formula $\psi_V$ such that for every $w \in B^*$ and every $i \in \pos{w}$, we have $w,i \models \psi_V$ if and only if $\suffix{w}{i} \in V$. One may now verify that $K = HV$ is defined by the formula,
  \[
    \psi_H \wedge \left(\left(\nex{\zeta} \wedge \psi_V\right) \vee \left(\left(\neg \nex{\zeta}\right) \wedge \left(\until{\left(\neg \nex{\zeta}\right)}{\left(\nex{\zeta} \wedge \psi_V\right)}\right)\right)\right).
  \]
  Hence, we get $K \in \ltlc{\eta}$. It remains to verify that \Kb is $(\eta,\beta,s)$-safe (it is in fact $(\eta,\beta,1_M)$-safe). Let $K \in \Kb$ and $w,w' \in K$, we show that $\eta(w) = \eta(w')$ and $\beta(w) = \beta(w')$ (which implies $\beta(w)s = \beta(w')s$). By definition, $K = HV$ with $H \in \Hb$ and $V \in \Vb$. Therefore, $w =uv$ and $w' = u'v'$ with $u,u' \in H$ and $v,v' \in V$. Since $\Hb$ and $\Vb$ are both $(\eta,\beta,1_M)$-safe by definition, we have $\eta(u) = \eta(u')$, $\beta(u) = \beta(u')$, $\eta(v) = \eta(v')$ and $\beta(v) = \beta(v')$. It follows that $\eta(w) = \eta(w')$ and $\beta(w) = \beta(w')$ which concludes the proof.
\end{proof}


\section{\Ratms}
\label{sec:ratms}
We now turn to separation and covering. We prove two results in the paper. In Section~\ref{sec:finite}, we show that \sfp{\Cs}-covering is decidable for every \emph{finite} \vari \Cs. Then, in Section~\ref{sec:group}, we prove that \sfp{\Gs}-covering is decidable for every group \vari \Gs that has decidable separation. In both cases, the algorithms are based on a generic framework which was introduced in~\cite{pzcovering2} for the specific purpose of handling separation and covering. It relies on simple algebraic objects called \emph{\ratms}. We recall this framework in this preliminary section.

We define \ratms and present two particular kinds: the \nice and the \tame ones. We use this notion to associate a computational problem with each lattice~\Cs: ``given as input a \nice \mratm $\rho$ and a regular language $L$, compute an optimal \Cs-cover of $L$ for $\rho$''. Then, we connect this problem to \Cs-covering. Finally, we present new notions that are not defined in~\cite{pzcovering2}. They are specifically designed for handling the classes of the form $Op(\Cs)$  built from an input class~\Cs using an operator. In the paper, we are interested in the case when $Op$ is star-free closure.

\subsection{Definition}

We first introduce \emph{\ratas}. A \emph{\rata} is a monoid $(R,+)$ which is \emph{commutative} ($q + r = r+ q$ for every $q,r \in R$) and \emph{idempotent} ($r+r=r$ for all $r \in R$). The binary operation $+$ is called \emph{addition} and we denote the neutral element of $R$ by $0_R$ (we use an additive notation here, since we are dealing with a commutative monoid). Given a \rata $R$, we also define a canonical ordering ``$\leq$'' over $R$ as follows:
\[\text{for all }
r, s\in R,\quad r\leq s \text{ when } r+s=s.
\]
One may verify that $\leq$ is a partial order and that it makes $R$ an ordered monoid (\emph{i.e.}, $\leq$ is compatible with addition).  It can be verified that every morphism between two \ratas is \emph{increasing} for the canonical orderings. We often use this property implicitly.

\begin{exa}
	For every set $S$, the algebra $(2^S,\cup)$ is a \rata whose neutral element is $\emptyset$. The canonical ordering ``$\leq$'' on $2^S$ is inclusion. Indeed, given $P,Q \in 2^S$, it is clear that $P \subseteq Q$ if and only if $P \cup Q = Q$. In particular, if $A$ is an alphabet, then $(2^{A^*},\cup)$ is a \rata. In practice, these are the only \emph{infinite} \ratas that we shall consider.
\end{exa}

\label{downset}%
We often apply a ``downset operator'' to subsets of our \ratas $R$. That is, for every $S \subseteq R$, we write $\dclosr S$ for the set $\dclosr S = \bigl\{r \mid \exists s \in S \text{ such that } r \leq s\bigr\}$. We also consider Cartesian products $X \times R$ of an arbitrary set $X$ with a \rata $R$. Given $S \subseteq X \times R$, we write $\dclosr S = \{(x,r) \in X \times R \mid \text{$\exists r' \in R$ such that $r \leq r'$ and $(x,r')\in S$}\}$.

\medskip
\noindent
{\bf Definition of a \ratm.} As seen above, $(2^{A^*},\cup)$ is a \rata. A \emph{\ratm} (over $A$) is a monoid morphism $\rho: (2^{A^*},\cup) \to (R,+)$ where $(R,+)$ is an arbitrary \emph{finite} \rata. In other words, we have $\rho(\emptyset) = 0_R$ and $\rho(K_1\cup K_2)=\rho(K_1)+\rho(K_2)$ for all $K_1,K_2 \subseteq A^*$. Note that since \ratms are morphisms of \ratas, they are necessarily increasing: if $K_1 \subseteq K_2$, then $\rho(K_1) \leq \rho(K_2)$. For the sake of improved readability, when applying a \ratm $\rho$ to a singleton language $K = \{w\}$ (\emph{i.e.}, $w \in A^*$ is a word), we write $\rho(w)$ for $\rho(\{w\})$. We often consider \ratms satisfying additional properties.

\medskip
\noindent
{\bf \Nice \ratms.} We say that a \ratm $\rho: 2^{A^*} \to R$ is \emph{\nice} to indicate that for every language $K \subseteq A^*$, there exists a \emph{finite} set $F \subseteq K$ such that $\rho(K) = \rho(F)$.

\begin{rem} \label{rem:notnice}
	Not all \ratms are \nice. Consider the \rata $R = \{0,1\}$ whose addition is defined by $i+j = \mathord{\text{max}}(i,j)$ for $i,j \in R$. We define $\rho: 2^{A^*} \to R$ by $\rho(K) = 0$ if $K \subseteq A^*$ is \emph{finite} and $\rho(K) = 1$ if $K \subseteq A^*$ is \emph{infinite}. One may verify that $\rho$ is \emph{not} \nice: if $K$ is infinite, then $\rho(K) = 1$ while $\rho(F) = 0$ for every finite subset $F \subseteq K$.
\end{rem}

The definition of \nice \ratms motivates the following object. For every \ratm $\rho: 2^{A^*} \to R$ (\nice or not), we associate a map $\rho_*: A^* \to R$ defined as the restriction of $\rho$ to $A^*$: for every $w \in A^*$, $\rho_*(w) = \rho(w)$. One may verify that when $\rho$ is \nice, it is characterized by $\rho_*$. More precisely, for every $K \subseteq A^*$, we have $\rho(K) = \sum_{w \in K} \rho_*(w)$ (the sum is well-defined as it boils down to a finite one since $\rho$ is nice and $R$ is idempotent and commutative).

\medskip
\noindent
{\bf \Mratms.} The \ratas of \mratms have more structure: they are \emph{idempotent semirings}. A \emph{semiring} is a tuple $(R,+,\cdot)$ where $R$ is a set and ``$+$'' and ``$\cdot$''  are two binary operations, such that the following axioms are satisfied:
\begin{itemize}
	\item $(R,+)$ is a commutative monoid (its neutral element is denoted by $0_R$).
	\item $(R,\cdot)$ is a monoid (its neutral element is denoted by $1_R$).
	\item The neutral element of $(R,+)$ is a zero for multiplication: $0_R r = r0_R = 0_R$ for all $r \in R$.
	\item The multiplication distributes over addition: $r(s + t) = rs + rt$ and $(r + s)t = rt + st$ for every $r,s,t \in R$.
\end{itemize}
Finally, a semiring $(R,+,\cdot)$ is \emph{idempotent} when $r+r = r$ for every $r \in R$ (on the other hand, there is no additional constraint on the multiplication). By definition, it follows that in this case, the additive monoid $(R,+)$ is a \rata.

\begin{exa}
  For every alphabet $A$, the triple $(2^{A^*},\cup,\cdot)$ is an idempotent semiring (here, we use language concatenation ``$\cdot$'' as the multiplication; its neutral element is the singleton $\{\veps\}$).
\end{exa}

A \ratm $\rho: 2^{A^*} \to R$ is \emph{\tame} when the \rata $(R,+)$ is equipped with a second binary operation ``$\cdot$'' such that $(R,+,\cdot)$ is an \emph{idempotent semiring} and $\rho$ is also a monoid morphism from $(2^{A^*},\cdot)$ to $(R,\cdot)$. Thus, the axioms are as follows:
\begin{enumerate}
	\item $\rho(\emptyset) = 0_R$ and for all $K_1,K_2 \subseteq A^*$, we have $\rho(K_1\cup K_2)=\rho(K_1)+\rho(K_2)$.
	\item$\rho(\varepsilon) = 1_R$ and for all $K_1,K_2 \subseteq A^*$, we have $\rho(K_1K_2) = \rho(K_1) \cdot \rho(K_2)$.
\end{enumerate}
Altogether, this exactly says that $\rho$ is a semiring morphism from $(2^{A^*},\cup,\cdot)$ to $(R,+,\cdot)$. 

A key point is that a \ratm $\rho: 2^{A^*}\to (R,+,\cdot)$ that is both \emph{\nice} and \emph{\tame} can be finitely represented. Indeed, since $\rho$ is \nice, it is characterized by the map $\rho_*: A^* \to R$. Moreover, since $\rho$ is \tame, $\rho_*$ is a monoid morphism from $A^*$ to $(R,\cdot)$. Altogether, it follows that $\rho$ is finitely representable: it suffices to store the image $\rho(a) \in R$ of each letter $a \in A$ as well as the addition and multiplication tables of $(R,+,\cdot)$. This means that we can speak of algorithms \emph{taking a \nice \mratm as input}.

\medskip
\noindent
{\bf Canonical \nice \mratm associated to a morphism.} We complete the presentation with a simple construction. It associates a canonical \nice \mratm to a morphism into a finite monoid. We shall use it~to make the link with covering. Consider a morphism $\alpha: A^* \to M$ into a finite monoid. We associate a \nice \mratm $\rho_\alpha: 2^{A^*} \to 2^{M}$ to~$\alpha$. Consider the triple $(2^{M},\cup,\cdot)$ whose multiplication is defined as follows for all $T,T' \in 2^{M}$:
\[
TT' = \{tt' \in M \mid \text{$t \in T$ and $t' \in T'$}\}.
\]
One may verify that $(2^{M},\cup,\cdot)$ is an idempotent semiring. For every $K \subseteq A^*$, we define,
\[
\rho_\alpha(K) = \alpha(K) = \{t \in M \mid \alpha\inv(t) \cap K \neq \emptyset\}.
\]
One may verify that $\rho_\alpha: 2^{A^*} \to 2^{M}$ is a \emph{\nice \mratm}.

\subsection{Optimal covers and optimal \imprints.}

Now that we have defined what \ratms are, we turn to \imprints. Consider a \ratm $\rho: 2^{A^*} \to R$. Given any finite set of languages~\Kb, we define the $\rho$-\imprint of~\Kb. Roughly speaking, when \Kb is a cover of some language $L$, this object measures the ``quality'' of \Kb, which is a subset of~$R$. Intuitively, the smaller the imprint, the better the cover. The  \emph{$\rho$-\imprint of\/ \Kb} is the subset of~$R$ defined by:
\[
\prin{\rho}{\Kb} = \dclosr \big\{\rho(K) \mid K \in\Kb\big\}.
\]
We now define optimality. Consider an arbitrary \ratm $\rho: 2^{A^*} \to R$ and a lattice~\Ds. Given a language $L$, an \emph{optimal} \Ds-cover of $L$ for $\rho$ is a \Ds-cover \Kb of $L$ having the least possible imprint among all \Ds-covers, \emph{i.e.}, which satisfies the following~property:
\[
\prin{\rho}{\Kb} \subseteq \prin{\rho}{\Kb'} \quad \text{for every \Ds-cover $\Kb'$ of $L$}.
\]
In general, there can be infinitely many optimal \Ds-covers for a given \ratm $\rho$. The key point is that there always exists at least one, provided that \Ds is a lattice. We state this simple property in the following lemma (proved in~\cite[Lemma 4.15]{pzcovering2}).

\begin{lem}\label{lem:bgen:opt}
	Let \Ds be a lattice. For every language $L$ and every \ratm~$\rho$, there exists an optimal \Ds-cover of\/ $L$ for $\rho$.
\end{lem}

Clearly, given a lattice \Ds, a language $L$ and a \ratm $\rho$, all optimal \Ds-covers of $L$ for $\rho$ have the same $\rho$-\imprint. Hence, this unique $\rho$-\imprint is a \emph{canonical} object for \Ds, $L$ and $\rho$. We call it~the \emph{optimal $\rho$-\imprint on $L$ for \Ds} and we denote it by $\opti{\Ds}{L,\rho}$:
\[
\opti{\Ds}{L,\rho} = \prin{\rho}{\Kb} \quad \text{for any optimal \Ds-cover \Kb of $L$  for $\rho$}.
\]
An important special case is when $L = A^*$. In this case, we write \opti{\Ds}{\rho} for \opti{\Ds}{A^*,\rho}. Let us present a few properties of optimal \imprints. First, we have the following useful fact (proved in~\cite[Facts 4.16 and 4.17]{pzcovering2}).

\begin{fct}\label{fct:linclus}
	Let \Cs and \Ds be lattices such that $\Cs \subseteq \Ds$, $\rho: 2^{A^*} \to R$ be a \ratm and $H,L \subseteq A^*$ be two languages such that $H \subseteq L$. Then, $\opti{\Ds}{H,\rho} \subseteq \opti{\Cs}{L,\rho}$.
\end{fct}

Additionally, we have the following lemma (proved in~\cite[Fact 4.7]{pzbpolcj}).

\begin{lem}\label{lem:optunion}
	Let \Ds be a lattice and let $\rho: 2^{A^*} \to R$ be a \ratm. Then, $\opti{\Ds}{\emptyset,\rho} = \emptyset$. Moreover, given two languages $H,L$, we have $\opti{\Ds}{H \cup L,\rho} = \opti{\Ds}{H,\rho} \cup \opti{\Ds}{L,\rho}$
\end{lem}

We complete Lemma~\ref{lem:optunion} with a similar statement for language concatenation instead of union (proved in~\cite[Lemma 5.8]{pzcovering2}). Note that it requires more hypotheses: \Ds must be a \vari and the \ratm $\rho$ must be \tame.

\begin{lem}\label{lem:mratmult}
	Let \Ds be a \vari and let $\rho: 2^{A^*} \to R$ be a \mratm. Given two languages $H,L \subseteq A^*$,
	we have $\opti{\Ds}{H,\rho}\cdot\opti{\Ds}{L,\rho}\subseteq \opti{\Ds}{HL,\rho}$.
\end{lem}

\subsection{Connection with covering.} We now connect these definitions to the covering problem. The key idea is that solving \Ds-covering for a class \Ds boils down to finding an algorithm that computes the optimal \imprint for~\Ds from a \nice \mratm given as input. In~\cite{pzcovering2}, two statements are presented. The first is simpler but it only applies Boolean algebras, while the second, more involved, applies to all lattices. Since all classes investigated in the paper are Boolean algebras, we only present the first~one.

\begin{prop}\label{prop:breduc}
	Let \Ds be a Boolean algebra. There exists an effective reduction from \Ds-covering to the following problem:
	
	\begin{tabular}{ll}
		{\bf Input:} & A \nice \mratm $\rho: 2^{A^*} \to R$ and $F \subseteq R$. \\
		{\bf Question:} & Is it true that $\opti{\Ds}{\rho} \cap F = \emptyset$?
	\end{tabular}
\end{prop}

\begin{proof}[Proof sketch]
	We briefly describe the reduction (we refer the reader to~\cite{pzcovering2} for details). Consider an input pair $(L_0,\{L_1,\dots,L_n\})$ for \Ds-covering. Since the languages $L_i$ are regular, for every $i \leq n$, one may compute a morphism $\alpha_i: A^* \to M_i$ into a finite monoid recognizing~$L_i$ together with the set $F_i \subseteq M_i$ such that $L_i = \alpha_i\inv(F_i)$. Consider the associated \nice \mratms $\rho_{\alpha_i} : 2^{A^*} \to 2^{M_i}$. Moreover, let $R$ be the idempotent semiring $2^{M_0} \times \cdots \times 2^{M_n}$ equipped with the componentwise addition and multiplication. We define a \nice \mratm $\rho: 2^{A^*} \to R$ by letting $\rho(K) = (\rho_{\alpha_0}(K),\dots,\rho_{\alpha_n}(K))$ for every $K \subseteq A^*$. Finally, let $F \subseteq R$ be the set of all tuples $(X_0,\dots,X_n) \in R$ such that $X_i \cap F_i \neq \emptyset$ for every $i \leq n$. One may now verify that $(L_0,\{L_1,\dots,L_n\})$ is \Ds-coverable if and only if $\opti{\Ds}{\rho} \cap F = \emptyset$. Let us point out that this equivalence is only true when \Ds is a Boolean algebra. When \Ds is only a lattice, one has to handle the language $L_0$ separately.
\end{proof}

In view of Proposition~\ref{prop:breduc}, for a Boolean algebra \Ds, getting a \Ds-covering algorithm boils down to finding a procedure computing the set $\opti{\Ds}{\rho}\subseteq R$ from a \nice \mratm $\rho: 2^{A^*} \to R$. In practice, these procedures are often presented as elegant \emph{characterization theorems}. The key idea is that such a theorem should provide a description of~the set $\opti{\Ds}{\rho} \subseteq R$, which yields an algorithm for computing it as an immediate corollary. Typically, such a result is called a ``characterization of optimal \imprints for \Ds''. For example a~characterization of optimal \imprints for~\sfr is presented in~\cite{pzcovering2} (recall that \sfr is the class of star-free languages). This procedure reformulates a result of~\cite{pzfoj} in the framework of \ratms.

\begin{exa} \label{exa:sfcov}
  It is shown in~\cite{pzfoj,pzcovering2} that for every \nice \mratm $\rho: 2^{A^*} \to R$, the optimal $\rho$-\imprint for \sfr, $\opti{\sfr}{\rho} \subseteq R$, is the least subset $S \subseteq R$ which satisfies the following conditions:
  \begin{enumerate}
	\item {\bf Trivial elements:} For every $w \in A^*$, we have $\rho(w) \in S$.
	\item {\bf Closure under downset.}\label{item:sfcov2} $\dclosr S = S$.
	\item {\bf Closure under multiplication.}\label{item	:sfcov3} For every $q,r \in S$, we have $qr \in S$.
	\item {\bf \sfr-closure.}\label{item:sfcov4} For every $r \in S$, we have $r^\omega + r^{\omega+1} \in S$.
  \end{enumerate}
  This characterization yields a \emph{least fixpoint procedure} that computes \opti{\sfr}{\rho} from $\rho$: it starts from the set of all trivial elements and saturates it with the operations given by conditions~\eqref{item:sfcov2}--\eqref{item:sfcov4} above. Together with Proposition~\ref{prop:breduc}, this yields the decidability of \sfr-covering.
\end{exa}

In the paper, our aim is to generalize the characterization presented in Example~\ref{exa:sfcov} to arbitrary classes of the form \sfp{\Cs} where \Cs is a \vari. We are able to handle two cases: the case when \Cs is a~\emph{finite \vari} and the case when \Cs is a \emph{group \vari}. We present two characterizations of optimal \imprints, one for each case. They are generic in the sense that each of them applies to families of classes rather than to a single class. This raises a question. Intuitively, a generic characterization of optimal \imprints for \sfp{\Cs} should be parametrized by the class \Cs. How is this achieved? It turns~out that this requires to work with \emph{more general objects} capturing \emph{additional information}. Roughly speaking, given a \nice \mratm $\rho: 2^{A^*} \to R$, our characterizations describe a single object that captures both $\opti{\sfp{\Cs}}{\rho}\subseteq R$ and $\copti{\rho}\subseteq R$, as well as extra information, which connects them. The key idea is that while we are only interested in the set \opti{\sfp{\Cs}}{\rho}, this more general object is required to formulate a characterization. This is not surprising since characterizations of optimal \imprints are often \emph{fixpoint descriptions}, as seen in Example~\ref{exa:sfcov}. An important point is that these more general objects are also optimal \imprints. However, they involve auxiliary \ratms built from $\Cs$ and $\rho$. The last part of the section is devoted to defining these objects.

\subsection{Nested optimal \imprints} We introduce a construction from~\cite{pseps3j,pzbpolcj}, which takes as input a lattice \Ds and a \ratm $\rho: 2^{A^*} \to R$ and builds a new \ratm \bratauxd whose \rata is $(2^R,\cup)$. We let,
\[
\begin{array}{llll}
	\bratauxd: & 2^{A^*} & \to & (2^R,\cup) \\
	& K & \mapsto & \opti{\Ds}{K,\rho}.
\end{array}
\]
Let us reformulate Lemma~\ref{lem:optunion}, which exactly states that $\bratauxd: 2^{A^*} \to (2^R,\cup)$ is itself a \hbox{\ratm}.

\begin{cor} \label{cor:optunion}
	Let \Ds be a lattice and $\rho: 2^{A^*} \to R$ a \ratm., Then, $\bratauxd: 2^{A^*} \to 2^R$ is a \ratm as well.
\end{cor}

Let us point out that \bratauxd is neither \nice nor \tame in general, even if this is the case for the original \ratm $\rho$. In practice, this is not an issue for the ``\tame'' property. Actually, \bratauxd is \tame when $\rho$ is \tame and the lattice \Ds is a \vari \emph{closed under concatenation} (such as when $\Ds = \sfp{\Cs}$). On the other hand, \bratauxd is rarely \nice (see~\cite[Example 6.3]{pzbpolcj} for a counterexample). This is why it is important that most results of the framework hold for \emph{arbitrary \ratms}.

Let us now provide some high-level intuition on why this construction is important when dealing with star-free closure. Consider some input \vari \Cs. Since \sfp{\Cs} is a Boolean algebra, we know by Proposition~\ref{prop:breduc} that in order to solve \sfp{\Cs}-covering, it suffices to characterize the set $\opti{\sfp{\Cs}}{\rho}\subseteq R$ for every \nice \mratm $\rho: 2^{A^*}\to R$. Roughly speaking, the two characterizations of optimal \imprints for \sfp{\Cs} that we present (for finite \varis \Cs in Section~\ref{sec:finite} and for \varis of group languages in Section~\ref{sec:group}) consider the set  $\copti{\bratauxsfc} \subseteq 2^R$, which is more general. Indeed, one may verify that $\opti{\sfp{\Cs}}{\rho}$ is the union of all sets in $\copti{\bratauxsfc}$. The point is that the extra information contained in this set is required in order to formulate these characterizations. Let us point out that this discussion is only meant to provide a sketchy general intuition. In practice, we shall refine these ideas by adopting an approach tailored to the two particular kinds of input \varis that we shall consider.


\section{Covering for finite input classes}
\label{sec:finite}
In this section, we prove that separation and covering are both decidable for the class \sfp{\Cs} when \Cs is a finite \vari. The algorithm is based on the framework introduced in Section~\ref{sec:ratms}: we present a generic effective characterization of optimal \imprints for \sfp{\Cs}. Given as input an arbitrary \emph{finite} \vari \Cs and a \nice \mratm $\rho: 2^{A^*} \to R$, it provides an effective description of the set $\opti{\sfp{\Cs}}{\rho} \subseteq R$. By Proposition~\ref{prop:breduc}, having this set in hand suffices to decide \sfp{\Cs}-covering. As announced at the end of Section~\ref{sec:ratms}, the characterization actually describes more information than just the set $\opti{\sfp{\Cs}}{\rho} \subseteq R$. The first part of the section is devoted to defining the full object that we characterize. In the second part, we present the characterization itself and its proof.

\subsection{Pointed optimal \imprints}

Consider an arbitrary finite \vari \Cs. Recall that since \Cs is \emph{finite}, Lemma~\ref{lem:cmdiv} implies that there exists a \emph{unique} (up to renaming) \Cs-morphism recognizing exactly all languages in \Cs. We denote it by $\etac: A^* \to \canc$ and call it the \emph{canonical \Cs-morphism}. The set~$\Kb=\{\etac\inv(t)\mid t\in\canc\}$ is the \emph{finest} partition of $A^*$ into languages of \Cs. Consequently, \Kb is an optimal \Cs-cover of $A^*$ \emph{for all \ratms}. In particular, if $\rho: 2^{A^*} \to R$ is a \nice \mratm, then \Kb is an optimal \Cs-cover of $A^*$ for the auxiliary \ratm $\brataux{\sfp{\Cs}}{\rho}: 2^{A^*} \to 2^R$. By definition, it follows that,
\[
  \opti{\Cs}{\brataux{\sfp{\Cs}}{\rho}} = \prin{\brataux{\sfp{\Cs}}{\rho}}{\Kb} = \dclosp{2^R} \{\brataux{\sfp{\Cs}}{\rho}(K) \mid K \in \Kb\} =\dclosp{2^R} \{\opti{\sfp{\Cs}}{\etac\inv(t),\rho} \mid t \in \canc\}.
\]
As explained in Section~\ref{sec:ratms}, the set $\opti{\Cs}{\brataux{\sfp{\Cs}}{\rho}}$ is exactly the information that our characterization of optimal \imprints for \sfp{\Cs} will describe. More precisely, we characterize the family of sets $\opti{\sfp{\Cs}}{\etac\inv(t),\rho} \subseteq R$ for $t \in \canc$. According to Lemma~\ref{lem:optunion}, the union of these sets is  $\opti{\sfp{\Cs}}{A^*,\rho}=\opti{\sfp{\Cs}}{\rho}$, whose knowledge is enough to decide \sfp\Cs-covering. Let us point out that this is also how the characterization depends on the finite \vari~\Cs: it is parametrized by the canonical \Cs-morphism $\etac: A^* \to \canc$. We now introduce additional notations that will be convenient in order to manipulate the family of sets \opti{\sfp{\Cs}}{\etac\inv(t),\rho} in our statements and our proofs.

\medskip\noindent{\bf Definition.} Let \Ds be a lattice, $\eta: A^* \to N$ be a morphism into a finite monoid and $\rho: 2^{A^*} \to R$ be a \ratm. The \emph{$\eta$-pointed optimal $\rho$-\imprint for \Ds} is defined as the following set $\popti{\Ds}{\eta}{\rho} \subseteq N \times R$:
\[
  \popti{\Ds}{\eta}{\rho} = \bigl\{(t,r) \in N \times R \mid r \in \opti{\Ds}{\eta\inv(t),\rho}\bigl\}.
\]
Clearly, \popti{\Ds}{\eta}{\rho} encodes all sets \opti{\Ds}{\eta\inv(t),\rho} for $t \in N$. We shall use the~notation in the case when $\Ds = \sfp{\Cs}$ for some \vari \Cs and $\eta$ is the canonical \Cs-morphism \etac.

We complete the definition with a simple result which implies that $\popti{\Ds}{\eta}{\rho} \subseteq N \times R$ is more general than the set $\opti{\Ds}{\rho} \subseteq R$. Indeed, since $\opti{\Ds}{\rho} = \opti{\Ds}{A^*,\rho}$, we have the following immediate corollary of Lemma~\ref{lem:optunion}.

\begin{cor}\label{cor:pointgen}
  Let \Ds be a lattice, $\eta: A^* \to N$ be a morphism into a finite monoid and $\rho: 2^{A^*} \to R$ be a \mratm. Then,
  \[
    \opti{\Ds}{\rho} = \bigcup_{t\in N}  \opti{\Ds}{\eta\inv(t),\rho} =\{r \in R \mid \text{there exists $t \in N$ such that $(t,r) \in \popti{\Ds}{\eta}{\rho}$}\}.
  \]
\end{cor}

\noindent
{\bf Pointed covers.} Pointed optimal \imprints are closer to being a notation rather than a new notion. Yet, it is possible to define them \emph{directly} in terms of ``covers''. This will be convenient for manipulating them. However, we have to slightly generalize the notion of cover in order to do so.

Consider a morphism $\eta: A^* \to N$ into a finite monoid and a language $L \subseteq A^*$. An \emph{$\eta$-pointed cover} of $L$ is a \emph{finite} set \Kb of pairs $(s,K) \in N \times 2^{A^*}$ such that for every $w \in L$, there exists $(s,K) \in \Kb$ such that $\eta(w) = s$ and $w\in K$. In other words, the set $\{K \mid (s,K) \in \Kb\}$ must be a cover of $L\cap\eta\inv(s)$ for every $s \in N$. Additionally, given some class \Ds, we say that \Kb is an $\eta$-pointed \Ds-cover when it also satisfies $K \in \Ds$ for each pair $(s,K) \in \Kb$.

We generalize \imprints to pointed covers. Let $\eta: A^* \to N$ be a morphism into a finite monoid and $\rho: 2^{A^*} \to R$ a be \ratm. If \Kb is an $\eta$-pointed cover of some language $L \subseteq A^*$, we write,
\[
  \pprin{\eta,\rho}{\Kb} = \dclosr \{(s,\rho(K)) \mid (s,K) \in \Kb\} \subseteq N \times R.
\]
Here, we use the extended definition of the downset operation (see the definition page~\pageref{downset}). The following lemma provides an alternate definition of pointed optimal \imprints. Roughly, it implies that when \Ds is a lattice, there always exists an ``optimal'' $\eta$-pointed \Ds-cover \Kb of $A^*$ (\emph{i.e.}, such that \pprin{\eta,\rho}{\Kb} is minimal for inclusion) and that it satisfies $\pprin{\eta,\rho}{\Kb} = \popti{\Ds}{\eta}{\rho}$.

\begin{lem} \label{lem:pcov}
  Let \Ds be a lattice, $\eta: A^* \to N$ be a morphism into a finite monoid and $\rho: 2^{A^*} \to R$ be a \ratm. The two following properties hold:
  \begin{itemize}
    \item For every $\eta$-pointed \Ds-cover \Kb of $A^*$, we have $\popti{\Ds}{\eta}{\rho} \subseteq  \pprin{\eta,\rho}{\Kb}$.
    \item There exists an $\eta$-pointed \Ds-cover \Kb of $A^*$ such that $\popti{\Ds}{\eta}{\rho} =  \pprin{\eta,\rho}{\Kb}$.
  \end{itemize}
\end{lem}

\begin{proof}
  For the first assertion, let \Kb be an $\eta$-pointed \Ds-cover \Kb of $A^*$ and let $(s,r) \in \popti{\Ds}{\eta}{\rho}$. Let $\Kb_s = \{K \mid (s,K) \in \Kb\}$, which is a \Ds-cover of $\eta\inv(s)$ by definition of pointed covers. Since $(s,r)\in\popti{\Ds}{\eta}{\rho}$, we have $r \in \opti{\Ds}{\eta\inv(s),\rho}$, which yields $K \in \Kb_s$ such that $r \leq \rho(K)$. By definition of $\Kb_s$, we have $(s,K) \in \Kb$. Hence, we get $(s,r) \in \pprin{\eta,\rho}{\Kb}$, as desired.

  We turn to the second assertion. For every $s \in N$, let $\Kb_s$ be an optimal \Ds-cover of $\eta\inv(s)$ for $\rho$: $\prin{\rho}{\Kb_s}=\opti{\Ds}{\eta\inv(s),\rho}$. We let $\Kb = \{(s,K) \mid s \in N \text{ and } K \in \Kb_s\}$. By definition, \Kb is an $\eta$-pointed cover of $L$. Let us prove that $\popti{\Ds}{\eta}{\rho} =  \pprin{\eta,\rho}{\Kb}$. The left to right inclusion is immediate from the first assertion. For the converse one, let $(s,r) \in \pprin{\eta,\rho}{\Kb}$. By definition of $\Kb$, we get $K \in \Kb_s$ such that $r \leq \rho(K)$, \emph{i.e.}, $r \in \prin{\rho}{\Kb_s}=\opti{\Ds}{\eta\inv(s),\rho}$. Hence, $(s,r) \in \popti{\Cs}{\eta}{\rho}$ by definition.
\end{proof}

\subsection{Characterization}

Let us first present the characterization. Given an arbitrary morphism $\eta: A^* \to N$ into a finite monoid and a \mratm $\rho: 2^{A^*} \to R$, we define the \sfr-saturated subsets of $N \times R$ for $\eta$ and $\rho$ (this notion makes sense for arbitrary \mratms, but it is only useful for those that are \emph{\nice}). Let $S \subseteq N \times R$. We say that $S$ is \emph{\sfr-saturated for $\eta$ and $\rho$} when it satisfies the following properties:
\begin{enumerate}
  \item {\bf Trivial elements.} For every $w \in A^*$, we have $(\eta(w),\rho(w)) \in S$.
  \item {\bf Closure under downset.} $\dclosr S = S$.
  \item {\bf Closure under multiplication.} For every $(s,q),(t,r) \in S$, we have $(st,qr) \in S$.
  \item {\bf \sfr-closure.} For every $(e,r) \in S$, if $e \in N$ is an idempotent, then $(e,r^\omega + r^{\omega+1}) \in S$.
\end{enumerate}

We are ready to present the characterization. Given a finite \vari \Cs, if $\rho$ is \nice, we show that the least \sfr-saturated subset of $\canc \times R$ for \etac and $\rho$ is exactly \popti{\sfp{\Cs}}{\etac}{\rho}.

\begin{thm} \label{thm:carfinite}
  Let \Cs be a finite \vari and $\rho: 2^{A^*} \to R$ be a \nice \mratm. Then, \popti{\sfp{\Cs}}{\etac}{\rho} is the least \sfr-saturated subset of $\canc \times R$ for $\etac$ and $\rho$.
\end{thm}

Given a \nice \mratm $\rho: 2^{A^*} \to R$ as input, it is clear the one may compute the least \sfr-saturated subset of $\canc \times R$ for $\etac$ and $\rho$. This is achieved with a least fixpoint procedure. Hence, Theorem~\ref{thm:carfinite} provides an algorithm for computing \popti{\sfp{\Cs}}{\etac}{\rho}. It then follows from Corollary~\ref{cor:pointgen} that one may compute \opti{\sfp{\Cs}}{\rho} from \popti{\sfp{\Cs}}{\etac}{\rho}:
\[
  \opti{\sfp{\Cs}}{\rho} = \{r \in R\mid \text{there exists $t \in N$ such that $(t,r) \in \popti{\sfp{\Cs}}{\etac}{\rho}$}\}.
\]
Together with Proposition~\ref{prop:breduc}, we obtain that \sfp{\Cs}-covering is decidable. Naturally, this result extends to separation by Lemma~\ref{lem:covsep}. Altogether, we get the following corollary.

\begin{cor} \label{cor:carfinite}
  Let \Cs be a finite \vari. Then, \sfp{\Cs}-covering and \sfp{\Cs}-separation are decidable.
\end{cor}

In practice, there are not many interesting applications of Corollary~\ref{cor:carfinite}. Indeed, the only important class that is the star-free closure of a finite \vari is the original class of star-free languages. Indeed, we have $\sfr = \sfp{\stzer}$ where $\stzer = \{\emptyset,A^*\}$. It was already know that \sfr has decidable covering: we presented an effective characterization of optimal \imprints for \sfr taken from~\cite{pzcovering2} in Example~\ref{exa:sfcov}. This specialized characterization is actually an immediate corollary of Theorem~\ref{thm:carfinite} since the canonical \stzer-morphism is the unique one $\eta_\stzer: A^* \to \{1\}$ into a trivial monoid $\{1\}$.

Nonetheless, Theorem~\ref{thm:carfinite} is an important result. Indeed, we shall use it as a subresult in the proof of our second characterization, which describes optimal \imprints for \sfp{\Gs} when \Gs is a group \vari (more precisely, we shall apply Propositions~\ref{prop:sound} and~\ref{prop:comp} below).

\medskip

We turn to the proof of Theorem~\ref{thm:carfinite}. It involves two independent statements, which correspond respectively to soundness and completeness of the least fixpoint procedure computing \popti{\sfp{\Cs}}{\etac}{\rho} from a \nice \mratm $\rho$. Let us start with soundness, which is simpler to establish and does not require the hypothesis that $\rho$ is \nice.

\begin{prop}[Soundness] \label{prop:sound}
  Consider  a finite \vari \Cs and a \mratm $\rho: 2^{A^*} \to R$. Then, the set $\popti{\sfp{\Cs}}{\etac}{\rho} \subseteq \canc \times R$ is \sfr-saturated for $\etac$ and $\rho$.
\end{prop}

\begin{proof}
  Recall that \sfp{\Cs} is a \vari by Proposition~\ref{prop:vari}. There are four properties to verify. We start with the first three, which are standard.  For the trivial elements, consider $w \in A^*$ and let \Kb be an optimal \sfp{\Cs}-cover of $\etac\inv(\etac(w))$. Since $w \in \etac\inv(\etac(w))$, there exists $K \in \Kb$ such that $w\in K$. Therefore, $\rho(w) \leq \rho(K)$, which yields $\rho(w) \in \prin{\rho}{\Kb} = \opti{\sfp{\Cs}}{\etac\inv(\etac(w)),\rho}$. By definition, this implies that $(\etac(w),\rho(w)) \in \popti{\sfp{\Cs}}{\etac}{\rho}$. Closure under downset is immediate by definition of \imprints. Finally, for closure under multiplication, consider $(s_1,r_1),(s_2,r_2) \in \popti{\sfp{\Cs}}{\etac}{\rho}$. We have $r_i \in \opti{\sfp{\Cs}}{\etac\inv(s_i),\rho}$ for $i = 1,2$. Since \sfp{\Cs} is a \vari, Lemma~\ref{lem:mratmult} yields $r_1r_2 \in \opti{\sfp{\Cs}}{\etac\inv(s_1)\etac\inv(s_2),\rho}$. Clearly, $\etac\inv(s_1)\etac\inv(s_2) \subseteq \etac\inv(s_1s_2)$. Thus, Fact~\ref{fct:linclus} yields $r_1r_2 \in \opti{\sfp{\Cs}}{\etac\inv(s_1s_2),\rho}$. By definition, this exactly says that $(s_1s_2,r_1r_2) \in  \popti{\sfp{\Cs}}{\etac}{\rho}$.

  It remains to handle \sfr-closure.  Let $(e,r) \in \popti{\sfp{\Cs}}{\etac}{\rho}$ be such that $e \in \canc$ is idempotent. We show that $(e,r^\omega + r^{\omega+1}) \in \popti{\sfp{\Cs}}{\etac}{\rho}$. Let \Kb be an optimal \sfp{\Cs}-cover of $\etac\inv(e)$. By definition, it now suffices to prove that $r^\omega + r^{\omega+1} \in \prin{\rho}{\Kb}$. Since \sfp{\Cs} is a \vari, Proposition~\ref{prop:genocm} yields an \sfp{\Cs}-morphism $\alpha: A^* \to M$  recognizing every $K \in \Kb$. Let \Hb be the set of all languages $\alpha\inv(x)$ for $x \in M$ such that $\alpha\inv(x) \cap \etac\inv(e)\neq \emptyset$.  By definition, \Hb is an \sfp{\Cs}-cover of $\etac\inv(e)$. Hence, since $(e,r) \in \popti{\sfp{\Cs}}{\etac}{\rho}$, we have $r \in \prin{\rho}{\Hb}$ which yields $H \in \Hb$ such that $r \leq \rho(H)$. By definition of $\Hb$, we get $x \in M$ such that $H = \alpha\inv(x)$ and some word $u \in H \cap \etac\inv(e)$. Let $p = \omega(M) \times \omega(R)$ and $L = \alpha\inv(x^p)$. We claim that:
  \begin{equation} \label{eq:sound}
    H^p \cup H^{p+1} \subseteq L.
  \end{equation}
  Let us first explain why this implies that $r^\omega + r^{\omega+1} \in \prin{\rho}{\Kb}$. Since $u^p \in \etac\inv(e)$ (recall that $e \in \canc$ is idempotent) and \Kb is a cover of $\etac\inv(e)$, we get $K \in \Kb$ such that $u^p \in K$. Moreover, since $u^p \in L = \alpha\inv(x^p)$ and $K$ is recognized by $\alpha$ (by definition), it follows that $L \subseteq K$. Hence, since $r \leq \rho(H)$ and $p = \omega(M) \times \omega(R)$, it follows from~\eqref{eq:sound} that $r^\omega + r^{\omega+1} \leq \rho(K)$ which yields $r^\omega + r^{\omega+1} \in \prin{\rho}{\Kb}$, as desired.

  It remains to prove~\eqref{eq:sound}. Let $w \in H^p \cup H^{p+1}$. We need to prove that $\alpha(w) = x^p$. This is immediate when $w \in H^p$ since $H = \alpha\inv(x)$. Assume now that $w \in H^{p+1}$. In that case, we have $\alpha(w) = x^{p+1}$. Hence, it suffices to prove that $x^{p+1} = x^p$. By hypothesis $\etac(u) = e$ is an idempotent of \canc. Since \etac is the canonical \Cs-morphism, it follows from Lemma~\ref{lem:cmdiv} that the image of $u$ under \emph{any} \Cs-morphism is an idempotent. Hence, since $\alpha$ is an \sfp{\Cs}-morphism, Proposition~\ref{prop:aper} implies that $(\alpha(u))^{\omega+1} = (\alpha(u))^{\omega}$. Since $x =\alpha(u)$ and $p = \omega(M) \times \omega(R)$, it follows that $x^{p+1} = x^p$, as~desired.
\end{proof}

We now turn to completeness. As usual, this is the most difficult part of the proof. Note that again, we shall rely on Theorem~\ref{thm:bsdcar}: we build our languages in \sfp{\Cs} using the operations available in the definition of \bsdp{\Cs}.

\begin{prop}[Completeness] \label{prop:comp}
  Let \Cs be a \vari, $\eta: A^* \to N$ be a \Cs-morphism, $\rho: 2^{A^*} \to R$ be a \emph{\nice} \mratm and $S \subseteq N \times R$ be an \sfr-saturated set for $\eta$ and $\rho$. There exists an $\eta$-pointed \sfp{\Cs}-cover \Kb of $A^*$ such that $\pprin{\eta,\rho}{\Kb} \subseteq S$.
\end{prop}

\begin{proof}
  We build the $\eta$-pointed \sfp{\Cs}-cover \Kb using induction. Let us start with some terminology. A first point is that we build particular $\eta$-pointed covers. Let $P \subseteq A^*$ be an arbitrary language. An \emph{$\eta$-pointed \sfp{\Cs}-partition} of $P$ is a \emph{finite} set $\Hb\subseteq N\times 2^{A^*}$ such that for every $(t,H) \in \Hb$, we have $H \in \sfp{\Cs}$ and for every $t \in N$, the set $\{H \mid (t,H) \in \Hb\}$ is a \emph{partition} of $P \cap \eta\inv(t)$ (this implies that $H \subseteq \eta\inv(t)$ for every $(t,H) \in \Hb$). Note that in particular, \Hb is an $\eta$-pointed \sfp{\Cs}-cover of~$P$.

  If \Hb is an $\eta$-pointed \sfp{\Cs}-partition, an \emph{\Hb-product} is a pair $(t_1 \cdots t_n,H_1 \cdots H_n)$ for some $n \in \nat$ where $(t_i,H_i) \in \Hb$ for every $i \leq n$. In particular, the pair $(1_N,\{\veps\})$ is an \Hb-product: this is the case $n = 0$ in the definition. When $n \geq 1$, we speak of \emph{strict \Hb-products}. Finally, an \emph{\Hb-union} (resp. strict \emph{\Hb-union}) is a pair $(t,G_1 \cup \cdots \cup G_m)$ for some $t \in N$ and $m \in \nat$ such that $(t,G_j)$ is an \Hb-product (resp. strict \Hb-product) for every $j \leq m$. Note that $(t,\emptyset)$ is a strict \Hb-union for every $t \in \nat$: this corresponds to the case $m=0$ in the definition.

  \newcommand{\csatp}[1]{\ensuremath{Q^+_{#1}}\xspace}
  \newcommand{\csats}[1]{\ensuremath{Q^*_{#1}}\xspace}

  Additionally, we write $Q = N \times R$. Observe that $Q$ is a monoid for the componentwise multiplication. Moreover, $S \subseteq Q$ by definition. Finally, for every $\eta$-pointed \sfp{\Cs}-partition of \Hb, we associate two subsets of $Q$. The definitions are as follows:
  \begin{itemize}
    \item We let $\csatp{\Hb} \subseteq Q$ be the set of all elements $(t,\rho(H)) \in Q$ where $(t,H)$ is a \emph{strict} \Hb-union.
    \item We let $\csats{\Hb} \subseteq Q$ be the set of all elements $(t,\rho(H)) \in Q$ where $(t,H)$ is an \Hb-union.
  \end{itemize}
  Clearly, $\csatp{\Hb} \subseteq \csats{\Hb}$. Moreover, we have the following simple fact which we shall use implicitly throughout the proof.

  \begin{fct} \label{fct:csatm}
    Let $P\subseteq A^*$ and \Hb be an $\eta$-pointed \sfp{\Cs}-partition of $P$. For all $(t_1,q_1),(t_2,q_2) \in \csats{\Hb}$, we have $(t_1t_2,q_1q_2) \in \csats{\Hb}$.
  \end{fct}

  \begin{proof}
    By definition, there are two \Hb-unions $(t_1,H_1),(t_2,H_2) \in \Hb$ such that $q_i = \rho(H_i)$ for $i = 1,2$. One may verify from the definition that $(t_1t_2,H_1H_2)$ remains an \Hb-union since language distributes over union. Since $q_1q_2 = \rho(H_1H_2)$, it follows that $(t_1t_2,q_1q_2) \in \csats{\Hb}$.
  \end{proof}

  We are ready to prove Proposition~\ref{prop:comp}. It is based on the following statement.

  \begin{lem} \label{lem:pumping}
    Let $P \subseteq A^+$ be a prefix code with bounded synchronization delay and \Hb be an $\eta$-pointed \sfp{\Cs}-partition of~$P$ such that $(t,\rho(H)) \in S$ for every $(t,H) \in \Hb$. Then, for every $(s,r) \in S$, there exists an $\eta$-pointed \sfp{\Cs}-partition \Kb of $P^*$ such that,
    \begin{equation} \label{eq:sfclos:covergoal}
      \text{for every $(t,K) \in \Kb$, we have $(t,\rho(K)) \in \csats{\Hb}$ and $(st,r\rho(K)) \in S$.}
    \end{equation}
  \end{lem}

  Before we prove Lemma~\ref{lem:pumping}, let us first complete the main argument. We have to build an $\eta$-pointed \sfp{\Cs}-cover \Kb of $A^*$ such that $\pprin{\eta,\rho}{\Kb} \subseteq S$. Observe that $A \subseteq A^+$ is a prefix code with bounded synchronization delay. Moreover, $\Hb = \{(\eta(a),\{a\}) \mid a \in A\}$ is an $\eta$-pointed \sfp{\Cs}-partition of $A$ and we have $(\eta(a),\rho(a)) \in S$ for every $a \in A$ since $S$ is \sfr-saturated (these are trivial elements). Finally, $(1_N,1_R) \in S$ (again, this is a trivial element). Therefore, we may apply Lemma~\ref{lem:pumping} in the case when $P = A$ and $(s,r) = (1_N,1_R) \in S$. This yields an $\eta$-pointed \sfp{\Cs}-partition \Kb of $A^*$ satisfying~\eqref{eq:sfclos:covergoal}. In particular, $(t,\rho(K)) \in S$ for every $(t,K) \in \Kb$. Therefore, since $S$ is \sfr-saturated, closure under downset yields $\pprin{\eta,\rho}{\Kb} \subseteq S$, which completes the proof.

  \smallskip

  We now prove Lemma~\ref{lem:pumping}. Let $P \subseteq A^+$ be a prefix code with bounded synchronization delay and \Hb be an $\eta$-pointed \sfp{\Cs}-partition of $P$ such that $(t,\rho(H)) \in S$ for all $(t,H) \in \Hb$. Finally, let $(s,r) \in S$. We need to build an $\eta$-pointed \sfp{\Cs}-partition \Kb of $P^*$ satisfying~\eqref{eq:sfclos:covergoal}. We proceed by induction on the three following parameters, listed by order of importance:
  \begin{enumerate}
    \item The size of the set $\csatp{\Hb} \subseteq Q$,
    \item The size of \Hb,
    \item The size of the set $(s,r) \cdot \csats{\Hb} \subseteq Q$.
  \end{enumerate}

  We distinguish two main cases depending on the following property. We say that \emph{$(s,r)$ is \Hb-stable} when the following holds,
  \begin{equation} \label{eq:sfclos:ratstable}
    \text{for every $(t,H) \in \Hb$,} \quad (s,r) \cdot \csats{\Hb} = (s,r) \cdot \csats{\Hb} \cdot (t,\rho(H)).
  \end{equation}
  We first consider the case when $(s,r)$ is \Hb-stable. This is the base case: we construct \Kb directly. Then, we handle the case when $(s,r)$ is not \Hb-stable using induction on our three parameters.

  \smallskip
  \noindent {\bf Base case: $(s,r)$ is \Hb-stable.} In this case, we let $\Kb = \{(t,P^* \cap \eta\inv(t)) \mid t \in N\}$. Let us first verify that this is an $\eta$-pointed \sfp{\Cs}-partition \Kb of $P^*$. It is immediate that $\{P^* \cap \eta\inv(t)\}$ is a partition of $P^* \cap \eta\inv(t)$ for every $t \in N$. Moreover, we have $P^* \cap \eta\inv(t) \in \sfp{\Cs}$ for every $t \in N$. Indeed, we have $P \in\sfp{\Cs}$: it is the disjoint union of all languages involved in the $\eta$-pointed \sfp{\Cs}-partition \Hb of $P$. Since $\sfp{\Cs} = \bsdp{\Cs}$ by Theorem~\ref{thm:bsdcar}, we obtain $P^* \cap \eta\inv(t) \in \sfp{\Cs}$ since $P$ is a prefix code with bounded synchronization delay and $\eta\inv(t) \in \Cs$ (recall that $\eta$ is a \Cs-morphism by hypothesis). It remains to prove that \Kb satisfies~\eqref{eq:sfclos:covergoal}: for every $(t,K) \in \Kb$, we show that $(t,\rho(K)) \in \csats{\Hb}$ and $(st,r\rho(K)) \in S$. We start with the former property (this is where we use the hypothesis that $\rho$ is~\nice).

  \begin{fct} \label{fct:images}
    For every $(t,K) \in \Kb$, we have $(t,\rho(K)) \in \csats{\Hb}$ .
  \end{fct}

  \begin{proof}
    By definition of \Kb, we have $K = P^* \cap \eta\inv(t)$. Since \Hb is an $\eta$-pointed partition of $P$, one may verify that $P^* \cap \eta\inv(t)$ is the (infinite) union of \emph{all} \Hb-products $(t',H)$ such that $t' = t$. Since $\rho$ is \nice, it follows that there exists \emph{finitely many} \Hb-products $(t,H_1),\dots,(t,H_\ell)$ such that $\rho(K) = \rho(H_1) + \cdots + \rho(H_\ell) = \rho(H_1 \cup \cdots \cup H_\ell)$. Clearly, $(t,H_1 \cup \cdots \cup H_\ell)$ is an \Hb-union and it follows that $(t,\rho(K)) \in \csats{\Hb}$, as desired.
  \end{proof}

  It remains to show that $(st,r\rho(K)) \in S$ for every $K \in\Kb$. We prove that for every $(t,q) \in \csats{\Hb}$, we have $(st,rq)  \in S$. In view of Fact~\ref{fct:images}, this yields the desired result. First, we use the hypothesis that $(s,r)$ is \Hb-stable to prove the following fact.

  \begin{fct}\label{fct:basecase}
    Let $(e,H)$ be an \Hb-product such that $(e,\rho(H)) \in Q$ is a pair of idempotents. For every $(t,q) \in \csats{\Hb}$, we have $(ste,rq\rho(H)) = (st,rq)$.
  \end{fct}

  \begin{proof}
    We first use the hypothesis that $(s,r)$ is $\Hb$-stable to prove the following preliminary result which holds regardless of whether $(e,\rho(H))$ is a pair of idempotents or not:
    \begin{equation} \label{eq:sfclos:auxproof}
      \text{there exists $(x,y) \in \csats{\Hb}$ such that $(sxe,ry\rho(H)) = (st,rq)$}.
    \end{equation}
    Since $(e,H)$ is an \Hb-product, we can find elements $(t'_1,H_1),\dots,(t'_n,H_n)$ of $\Hb$ such that $(e,H) = (t'_1 \cdots t'_n,H_1 \cdots H_n)$. We proceed by induction on $n$. If $n = 0$, then $e = 1_N$ and $H = \{\veps\}$.  It suffices to choose $(x,y) = (t,q) \in \csats{\Hb}$. We now assume that $n \geq 1$. By induction, we get $(x',y') \in \csats{\Hb}$ such that $(sx't'_2 \cdots t'_n,ry'\rho(H_2 \cdots H_n)) = (st,rq)$. Since $(s,r)$ is $\Hb$-stable, Property \eqref{eq:sfclos:ratstable} yields $(x,y) \in \csats{\Hb}$ such that $(sx',ry') =  (sxt'_1,ry\rho(H_1))$. Altogether, it follows that $(sxe,ry\rho(H)) = (st,rq)$. This concludes the proof of~\eqref{eq:sfclos:auxproof}.

    We use~\eqref{eq:sfclos:auxproof} to conclude Fact~\ref{fct:basecase}. Indeed, since $(sxe,ry\rho(H)) = (st,rq)$ for $(x,y) \in \csats{\Hb}$, if $(e,\rho(H))$ is a pair of multiplicative idempotents, then we obtain $(ste,rq\rho(H)) = (st,rq)$.
  \end{proof}

  We are ready to prove that $(st,rq)  \in S$ for every $(t,q) \in \csats{\Hb}$. We first treat the special case when~$t$ is an idempotent of $N$. Then, we reuse this special case to treat the general one. Assume that $t$ is an idempotent $e \in E(N)$. Hence, $(e,q)\in\csats{\Hb}$ and we need to prove that $(se,rq) \in S$. By definition, there are finitely many \Hb-products $(e,H_1),\dots,(e,H_\ell)$ such that $q = \rho(H_1) + \cdots + \rho(H_\ell)$. Consider an index $i \leq \ell$. Since $(e,H_i)$ is an \Hb-product and we know that $(t',\rho(H')) \in S$ for every $(t',H') \in \Hb$, it follows from closure under multiplication for $S$ (recall that $S$ is \sfr-saturated) that $(e,\rho(H_i)) \in S$. Since $e$ is idempotent, it then follows from \sfr-closure that $(e,(\rho(H_i))^\omega + (\rho(H_i))^{\omega+1}) \in S$. Since this holds for all $i \leq \ell$, $e$ is idempotent and $(s,r) \in S$, closure under multiplication yields,
  \[
    \Bigl(se,\quad r\prod_{1 \leq i \leq \ell} \left((\rho(H_i))^\omega + (\rho(H_i))^{\omega+1}\right)\Bigr) \in S.
  \]
  Let $k = \omega(R)$. For every $i \leq \ell$, we have $(e,(\rho(H_i))^\omega) = (e,\rho(H_i^k))$ and it is clear that $(e,H_i^k)$ is an \Hb-product since this is the case for $(e,H_i)$. Therefore, since $(e,(\rho(H_i))^\omega)$ is a pair of idempotents, Fact~\ref{fct:basecase} implies that $(st'e,rq'(\rho(H_i))^\omega)) = (st',rq')$ for every $(t',q') \in \csats{\Hb}$. This yields:
  \[
    \Bigl(se,\quad r\prod_{1 \leq i \leq \ell} \left((\rho(H_i))^\omega + (\rho(H_i))^{\omega+1}\right)\Bigr) = \Bigl(se,\quad r\prod_{1 \leq i \leq \ell} \left(1_R + \rho(H_i)\right)\Bigl),
  \]
  Therefore, $(se, r\prod_{1 \leq i \leq \ell} \left(1_R + \rho(H_i)\right))\in S$.
  Note that $q = \rho(H_1) + \cdots + \rho(H_\ell) \leq \prod_{1 \leq i \leq \ell} \left(1_R + \rho(H_i)\right)$. By closure under downset for $S$, we get $(se,rq) \in S$, which concludes the case when $t$ is~idempotent.

  We now consider an arbitrary element $(t,q) \in \csats{\Hb}$ (\emph{i.e.}, $t \in N$ need not be idempotent) and show that $(st,rq) \in S$. By definition of \csats{\Hb} there are finitely many \Hb-products $(t,H_1), \dots (t,H_n)$ such that $q = \rho(H_1) + \cdots \rho(H_n)$. Since every $(t',H') \in \Hb$ satisfies $(t',\rho(H')) \in S$ by hypothesis and $S$ is closed under multiplication, we get $(t,\rho(H_i)) \in S$ for every $i \leq n$. Moreover, there exists a number $k \geq 1$ such that $(t^k,(\rho(H_i))^k) \in \csats{\Hb}$ is a pair of idempotents for each $i \leq n$. In particular, $t^k$ is an idempotent of $N$. Clearly, we have $(t^k,q(\rho(H_1))^{k-1}) \in \csats{\Hb}$. Since $t^k \in N$ is an idempotent, we obtain from the above special case that $(st^k,rq(\rho(H_1))^{k-1}) \in S$. Since we also have $(t,\rho(H_1)) \in S$, it then follows from closure under multiplication that $(st^{k+1},rq(\rho(H_1))^{k}) \in S$. Finally, since $(t^k,(\rho(H_1))^k) $ is a pair of idempotents, we obtain from Fact~\ref{fct:basecase} that $(st,rq) = (st^{k+1},rq(\rho(H_1))^{k}) \in S$. This concludes the proof for the base case.

  \medskip
  \noindent
  {\bf Inductive step: $(s,r)$ is not $\Hb$-stable.} Our hypothesis yields a pair  $(t,H) \in \Hb$ such that the following \emph{strict} inclusion holds:
  \begin{equation} \label{eq:sfclos:coverind}
    (s,r) \cdot \csats{\Hb} \cdot (t,\rho(H)) \subsetneq (s,r) \cdot \csats{\Hb}.
  \end{equation}
  We fix this pair $(t,H) \in \Hb$ for the remainder of the proof. First, we use induction on our second parameter in Lemma~\ref{lem:pumping} to prove the following fact.

  \begin{fct} \label{fct:covalphind}
    There exists an $\eta$-pointed \sfp{\Cs}-partition \Ub of $(P \setminus H)^*$ such that $(x,\rho(U)) \in \csats{\Hb} \cap S$ for every $(x,U) \in \Ub$.
  \end{fct}

  \begin{proof}
    Since $P$ is a prefix code with bounded synchronization delay, Fact~\ref{fct:sprefsub} implies that so is $P \setminus H$. We want to apply induction in Lemma~\ref{lem:pumping} for the case when $P$ has been replaced by $P \setminus H$. Let $\Gb = \Hb \setminus \{(t,H)\}$. By hypothesis on \Hb, one may verify that $\Gb$ is an $\eta$-pointed \sfp{\Cs}-partition of $P \setminus H$ and that $(t',\rho(G)) \in S$ for every $(t',G) \in \Gb$. Finally, it is immediate that $\csatp{\Gb} \subseteq \csatp{\Hb}$ (our first induction has not increased) and $\Gb \subsetneq \Hb$ (our second induction parameter has decreased). Hence, we may apply Lemma~\ref{lem:pumping} in the case when $P,\Hb$ and $(s,r) \in S$ have been replaced by $P \setminus H,\Gb$ and $(1_N,1_R) \in S$. This yields an $\eta$-pointed \sfp{\Cs}-partition \Ub of $(P \setminus H)^*$ such that $(x,\rho(U)) \in \csats{\Hb} \cap S$ for every $(x,U) \in \Ub$.
  \end{proof}

  We fix the $\eta$-pointed \sfp{\Cs}-partition \Ub of $(P \setminus H)^*$ given by Fact~\ref{fct:covalphind} for the remainder of the proof. We distinguish two subcases. Since $(t,H) \in \Hb$, one may verify from the definitions of \csatp{\Hb} and \csats{\Hb} that $\csats{\Hb} \cdot (t,\rho(H)) \subseteq \csatp{\Hb}$. We consider two subcases depending on whether this inclusion is strict.

  \medskip\noindent{\bf Subcase~1: we have the equality $\csats{\Hb} \cdot (t,\rho(H)) = \csatp{\Hb}$.} We use the following fact, which is proved using our hypotheses and induction on our third parameter (\emph{i.e.}, the size of $(s,r) \cdot \csats{\Hb}$).

  \begin{fact} \label{fct:sfclos:covbufind}
    For every $(x,U) \in \Ub$, there exists an $\eta$-pointed \sfp{\Cs}-partition $\Wb_{x,U}$ of $P^*$ such that $(y,\rho(W)) \in \csats{\Hb}$ and $(sxty,r\rho(UHW)) \in S$ for every $(y,W) \in \Wb_{x,U}$.
  \end{fact}

  \begin{proof}
    We fix $(x,U) \in \Ub$ for the proof. By definition of \Ub in Fact~\ref{fct:covalphind}, $(x,\rho(U)) \in \csats{\Hb}$ and $(x,\rho(U)) \in S$. Hence, since $(s,r),(t,\rho(H))\in S$ by hypothesis and $S$ is closed under multiplication, we get $(sxt,r\rho(UH)) \in S$. Moreover, it is clear that we have the inclusions $(sxt,r\rho(UH))\cdot \csats{\Hb} \subseteq (s,r) \cdot \csats{\Hb} \cdot (t,\rho(H)) \cdot \csats{\Hb} \subseteq (s,r) \cdot \csatp{\Hb}$. Combined with our hypothesis in Subcase~1 (\emph{i.e.}, $\csats{\Hb} \cdot (t,\rho(H)) = \csatp{\Hb}$), this yields $(sxt,r\rho(UH))\cdot \csats{\Hb} \subseteq (s,r) \cdot \csats{\Hb} \cdot (t,\rho(H))$. We may then use~\eqref{eq:sfclos:coverind} (\emph{i.e.}, the inclusion $(s,r) \cdot \csats{\Hb} \cdot (t,\rho(H)) \subsetneq (s,r) \cdot \csats{\Hb}$) to get the \textbf{strict} inclusion $(sxt,r\rho(UH)) \cdot \csats{\Hb} \subseteq (s,r) \cdot \csats{\Hb}$. Consequently, induction on our third parameter (\emph{i.e.}, the size of $(s,r) \cdot \csats{\Hb}$) in Lemma~\ref{lem:pumping} (we consider the case when $(s,r) \in S$ has been replaced by $(sxt,r\rho(UH))  \in S$) yields the desired $\eta$-pointed \sfp{\Cs}-partition $\Wb_{x,U}$ of $P^*$. Note that here, our first two parameters have not increased as they only depend on \Hb, which remains unchanged.
  \end{proof}

  It remains to use Fact~\ref{fct:sfclos:covbufind} to conclude the proof of Subcase~1. We build our \sfp{\Cs}-partition \Kb of $P^*$ as follows,
  \[
    \Kb = \Ub \cup \bigcup_{(x,U) \in \Ub} \{(xty,UHW) \mid (y,W) \in \Wb_{x,U}\}.
  \]
  Let us show that \Kb is an $\eta$-pointed \sfp{\Cs}-partition of $P^*$ satisfying~\eqref{eq:sfclos:covergoal}. First, observe that for every $(t',K) \in \Kb$, we have $K \in \sfp{\Cs}$. This is immediate by hypothesis on \Ub when $(t',K) \in \Ub$. Otherwise, $K = UHW$ for $(x,U) \in \Ub$ and $(y,W) \in \Wb_{x,U}$ and $U,H,W \in \sfp{\Cs}$. Hence, $K \in \sfp{\Cs}$ since \sfp{\Cs} is closed under concatenation.

  That \Kb is an $\eta$-pointed \sfp{\Cs}-partition of $P^*$ is also simple to verify since $P$ is a prefix code, $H \subseteq P$ and $\Ub,\Wb_{x,U}$ are $\eta$-pointed partitions of $(P \setminus H)^*$ and $P^*$ respectively. Each word $w \in P^*$ admits a \emph{unique} decomposition $w = w_1 \cdots w_n$ with $w_1,\dots,w_n \in P$. We partition $P^*$ by looking at the leftmost factor belonging to $H$ (if it exists).

  It remains to prove that~\eqref{eq:sfclos:covergoal} holds. Consider $(t',K) \in \Kb$, we show that $(t',\rho(K)) \in \csats{\Hb}$ and $(st',r\rho(K)) \in S$. If $(t',K) \in \Ub$, this is immediate by definition of \Ub in Fact~\ref{fct:covalphind}. Otherwise, $(t',K) = (wty,UHW)$ with $(x,U) \in \Ub$ and $(y,W) \in \Wb_{x,U}$. By definition of \Ub and $\Wb_{x,U}$, we have $(x,\rho(U)),(y,\rho(W)) \in \csats{\Hb}$. Thus, $(xty,\rho(UHW)) \in \csats{\Hb}$. Moreover, $(sxty,r\rho(UHW)) \in S$ by definition of $\Wb_{x,U}$ in Fact~\ref{fct:sfclos:covbufind}. This concludes the first subcase.

  \medskip\noindent{\bf Subcase~2: we have the strict inclusion $\csats{\Hb} \cdot (t,\rho(H)) \subsetneq \csatp{\Hb}$.} Recall that our objective is to construct an $\eta$-pointed \sfp{\Cs}-partition \Kb of $P^*$ satisfying~\eqref{eq:sfclos:covergoal}. We begin by giving a brief overview of the construction. Consider a word $w \in P^*$. Since $P$ is a prefix code, $w$ admits a unique decomposition as a concatenation of factors in $P$. We may look at the rightmost factor in $H \subseteq P$ to uniquely decompose $w$ in two parts (each of them possibly empty): a prefix in $((P \setminus H)^*H)^*$ and a suffix in $(P \setminus H)^*$. We use induction to construct $\eta$-pointed \sfp{\Cs}-partitions of the sets of possible prefixes and suffixes. Then, we combine them to construct an $\eta$-pointed \sfp{\Cs}-partition of the whole set $P^*$. Actually, we already constructed a suitable $\eta$-pointed \sfp{\Cs}-partition of the possible suffixes in $(P \setminus H)^*$: \Ub (see Fact~\ref{fct:covalphind}). Hence, it remains to partition the prefixes. We do so in the following lemma, which is proved using the hypothesis of Subcase~2 and induction on our first parameter.

  \begin{fct} \label{fct:covindratm}
    There is an $\eta$-pointed \sfp{\Cs}-partition \Vb of $((P \setminus H)^*H)^*$ such that $(z,\rho(V)) \in \csats{\Hb}$ and $(z,\rho(V)) \in S$ for every $V \in \Vb$.
  \end{fct}

  \begin{proof}
    Let $L = (P\setminus H)^*H$. Since $P$ is a prefix code with bounded synchronization delay, so is $L$ by Fact~\ref{fct:sprefnest}. We want to apply induction in Lemma~\ref{lem:pumping} for the case when $P$ has been replaced by $L$. Doing so requires building an appropriate $\eta$-pointed \sfp{\Cs}-partition of $L$ and proving that one of our induction parameters has decreased.

    Let $\Fb = \{(xt,UH) \mid (x,U) \in \Ub\}$. Since \Ub is an $\eta$-pointed \sfp{\Cs}-partition of $(P\setminus H)^*$ and $P$ is a prefix code, one may verify that \Fb is an $\eta$-pointed \sfp{\Cs}-partition of $L=(P\setminus H)^*H$. Finally, given $(y,F) \in \Fb$, we have $(y,F) = (xt,UH)$ for $(x,U) \in \Ub$, which means that $(y,\rho(F)) = (xt,\rho(UH)) \in S$ since $S$ is closed under multiplication. It remains to show that our induction parameters have decreased. Since $\Fb = \{(xt,UH) \mid (x,U) \in \Ub\}$ and $(x,\rho(U)) \in \csats{\Hb}$ for every $(x,U) \in \Ub$ (by definition of \Ub in Fact~\ref{fct:covalphind}), one may verify that $\csatp{\Fb} \subseteq \csats{\Hb} \cdot (t,\rho(H))$. Hence, since $\csats{\Hb} \cdot (t,\rho(H)) \subsetneq \csatp{\Hb}$ by hypothesis in Subcase~2, we have $\csatp{\Fb} \subsetneq \csatp{\Hb}$. Our first induction parameter has decreased. Altogether, it follows that we may apply Lemma~\ref{lem:pumping} in the case when $P,\Hb$ and $(s,r) \in S$ have been replaced by $L,\Fb$ and $(1_N,1_R) \in S$. This yields an $\eta$-pointed \sfp{\Cs}-partition \Vb of $L^* = ((P \setminus H)^*H)^*$ such that for every $(z,V) \in \Vb$, $(z,\rho(V)) \in \csats{\Fb}$ and $(z,\rho(V)) \in S$. Finally, it is clear by definition that $\csats{\Fb} \subseteq \csats{\Hb}$. Hence, the lemma follows.
  \end{proof}

  We are ready to construct the $\eta$-pointed \sfp{\Cs}-partition \Kb of $P^*$ and conclude the main argument. We let $\Kb = \{(zx,VU) \mid (z,V) \in \Vb \text{ and } (x,U) \in \Ub\}$. It is immediate by definition and Fact~\ref{fct:sprefnest} that \Kb is an $\eta$-pointed partition of $P^*$ since $P$ is a prefix code and $\Vb,\Ub$ are $\eta$-pointed partitions of $((P \setminus H)^*H)^*$ and $(P \setminus H)^*$ respectively (see the above discussion). Additionally, it is immediate by definition that \Kb is actually an $\eta$-pointed $\sfp{\Cs}$-partition of $P^*$ (it only contains concatenations of languages in \sfp{\Cs}). It remains to prove that \Kb satisfies~\eqref{eq:sfclos:covergoal}. Let $(t',K) \in \Kb$. By definition, there are $(z,V)\in\Vb$ and $(x,U) \in \Ub$ such that $(t',K) = (zx,VU)$. By definition of \Ub and \Vb, we have $(x,\rho(U))(z,\rho(V)) \in \csats{\Hb}$ and $(x,\rho(U))(z,\rho(V)) \in S$. Moreover, $(s,r) \in S$ by hypothesis. Therefore, since both $\csats{\Hb}$ and $S$ are closed under multiplication, it follows that $(t',\rho(K)) \in \csats{\Hb}$ and $(st',r\rho(K)) \in S$. This completes the proof of Lemma~\ref{lem:pumping}.
\end{proof}

We may now combine Proposition~\ref{prop:sound} and Proposition~\ref{prop:comp} to prove Theorem~\ref{thm:carfinite}. The argument is standard.

\begin{proof}[Proof of Theorem~\ref{thm:carfinite}]
  Let \Cs be a finite \vari and $\rho: 2^{A^*} \to R$ be a \emph{\nice} \mratm. We have to prove that \popti{\sfp{\Cs}}{\etac}{\rho} is the least \sfr-saturated subset of $\canc \times R$ for $\etac$ and $\rho$. We proved in Proposition~\ref{prop:sound} that \popti{\sfp{\Cs}}{\etac}{\rho} is \sfr-saturated. We need to show that it is the least such set. Let $S \subseteq \canc \times R$ which is \sfr-saturated for $\etac$ and $\rho$. Proposition~\ref{prop:comp} yields an $\etac$-pointed \sfp{\Cs}-cover of \Kb of $A^*$ such that $\pprin{\etac,\rho}{\Kb} \subseteq S$. Since Lemma~\ref{lem:pcov} implies that $\popti{\sfp{\Cs}}{\etac}{\rho} \subseteq \pprin{\etac,\rho}{\Kb}$, we obtain $\popti{\sfp{\Cs}}{\etac}{\rho} \subseteq S$, which completes the proof.
\end{proof}


\section{Covering for group input classes}
\label{sec:group}
We now consider separation and covering for the classes \sfp{\Gs} when \Gs is a group \vari. We prove that both problems are decidable when \Gs-separation is decidable. In this case as well, the algorithm is based on the framework introduced in Section~\ref{sec:ratms}: we present a generic effective characterization of optimal \imprints for \sfp{\Gs}. Given an arbitrary \emph{group} \vari \Gs and a \nice \mratm $\rho: 2^{A^*} \to R$, it describes the set $\opti{\sfp{\Gs}}{\rho} \subseteq R$. Moreover, this description is effective when \Gs-separation is decidable. As announced at the end of Section~\ref{sec:ratms}, the characterization actually describes a more general object than the set $\opti{\sfp{\Gs}}{\rho}$. The first part of the section is devoted to defining this object. In the second part, we present and prove the characterization~itself.

\subsection{Optimal \gids}

As explained above, characterizing the optimal \imprints for \sfp{\Gs} (when \Gs is a group \vari) requires working with more general objects than the ones that we actually want to compute:  the sets $\opti{\sfp{\Gs}}{\rho} \subseteq R$ associated to a \nice \mratm $\rho: 2^{A^*} \to R$. In this case as well, we look at sets which are strongly related to the ``nested optimal \imprints'' $\opti{\Gs}{\bratauxsfg} \subseteq 2^R$ associated to the auxiliary \ratms $\bratauxsfg: 2^{A^*} \to 2^R$. Yet, since we are dealing with input classes that are \emph{group \varis}, it will not be necessary to consider the whole set: a single special element of $\opti{\Gs}{\bratauxsfg} \subseteq 2^R$ suffices. Let us first define it. The definition makes sense for \emph{all \ratms}: given an arbitrary \ratm $\tau: 2^{A^*}\to Q$, we identify a special element of $\opti{\Gs}{\tau} \subseteq Q$.

The definition is based on a simple idea: a group \vari \Gs does not contain any finite language. In particular, $\{\veps\} \not\in \Gs$. Hence, given a \ratm $\tau: 2^{A^*} \to Q$, the optimal $\tau$-\imprint for \Gs on the singleton $\{\veps\}$, \emph{i.e.}, $\opti{\Gs}{\{\veps\},\tau}$, is an important object. This leads to the following definition.

\medskip
\noindent
\textbf{\Goptids.} Let \Gs be a lattice and let $\tau: 2^{A^*} \to Q$ be a \ratm. We call \emph{\gid} any language in \Gs containing \veps. An \emph{\goptid for $\tau$} is a \gid $L$ such that for every \gid $L'$, we have $\tau(L) \leq \tau(L')$. In practice, we use this notion when \Gs is a group \vari, but this is not required for the definition. As expected, \goptids for $\tau$ always exist.

\begin{lem} \label{idenex}
  For any lattice \Gs and \ratm $\tau: 2^{A^*} \to Q$, there exists an \goptid for~$\tau$.
\end{lem}

\begin{proof}
  Let $U = \{\tau(L) \mid L \in \Gs \text{ and } \veps \in L\}=\{\tau(L)\mid L\text{ is a \gid}\}$. Clearly, $U$ is nonempty: $\tau(A^*) \in U$ since $A^* \in \Gs$, as \Gs is a lattice. For every $q \in U$, fix an arbitrary \gid $L_q$ such that $q = \tau(L_q)$ and let $L =  \bigcap_{q \in U} L_q$. Since \Gs is a lattice, we have $L \in \Gs$. Moreover, $\veps \in L$ by definition. Since $L \subseteq L_q$ for all $q \in U$, it follows that $\tau(L) \leq q$ for every $q \in U$. By definition of $U$, this implies that $\tau(L) \leq \tau(L')$ for every \gid $L'$. Hence, $L$ is an \goptid for $\tau$.
\end{proof}

We complete the definition with a second notion, which is the counterpart of optimal \imprints in this context. By definition, all \goptids for $\tau$ have the same image under $\tau$. Hence, this image (in $Q$) is a canonical object for \Gs and $\tau$. We write it $\ioptig{\tau} \in Q$. In other words, $\ioptig{\tau} = \tau(K)$ for every \goptid $K$ for $\tau$. We now connect this object with optimal \imprints.

\begin{lem} \label{lem:corresp}
  If\/ \Gs is a lattice and $\tau: 2^{A^*} \to Q$ is a \ratm, $\opti{\Gs}{\{\veps\},\tau} = \dclosq \{\ioptig{\tau}\}$.
\end{lem}

\begin{proof}
  For the left to right inclusion, let $L$ be an \goptid for $\tau$. By definition, we have $L \in \Gs$, $\veps \in L$ and $\ioptig{\tau} = \tau(L)$. Clearly, $\Kb = \{L\}$ is a \Gs-cover of $\{\veps\}$. It follows that $\opti{\Gs}{\{\veps\},\tau} \subseteq \prin{\tau}{\Kb} = \dclosq \{\ioptig{\tau}\}$. Conversely, let \Hb be an optimal \Gs-cover of $\{\veps\}$ for $\tau$. There exists $H \in \Hb$ such that $\veps \in H$. Hence, $\{H\}$ is also an optimal \Gs-cover of $\{\veps\}$  for $\tau$ and it follows that $\dclosq \{\tau(H)\} = \prin{\tau}{\{H\}} = \opti{\Gs}{\{\veps\},\tau}$. Finally, since $H \in \Gs$ and $\veps \in H$, we have $\ioptig{\tau} \leq \tau(H)$, which yields $\dclosq \{\ioptig{\tau}\} \subseteq \dclosq \{\tau(H)\}$. Altogether, we obtain $\dclosq \{\ioptig{\tau}\} \subseteq\opti{\Gs}{\{\veps\},\tau}$, as desired.
\end{proof}

We complete the definition with a key property. When $\tau$ is a \emph{\nice} \ratm, the element $\ioptig{\tau} \in Q$ can be specified in terms of \Gs-separation. 

\begin{lem} \label{lem:sepiopt}
  Let \Gs be a lattice and $\tau: 2^{A^*} \to Q$ be a \nice \ratm. Let $S \subseteq Q$ be the set of all elements $s \in Q$ such that $\{\veps\}$ is not \Gs-separable from $\tau_*\inv(s)$. Then, $\ioptig{\tau} = \sum_{s \in S} s$.
\end{lem}

\begin{proof}
  First, let us prove that $\sum_{s \in S} s \leq \ioptig{\tau}$. This boils down to proving that $s \leq \ioptig{\tau}$ for every $s \in S$. By definition of $S$, we know that $\{\veps\}$ is not \Gs-separable from $\tau_*\inv(s)$. Let $L$ be an \goptid for $\tau$. By definition, we have $L \in \Gs$, $\veps \in L$ and $\ioptig{\tau} = \tau(L)$. Since $\{\veps\}$ is not \Gs-separable from $\tau_*\inv(s)$, $L\in\Gs$ cannot separate $\{\veps\}$ and $\tau_*\inv(s)$, and therefore $L \cap \tau_*\inv(s) \neq \emptyset$. It follows that $s \leq \tau(L)$, \emph{i.e.}, that  $s \leq \ioptig{\tau}$.

  Conversely, we prove that $\ioptig{\tau} \leq \sum_{s \in S} s$. For every $q \in Q \setminus S$, we know that $\{\veps\}$ is \Gs-separable from $\tau_*\inv(q)$. Hence, we get $H_q \in \Gs$ such that $\veps \in H_q$ and $H_q \cap \tau_*\inv(q) = \emptyset$. Let $H = \bigcap_{q \in Q \setminus S} H_q$. Since \Gs is a lattice, we have $H \in \Gs$. Moreover, $\veps \in H$ by definition. Therefore, $\ioptig{\tau} \leq \tau(H)$. Since $\tau$ is \emph{\nice}, there are finitely many words $w_1,\dots,w_k \in H$ such that $\tau(H) = \tau(w_1) + \cdots + \tau(w_k)$. Finally, it follows from the definition of $H$ that $\tau(w_i) \not\in Q \setminus S$ for every $i \leq k$. In other words, we have $\tau(w_i) \in S$ for every $i \leq k$ and we obtain that $\ioptig{\tau} \leq \tau(w_1) + \cdots + \tau(w_k) \leq \sum_{s \in S} s$, as desired.
\end{proof}

In particular when \Gs has decidable separation, Lemma~\ref{lem:sepiopt} yields an algorithm taking a \emph{\nice \mratm} $\tau: 2^{A^*} \to Q$ as input and computing the element $\ioptig{\tau} \in Q$.

\begin{cor} \label{cor:epsep}
  Let \Gs be a lattice \vari with decidable separation. Given a \nice \mratm $\tau: 2^{A^*} \to Q$ as input, the element $\ioptig{\tau} \in Q$ is computable.
\end{cor}

\noindent
{\bf Application to \sfp{\Gs}-covering.} We now explain how these notions are used in the context of \sfp{\Gs}-covering for a group \vari \Gs. Consider a \vari \Ds containing \Gs and which is closed under concatenation. Intuitively, \Ds is meant to be a class that has been built from \Gs using an operator. In practice, we shall use the case when $\Ds = \sfp{\Gs}$. The high levels ideas outlined in Section~\ref{sec:ratms} suggest that given a \nice \mratm $\rho: 2^{A^*} \to R$, characterizing the optimal $\rho$-\imprint for \Ds, \emph{i.e.}, $\opti{\Ds}{\rho}\subseteq R$, involves working with the set $\opti{\Gs}{\bratauxd}\subseteq 2^R$ defined from the auxiliary \ratm $\bratauxd: 2^{A^*}\to 2^R$. Yet, because \Gs is a group \vari and \Ds~is closed under concatenation, partial information on \opti{\Gs}{\bratauxd} suffices. As we prove below, it is enough to consider the \emph{single element} $\ioptig{\bratauxd} \in 2^R$ (which belongs to \opti{\Gs}{\bratauxd} by Lemma~\ref{lem:corresp}).

First, observe that regardless of our hypotheses on \Gs and \Ds, this set is directly connected to the object that we truly want to characterize: the set $\opti{\Ds}{\rho} \subseteq R$. Indeed, by definition, we have $\ioptig{\bratauxd}= \bratauxd(L)=\opti{\Ds}{L,\rho}$ where $L \in \Gs$ is an \goptid for \bratauxd. Hence Fact~\ref{fct:linclus} implies that,
\[
  \ioptig{\bratauxd} = \opti{\Ds}{L,\rho} \subseteq \opti{\Ds}{A^*,\rho} = \opti{\Ds}{\rho}.
\]
It turns out that when \Gs is a group \vari and \Ds is closed under concatenation, \opti{\Ds}{\rho} is \emph{characterized} by its subset \ioptig{\bratauxd}. That is, one may compute the former from the latter.

\begin{prop} \label{prop:idencomp}
  Let \Gs be a group \vari and \Ds be a \vari closed under concatenation such that $\Gs \subseteq \Ds$. Let $\rho: 2^{A^*} \to R$ be a \mratm. Then, \opti{\Ds}{\rho} is the least subset $S \subseteq R$ containing \ioptig{\bratauxd} and satisfying the following conditions:
  \begin{enumerate}
  \item {\bf Trivial elements:} For every $w \in A^*$, we have $\rho(w) \in S$.
  \item {\bf Closure under downset:} We have $\dclosr S = S$.
  \item {\bf Closure under multiplication:} For every $q,r \in S$, we have $qr \in S$.
  \end{enumerate}
\end{prop}

Clearly, Proposition~\ref{prop:idencomp} yields a least fixpoint procedure for computing the set $\opti{\Ds}{\rho}$ from the set $\ioptig{\bratauxd}$. This is exactly how we handle optimal \imprints for \sfp{\Gs} below. Rather than directly characterizing the sets $\opti{\sfp{\Gs}}{\rho} \subseteq R$ for all \nice \mratms $\rho: 2^{A^*} \to R$, we instead characterize the sets $\ioptig{\bratauxsfg} \subseteq R$, which \emph{carry more information} by Proposition~\ref{prop:idencomp}. Before we present the characterization, let us prove Proposition~\ref{prop:idencomp}.

\begin{proof}[Proof of Proposition~\ref{prop:idencomp}]
 Let $S \subseteq R$ be the least subset of $R$ containing \ioptig{\bratauxd} and satisfying the three conditions in Proposition~\ref{prop:idencomp}. We first prove that $S \subseteq \opti{\Ds}{\rho}$. This is immediate. Indeed, as seen above, we have $\ioptig{\bratauxd} \subseteq \opti{\Ds}{\rho}$. Also, $\opti{\Ds}{\rho}$ contains all trivial elements: clearly, $\rho(w) \in \opti{\Ds}{A^*,\rho} = \opti{\Ds}{\rho}$ for every $w \in A^*$. Moreover, $\opti{\Ds}{\rho}$ is closed under downset since it is an \imprint. Finally, since \Ds is a \vari, Lemma~\ref{lem:mratmult} implies that $\opti{\Ds}{\rho}$ is closed under multiplication. Altogether, we get $S \subseteq \opti{\Ds}{\rho}$ as desired.

  We turn to the inclusion $\opti{\Ds}{\rho} \subseteq S$. We build a \Ds-cover \Kb of $A^*$ such that $\prin{\rho}{\Kb} \subseteq S$. Since $\opti{\Ds}{\rho}\subseteq \prin{\rho}{\Kb}$ by definition, this yields the desired result. Let $L\subseteq A^*$ be an \goptid for \bratauxd. By definition, $L \in \Gs$, $\veps\in L$ and $\ioptig{\bratauxd} = \bratauxd(L) = \opti{\Ds}{L,\rho}$.  We now use the following lemma to build a special cover of $A^*$.

  \begin{lem}\label{lem:ucov}
  There exists a cover \Ub of $A^*$ such that each $U \in \Ub$ satisfies $U = La_1L \cdots a_\ell L$ where $\ell \in \nat$ and $a_1, \dots,a_\ell \in A$ (if $\ell = 0$, then $U = L$).
  \end{lem}

  \begin{proof}
    With every word $w=a_1 \cdots a_\ell \in A^*$, we associate the language $U_w =  La_1L \cdots a_\ell L$ (in particular, $U_\veps = L$). Since $L \in \Gs$, Proposition~\ref{prop:genocm} yields a \Gs-morphism $\eta: A^* \to G$ recognizing~$L$. By hypothesis on \Gs, Lemma~\ref{lem:gmorph} implies that $G$ is a group. We let $k = |G|$ and $\Ub=\{U_w\mid w\in A^* \text{ and } |w| \leq k\}$. All languages in \Ub have the desired form. Hence, it suffices to verify that \Ub is a cover of $A^*$. We fix $v \in A^*$ and exhibit $U \in \Ub$ such that $v \in U$. It can be verified using a pumping argument that there exist $\ell \leq k$, $a_1,\dots,a_\ell \in A$ and $v_0,\dots,v_\ell \in A^*$ such that $v=v_0a_1v_1\cdots a_\ell v_\ell$ and $\eta(v_0a_1v_1 \cdots a_iv_i) =  \eta(v_0a_1v_1 \cdots a_i)$ for every $i \leq \ell$. Since $G$ is a group, this yields $\eta(v_i) = 1_G$ for all $i \leq \ell$. Hence, since $\eta$ recognizes $L$ and $\veps \in L$, we get $v_i \in L$ for every $i \leq \ell$. Thus $v \in U_w$ for $w = a_1 \cdots a_\ell$. Finally, we have $U_w \in \Ub$ since $|w| = \ell \leq k$. This completes the proof.
  \end{proof}

  Let \Hb be an optimal \Ds-cover of $L$ for $\rho$. By definition, we have $\prin{\rho}{\Hb} = \opti{\Ds}{L,\rho}=\ioptig{\bratauxd}$. Consider the cover~ \Ub of $A^*$ given by Lemma~\ref{lem:ucov}. Let $U \in \Ub$. By definition $U=La_1L\cdots a_n L$ for $a_1,\dots,a_n\in A$. Let $\Kb_U = \{H_0a_1H_1 \cdots a_nH_n \mid H_0,\dots,H_n \in \Hb\}$. Since \Hb is a \Ds-cover of $L$ and \Ds is closed under concatenation, $\Kb_U$ is a \Ds-cover of $U$. Finally, let $\Kb = \bigcup_{U\in \Ub} \Kb_U$. Since \Ub is a cover of $A^*$ and $\Kb_U$ is a \Ds-cover of $U$ for every $U \in \Ub$, it is immediate that \Kb is a \Ds-cover of $A^*$.

  It remains to prove that $\prin{\rho}{\Kb}\subseteq S$. Since $S$ is closed under downset, it suffices to prove that $\rho(K) \in S$ for every $K \in \Kb$. By definition, there exist $a_1,\dots,a_n \in A$ and $H_0,\dots,H_n \in \Hb$ such that $K= H_0a_1H_1 \cdots a_nH_n$. Hence, 
 it suffices to prove that $\rho(H_0a_1H_1 \cdots a_nH_n)\in S$. By definition of \Hb, we have $\prin{\rho}{\Hb} = \ioptig{\bratauxd}$. Therefore, $\rho(H_i) \in \ioptig{\bratauxd}$ for every $i \leq n$. Since $S$ contains \ioptig{\bratauxd}, this yields $\rho(H_i) \in S$ for every $i \leq n$. Moreover, since $S$ contains the trivial elements, we have $\rho(a_i) \in S$ for every $i \leq n$. Finally, since $S$ is closed under multiplication, we obtain $\rho(H_0a_1H_1 \cdots a_nH_n) \in S$, which completes the proof.
\end{proof}

\subsection{Characterization}

\newcommand{\sfratp}[1]{\ensuremath{\mu_{\rho,#1}}\xspace}
\newcommand{\sfrats}{\sfratp{S}}

Let us first present the characterization. Consider an arbitrary group \vari \Gs and an arbitrary \mratm $\rho: 2^{A^*} \to R$.  We define special subsets of $R$, which we call \emph{\sfr-complete} for \Gs and $\rho$. The definition is based on an auxiliary \nice \mratm $\sfrats: 2^{A^*} \to 2^{R}$, which we associate to every subset $S \subseteq R$. Its \rata is $(2^{R}, \cup,\cdot)$ (the multiplication on $2^R$ is obtained by lifting the multiplication on $R$ to subsets). Since we are defining a \emph{\nice} \mratm, it suffices to specify the image of each letter $a \in A$.  We let,
\[
  \sfrats(a) = \{s\rho(a)s' \mid s,s' \in S\} \in 2^R.
\]
Observe that by definition, we have $\ioptig{\sfrats} \subseteq R$ for each set $S \subseteq R$. Now, consider an arbitrary subset $S\subseteq R$. We say that $S$ is \emph{\sfr-complete for \Gs and $\rho$} when it satisfies the following~properties:
\begin{enumerate}
  \item {\bf Closure under downset:} $\dclosr S = S$.
  \item {\bf Closure under multiplication:} For every $q,r \in S$, we have $qr \in S$.
  \item {\bf \Gs-operation:} We have $\ioptig{\sfrats} \subseteq S$.
  \item {\bf \sfr-closure.} For every $r \in S$, we have $r^\omega + r^{\omega+1} \in S$.
\end{enumerate}

\begin{rem} \label{rem:compne}
  The definition does not explicitly require that an \sfr-complete subset $S \subseteq R$ for \Gs and $\rho$ contains some trivial elements. Yet, this is a consequence of\/ \Gs-operation. Indeed, we have $\sfrats(\veps) = \{1_R\}$ (this is the multiplicative neutral element of $2^{R}$). Thus, $1_R \in \ioptig{\sfrats} \subseteq S$ by definition.
\end{rem}

With this definition in hand, we may state the main theorem of this section. When $\rho$ is \nice, \ioptig{\bratauxsfg} is the least $\sfr$-complete subset of $R$.

\begin{thm} \label{thm:cargroup}
  Let \Gs be a group \vari and $\rho: 2^{A^*} \to R$ a \emph{\nice} \mratm. Then, \ioptig{\bratauxsfg} is the least $\sfr$-complete subset of $R$ for \Gs and $\rho$.
\end{thm}

When \Gs-separation is decidable, Theorem~\ref{thm:cargroup} yields a least fixpoint procedure for computing \ioptig{\bratauxsfg} from a \nice \mratm $\rho: 2^{A^*} \to R$. The computation starts from the empty set and saturates it with the operations in the definition of $\sfr$-complete subsets. It is clear that we may implement downset, multiplication and \sfr-closure. Moreover, we may compute \ioptig{\sfrats} from a set $S \subseteq R$ by Corollary~\ref{cor:epsep} since \Gs-separation is decidable. Eventually, the computation reaches a fixpoint and it is straightforward to verify that this set is the least $\sfr$-complete subset of $R$, \emph{i.e.}, \ioptig{\bratauxsfg} by Theorem~\ref{thm:cargroup}. One may then use a second least fixpoint procedure provided by Proposition~\ref{prop:idencomp} to compute the set $\opti{\sfp{\Gs}}{\rho} \subseteq R$ from \ioptig{\bratauxsfg}.

Consequently, it follows from Proposition~\ref{prop:breduc} that \sfp{\Gs}-covering is decidable if \Gs-separation is decidable.  As before, the result can be lifted to separation for \sfp{\Gs} using Lemma~\ref{lem:covsep}. Altogether, we obtain  the following corollary.

\begin{cor} \label{cor:cargroup}
  Let \Gs be a group \vari with  decidable separation. Then, separation and covering are decidable for \sfp{\Gs}.
\end{cor}

Corollary~\ref{cor:cargroup} has three important applications: the group \varis \md (the modulo languages), \abg (the alphabet modulo testable languages) and \grp (all group languages). Indeed, since it is known that the three of them have decidable separation (see \emph{e.g.}, \cite{pzgroups}), we obtain the decidability of covering for the classes \sfp{\md}, \sfp{\abg} and \sfp{\grp}.

\medskip

We now concentrate on the proof of Theorem~\ref{thm:cargroup}. In this case as well, we present two independent statements corresponding respectively to soundness and completeness. Let us start with the former. The argument is based on Proposition~\ref{prop:sound}, which addresses soundness for the characterization of optimal \imprints for \sfp{\Cs} when \Cs is a \emph{finite} \vari.

\begin{prop}[Soundness] \label{prop:soundg}
  Let \Gs be a group \vari of group languages and $\rho: 2^{A^*}\to R$ be a \mratm. Then, $\ioptig{\bratauxsfg}\subseteq R$ is \sfr-complete for \Gs and~$\rho$.
\end{prop}

\begin{proof}
  For the sake of avoiding clutter, let $S = \ioptig{\bratauxsfg}$. We start with a preliminary~result.

  \begin{lem} \label{lem:soundg}
  There exists a finite group \vari \Hs that satisfies $\Hs \subseteq \Gs$ and such that $S = \{r \in R \mid (1_G,r) \in \popti{\sfp{\Hs}}{\eta_\Hs}{\rho}\}$ where $\eta_\Hs: A^* \to G$ is the canonical \Hs-morphism.
  \end{lem}

  Before we prove Lemma~\ref{lem:soundg}, let us use it to complete the main argument. Note that Proposition~\ref{prop:sound} implies that $\popti{\sfp{\Hs}}{\eta_\Hs}{\rho}\subseteq G\times R$ is \sfr-saturated for $\eta_\Hs$ and $\rho$. We shall use this property multiple times. We show that $S$ is \sfr-complete for \Gs and $\rho$. That $S$ is closed under downset, multiplication and \sfr-closure is immediate from Lemma~\ref{lem:soundg} (more precisely, it follows from the equality $S = \{r \in R \mid (1_G,r) \in \popti{\sfp{\Hs}}{\eta_\Hs}{\rho}\}$ and the fact that \popti{\sfp{\Hs}}{\eta_\Hs}{\rho} satisfies the three corresponding properties in the definition of \sfr-saturated sets). Hence, we concentrate on \Gs-operation. Given $s \in \ioptig{\sfrats}$, we prove that $s \in S$. By hypothesis in Lemma~\ref{lem:soundg}, we have $\eta_\Hs\inv(1_G) \in \Hs \subseteq \Gs$. Moreover, it is clear that $\veps \in \eta_\Hs\inv(1_G)$. Therefore, we obtain $\ioptig{\sfrats} \subseteq \sfrats(\eta_\Hs\inv(1_G))$. This yields $w \in \eta_\Hs\inv(1_G)$ such that $s \in \sfrats(w)$. There are now two cases. First, if $w = \veps$, then $\sfrats(\veps) = \{1_R\}$ which yields $s = 1_R$. Clearly, we have $(1_G,1_R) \in \popti{\sfp{\Hs}}{\eta_\Hs}{\rho}$ (this is a trivial element), hence we obtain $s = 1_R \in S$ by Lemma~\ref{lem:soundg}. Assume now that $w \in A^+$ and let $a_1,\dots,a_n \in A$ be the letters such that $w = a_1 \cdots a_n$. By definition of \sfrats, the fact that  $s \in \sfrats(w)$ yields $s_1,\dots,s_n,t_1,\dots,t_n \in S$ such that $s = s_1\rho(a_1)t_1 \cdots s_n\rho(a_n)t_n$. Clearly, we have $(\eta_\Hs(a_i),\rho(a_i)) \in \popti{\sfp{\Hs}}{\eta_\Hs}{\rho}$ for every $i \leq n$ (this is a trivial element). Also, Lemma~\ref{lem:soundg} yields $(1_G,s_i),(1_G,t_i) \in  \popti{\sfp{\Hs}}{\eta_\Hs}{\rho}$ for every $i \leq n$. Altogether, since \popti{\sfp{\Hs}}{\eta_\Hs}{\rho} is closed under multiplication, we get $(\eta_\Hs(w),s) = (\eta_\Hs(a_1 \cdots a_n),s) \in \popti{\sfp{\Hs}}{\eta_\Hs}{\rho}$. Finally, since $w \in \eta_\Hs\inv(1_G)$, we get $(1_G,s) \in \popti{\sfp{\Hs}}{\eta_\Hs}{\rho}$ and  Lemma~\ref{lem:soundg} yields $s \in S$, as desired.

  \smallskip

  It remains to prove Lemma~\ref{lem:soundg}. We first define \Hs. Let $L \in \Gs$ be an \goptid for \bratauxsfg. By definition, $\veps \in L$ and $S = \ioptig{\bratauxsfg}=\bratauxsfg(L)$. Let \Kb be an optimal \sfp{\Gs}-cover of $L$ for $\rho$, so that $\prin{\rho}{\Kb}=\opti{\sfp{\Gs}}{L,\rho} = \bratauxsfg(L) = S$. By definition of \sfp\Gs, there exists a \emph{finite} set of languages $\Hb \subseteq \Gs$ such that every language $K \in \Kb\subseteq \sfp{\Gs}$ is built from the languages in \Hb and the singletons $\{a\}$ for $a \in A$ using only concatenation and Boolean combinations. It follows from Proposition~\ref{prop:vari} that \sfp{\Gs} is a \vari. Therefore, Proposition~\ref{prop:genocm} yields a \Gs-morphism $\eta: A^*\to G$ all languages in $\{L\} \cup \Hb$. Moreover, since \Gs is a group \vari, Lemma~\ref{lem:gmorph} implies that $G$ is a group. We write \Hs for the class of all languages recognized by $\eta$. That is, $\Hs = \{\eta\inv(F) \mid F \subseteq G\}$. One may verify that \Hs is a finite group \vari such that $\Hs \subseteq \Gs$. Moreover, $\eta$ is an \Hs-morphism that recognizes every language in \Hs. Hence, $\eta$ is the canonical \Hs-morphism by Lemma~\ref{lem:cmdiv}.

  It remains to prove that $S = \{r \in R \mid (1_G,r) \in \popti{\sfp{\Hs}}{\eta}{\rho}\}$. We start with the left to right inclusion. Let $r \in S$. By definition, $S = \ioptig{\bratauxsfg}$. Hence, since $\eta\inv(1_G) \in \Gs$ and $\veps \in \eta\inv(1_G)$, we have $r \in \bratauxsfg(\eta\inv(1_G)) = \opti{\sfp{\Gs}}{\eta\inv(1_G),\rho}$. Finally, since $\Hs \subseteq \Gs$, this implies by Fact~\ref{fct:linclus} that $r \in \opti{\sfp{\Hs}}{\eta\inv(1_G),\rho}$, which exactly says that $(1_G,r) \in \popti{\sfp{\Hs}}{\eta}{\rho}$, as desired. We turn to the converse inclusion. Let $r \in R$ such that $(1_G,r) \in \popti{\sfp{\Hs}}{\eta}{\rho}$. We show that $r \in S$. By definition $r \in \opti{\sfp{\Hs}}{\eta\inv(1_G),\rho}$. Moreover, since $\eta$ recognizes $L$ and we have $\veps \in L$, we have $\eta\inv(1_G) \subseteq L$. Therefore, Fact~\ref{fct:linclus} yields $r \in \opti{\sfp{\Hs}}{L,\rho}$. Moreover, it is immediate by definition of \Hs that every $K \in \Kb$ belongs to \sfp{\Hs}. Hence, \Kb is an \sfp{\Hs}-cover of $L$ and we get $r \in \prin{\rho}{\Kb} = S$. This completes the proof.
\end{proof}

We turn to completeness. The argument is based on Proposition~\ref{prop:comp},  the direction addressing completeness for the characterization of optimal \imprints for \sfp{\Cs} when \Cs is a \emph{finite} \vari.

\begin{prop}[Completeness]\label{prop:compg}
  Let \Gs be a \vari of group languages and $\rho: 2^{A^*} \to R$ be a \nice \mratm. Let $S \subseteq R$ be $\sfr$-complete for \Gs and $\rho$. Then, there exists a language $L \in \Gs$ such that $\veps \in L$, as well as an \sfp{\Gs}-cover $\Kb$ of\/ $L$ such that $\prin{\rho}{\Kb} \subseteq S$.
\end{prop}

\begin{proof}
  We start with a preliminary lemma. It is based on the hypothesis that $S$ is \sfr-complete for \Gs and $\rho$. We need it in order to apply Proposition~\ref{prop:comp}.

  \begin{lem} \label{lem:compg}
    There exist a \Gs-morphism $\eta: A^* \to G$ and a set $S' \subseteq G \times R$ which is \sfr-saturated for $\eta$ and $\rho$, and such that $S = \{r \mid (1_G,r) \in S'\}$.
  \end{lem}

  We first apply Lemma~\ref{lem:compg} to complete the proof of Proposition~\ref{prop:compg}. We define $L = \eta\inv(1_G)$. It is clear that $\veps \in L$ and $L\in \Gs$ since $\eta$ is a \Gs-morphism. Moreover, since $S'$ is \sfr-saturated for $\eta$ and $\rho$, Proposition~\ref{prop:comp} yields an $\eta$-pointed \sfp{\Gs}-cover \Hb of $A^*$ such that $\pprin{\eta,\rho}{\Hb} \subseteq S'$. We define $\Kb = \{K \mid (1_G,K) \in \Hb\}$.  By definition, \Kb is an \sfp{\Gs}-cover of $L = \eta\inv(1_G)$. Moreover, since $S = \{r \mid (1_G,r) \in S'\}$ and $\pprin{\eta,\rho}{\Hb} \subseteq S'$, it is immediate from the definition that $\prin{\rho}{\Kb} \subseteq S$, which completes the proof.

  \smallskip

  It remains to prove Lemma~\ref{lem:compg}. Let us first define $\eta$.  We let $L$ be an \goptid for the \nice \mratm $\sfrats: 2^{A^*} \to 2^R$. That is, we have $L \in \Gs$, $\veps \in L$ and $\sfrats(L) = \ioptig{\sfrats}$. Proposition~\ref{prop:genocm} yields a \Gs-morphism $\eta: A^* \to G$ recognizing $L$. By hypothesis on \Gs, Lemma~\ref{lem:gmorph} implies that $G$ is a group. Also since $\veps \in L$ and $L$ is recognized by $\eta$, it is immediate that $\eta\inv(1_G)\subseteq L$. Finally, we define $S' \subseteq G \times R$ as the following set:
  \[
    S' = \{(1_{G},s) \mid s \in S\} \cup \{(t,r) \mid r \in \dclosr \sfrats(\eta\inv(t))\}.
  \]
  Let us first prove that $S = \{r \mid (1_G,r) \in S'\}$. The left to right inclusion is immediate by definition: if $r \in S$, then $(1_G,r) \in S'$. We turn to the converse inclusion. Let $r \in R$ such that $(1_G,r) \in S'$. By definition of $S'$, either $r \in S$ or $r \in \dclosr \sfrats(\eta\inv(1_{G}))$. In the former case, we are done. Hence, we assume that $r \in \dclosr \sfrats(\eta\inv(1_{G}))$. By definition, $\eta\inv(1_{G}) \subseteq L$ ($L$ is recognized by $\eta$ and $\veps \in L$). Hence, since $\sfrats(L) = \ioptig{\sfrats}$ we get $r \in \dclosr \ioptig{\sfrats}$. Finally, since $S$ is \sfr-complete for \Gs and $\rho$, we have $\dclosr \ioptig{\sfrats} \subseteq S$ and we get $r \in S$, as desired.

  \smallskip
  It remains to prove that $S'$ is \sfr-saturated for $\eta$ and $\rho$. This involves four properties. Let us start with trivial elements. Consider $w \in A^*$, we show that $(\eta(w),\rho(w)) \in S'$.  By definition of $S'$, it suffices to prove that $\rho(w) \in \dclosr \sfrats(\eta\inv(\eta(w)))$. Clearly, $w \in \eta\inv(\eta(w))$. Hence, it now suffices to prove that $\rho(w) \in \sfrats(w)$. If $w = \veps$, then $\rho(w) = 1_R$ and $\sfrats(\veps)= \{1_R\}$ (this is the neutral element of $2^R$). Thus, it is immediate that $\rho(w) \in \sfrats(w)$. Assume now that $w \in A^+$. We have $a_1,\dots,a_n \in A$ such that $w=a_1 \cdots a_n$. Since $S$ is \sfr-complete, we have $1_R \in S$ (see Remark~\ref{rem:compne}). Hence, we get $\rho(a_i) \in S \cdot \{\rho(a_i)\} \cdot S = \sfrats(a_i)$ for every $i \leq n$. It then follows that $\rho(w) \in \sfrats(w)$ which concludes this case. We turn to downset. Consider $(t,r) \in S'$ and $q \leq r$. We show that $(t,q) \in S'$. By definition, there are two possible cases. First, it may happen that $t = 1_{G}$ and $r \in S$. In that case, $q \in S$ since $\dclosr S = S$ ($S$ is complete) and we get that $(1_{G},q)\in S'$ by definition. Otherwise, $r \in \dclosr \sfrats(\eta\inv(t))$ which yields $q \in \dclosr \sfrats(\eta\inv(t))$ as well and we get $(t,q) \in S'$, concluding the proof for downset.

  \smallskip
  We turn to closure under multiplication. Let $(s,q),(t,r) \in S'$. We have to show that $(st,qr) \in S'$. There are several cases. Assume first that $s = t = 1_{G}$. In that case, we proved above that $q,r  \in S$. Since $S$ is complete, this implies that $qr \in S$ and we obtain that $(st,qr) = (1_{G},qr) \in S'$. We now assume that either $s$ or $t$ is distinct from $1_{G}$ for the remainder of this case. If both $s$ and $t$ are distinct from $1_{G}$, we have $q \in \dclosr\sfrats(\eta\inv(s))$ and $r \in \dclosr \sfrats(\eta\inv(t))$ by definition of $S'$. It follows that $qr \in \dclosr \sfrats(\eta\inv(s)\eta\inv(t))$. Since $\eta\inv(s)\eta\inv(t) \subseteq \eta\inv(st)$, this yields $qr \in  \dclosr \sfrats(\eta\inv(st))$ and we get $(st,qr) \in S'$. Finally, we handle the case when $s = 1_{G}$ and $t \neq 1_{G}$ (the symmetrical case is left to the reader). Since $S = S'(1_G)$, the hypothesis that $t = 1_{G}$ yields $q \in S$. Moreover, we have $r \in \dclosr \sfrats(\eta\inv(t))$ since $t \neq 1_{G}$. Note that $t \neq 1_{G}$ also implies that $\eta\inv(t) \subseteq A^+$. Thus, since $SS \subseteq S$ ($S$ is complete), one may verify from the definition of \sfrats that $S \cdot \sfrats(\eta\inv(t)) \subseteq \sfrats(\eta\inv(t))$. Thus, since $q \in S$ and $r \in \dclosr \sfrats(\eta\inv(t))$, we obtain that $qr \in \dclosr \sfrats(\eta\inv(t))$. Finally, $1_{G}t = t$. Thus, $qr \in \dclosr \sfrats(\eta\inv(1_{G}t))$ which yields $(1_{G}t,qr) \in S'$, concluding the proof.

  \smallskip
  It remains to prove that $S'$ satisfies \sfr-closure. Since $G$ is a group, $1_{G}$ is the only idempotent of~$G$. Hence, we have to show that for $r \in R$ such that $(1_{G},r) \in S'$, we have $(1_{G},r^{\omega} + r^{\omega+1}) \in S'$. Since $S = \{r \mid (1_G,r) \in S'\}$, we know that $r \in S$. Hence, since $S$ is \sfr-complete, we get $r^{\omega} + r^{\omega+1} \in S$ which implies $(1_{G},r^{\omega} + r^{\omega+1}) \in S'$ since $S = \{r \mid (1_G,r) \in S'\}$. This completes the proof.
\end{proof}

We may now combine Proposition~\ref{prop:soundg} and Proposition~\ref{prop:compg} to prove Theorem~\ref{thm:cargroup}.

\begin{proof}[Proof of Theorem~\ref{thm:cargroup}]
  Let \Gs be a group \vari and $\rho: 2^{A^*} \to R$ be a \nice \mratm. We prove that \ioptig{\bratauxsfg} is the least $\sfr$-complete subset of $R$ for \Gs and $\rho$. Proposition~\ref{prop:soundg} implies that \ioptig{\bratauxsfg} is $\sfr$-complete. We have to show that it is the least such~set.

  Let $S \subseteq R$ be $\sfr$-complete for \Gs and $\rho$. We prove that $\ioptig{\bratauxsfg} \subseteq S$. Proposition~\ref{prop:compg} yields $L \in \Gs$ such that $\veps \in L$ and an \sfp{\Gs}-cover $\Kb$ of $L$ such that $\prin{\rho}{\Kb} \subseteq S$. By definition of \goptids, we have $\ioptig{\bratauxsfg} \subseteq \bratauxsfg(L)$. Hence, it suffices to show that $\bratauxsfg(L) \subseteq S$. By definition, we have $\bratauxsfg(L) = \opti{\sfp{\Gs}}{L,\rho} \subseteq \prin{\rho}{\Kb} \subseteq S$. This concludes the proof.
\end{proof}


\section{Conclusion}
\label{sec:conc}
We investigated the star-free closure operator $\Cs \mapsto \sfp{\Cs}$ applied \varis or regular languages. First, we presented several equivalent ways of defining \sfp\Cs, including a generic algebraic characterization that generalizes earlier results~\cite{schutzsf,STRAUBING1979319,Pinambigu}. It implies that \sfp{\Cs}-\emph{membership} is decidable when \Cs-\emph{separation} is decidable. A key proof ingredient for all these results is an alternative definition of star-free closure: the operator $\Cs \mapsto \bsdp{\Cs}$, which we prove to coincide with star-free closure. This correspondence generalizes the work of Schützenberger~\cite{schutzbd} who introduced a single class \bsd ({\em i.e.}, \bsdp{\stzer}) corresponding to the star-free languages ({\em i.e.},~\sfp{\stzer}). Moreover, we presented two generic logical characterizations of star-free closure. The first one is based on first-order logic and generalizes the work of McNaughton and Papert on the star-free languages~\cite{mnpfosf}: we have $\sfp{\Cs} = \fo(\infsigc)$ for every \vari \Cs. The second one is based on linear temporal logic and generalizes Kamp's theorem~\cite{kltl}: we have $\sfp{\Cs} = \ltlc{\Cs} = \ltlpc{\Cs}$ for every \vari \Cs. Finally, we gave two generic characterizations of optimal \imprints for \sfp{\Cs}, for two particular kinds of \vari \Cs: the \emph{finite} \varis and the \emph{group} \varis. They imply that \sfp{\Cs} has decidable separation and covering when \Cs is a finite \vari and that \sfp{\Gs} has decidable separation and covering when \Gs is a group \vari with decidable separation.


\printbibliography


\end{document}